\documentclass{article}

\usepackage{arxiv}

\usepackage[utf8]{inputenc}
\usepackage[T1]{fontenc}
\usepackage{hyperref}
\usepackage{url}
\usepackage{booktabs}
\usepackage{amsfonts}
\usepackage{amsmath}
\usepackage{amssymb}
\usepackage{amsthm}
\usepackage{nicefrac}
\usepackage{microtype}
\usepackage{cleveref}
\usepackage{graphicx}
\usepackage[numbers]{natbib}
\usepackage{doi}

\usepackage{algorithm}
\usepackage{algpseudocode}
\usepackage{xspace}
\usepackage{float}
\usepackage{placeins}
\usepackage{array}
\usepackage{multirow}
\usepackage{pifont}

\newtheorem{definition}{Definition}
\newtheorem{proposition}{Proposition}

\providecommand{\tsc}[1]{\expandafter\providecommand\csname #1\endcsname{\textsc{\lowercase{#1}}\xspace}}
\tsc{WGM}
\tsc{QE}

\title{SCARA: A Semantics-Constrained Autonomous Remediation Agent
for Opaque Industrial Software Vulnerabilities}

\usepackage{authblk}

\author[1,3]{Bowei Ning\thanks{\texttt{2020183@stu.syuct.edu.cn}}}
\author[2,3]{Xuejun Zong\thanks{Corresponding author. \texttt{xuejun\_zong@syuct.edu.cn}}}
\author[2,3]{Lian Lian\thanks{Corresponding author. \texttt{lianlian@syuct.edu.cn}}}
\author[2,3]{Kan He\thanks{\texttt{hekan1978@syuct.edu.cn}}}
\author[2,3]{Guogang Wang}
\author[2,3]{Yifei Sun\thanks{\texttt{sunyifei@syuct.edu.cn}}}
\author[1,3]{Jinyang Liu\thanks{\texttt{z2022310@stu.syuct.edu.cn}}}

\affil[1]{Shenyang University of Technology, Shenyang 110870, Liaoning, China}
\affil[2]{Shenyang University of Chemical Technology, Shenyang 110142, Liaoning, China}
\affil[3]{Key Laboratory of Information Security for Petrochemical Industry in Liaoning Province, Shenyang 110142, China}

\renewcommand{\shorttitle}{SCARA}

\hypersetup{
pdftitle={SCARA: A Semantics-Constrained Autonomous Remediation Agent for Opaque Industrial Software Vulnerabilities},
pdfauthor={Bowei Ning, Xuejun Zong, Lian Lian, Kan He, Guogang Wang, Yifei Sun, Jinyang Liu},
pdfkeywords={Opaque industrial software, Binary vulnerability analysis, Automated vulnerability remediation, Symbolic execution, Industrial control systems},
}

\begin{document}
\maketitle

% Float placement parameters preserved from the source manuscript.
\renewcommand{\topfraction}{0.95}
\renewcommand{\bottomfraction}{0.85}
\renewcommand{\textfraction}{0.05}
\renewcommand{\floatpagefraction}{0.85}
\renewcommand{\dbltopfraction}{0.95}
\renewcommand{\dblfloatpagefraction}{0.85}
\setcounter{topnumber}{4}
\setcounter{bottomnumber}{3}
\setcounter{totalnumber}{6}
\setcounter{dbltopnumber}{4}

\begin{abstract}
Critical-infrastructure operators are increasingly expected, under policy and
standardization regimes such as Executive Order~14028 and IEC~62443, to assess
and remediate vulnerabilities in deployed software. However, a substantial
fraction of deployed industrial software is delivered as \emph{opaque industrial software} (OIS): stripped firmware, proprietary protocol handlers, and compiled
control logic for which source code, debug symbols, build environments, and
hardware interfaces are unavailable to defenders. Although binary analysis can
flag vulnerability candidates in such artifacts, existing automated vulnerability
repair systems overwhelmingly assume source code, a compilable artifact,
sanitizer feedback, or an instrumentable build. These assumptions do not hold
for opaque industrial binaries, leaving a gap between binary-level vulnerability
discovery and validated remediation.

This paper presents SCARA, a Semantics-Constrained Autonomous Remediation Agent
for OIS. SCARA operates under a source-unavailable defender model in which
upstream binary detectors or analysts provide vulnerability candidates. Its
four-stage pipeline combines operational-state-aware verification (OSVA),
remediation synthesis (RSA), and correctness validation (CVA). OSVA prunes
operationally infeasible candidates under a nine-component industrial state
model. RSA emits the strongest remedy permitted by artifact availability across
three tiers: protocol mitigation, binary hardening, and SSCKG-constrained source
patches. CVA provides conditional correctness evidence through behavioral
coverage preservation, independent replay, and a typed rejection-feedback loop
from CVA to RSA. The key technical novelty is the joint treatment of the
operational-state envelope as a first-class correctness condition and a tiered
remedy selector that degrades gracefully as artifact opacity increases.

On OIS-RemedBench, a 15-case benchmark spanning firmware, protocol
handlers, and ICS/PLC artifacts, SCARA achieves observed 100\% precision with
no false positives in the evaluated partitions, refutes 20.0\% (3/15) of cases
as operationally infeasible, and reaches 88.9\% (8/9) remediation success after
targeted reruns. The UNKNOWN rate drops from 13.3\% to 6.7\%, Tier-2
overblocking from 80.0\% to 25.0\%, and OIS-ICS analyst burden from 17.04 to
2.42 hours per case, corresponding to a $7.0\times$ throughput gain. To our
knowledge, SCARA is the first end-to-end framework that connects binary
vulnerability candidates to conditionally validated remediation for opaque
industrial software.
\end{abstract}

\keywords{Opaque industrial software \and Binary vulnerability analysis \and Automated vulnerability remediation \and Symbolic execution \and Industrial control systems}

%% Highlights (Elsevier-specific) preserved here as an itemised list.
\paragraph{Highlights.}
\begin{itemize}
\item SCARA: first end-to-end framework for validated OIS vulnerability remediation.
\item Tiered remediation synthesis provides non-zero remediation without source code.
\item OIS-RemedBench: first benchmark for opaque industrial software remediation.
\end{itemize}

%% ============================================================
\section{Introduction}\label{sec:intro}
%% ============================================================
Industrial software supply chains are a growing attack surface for critical infrastructure. Executive Order 14028~\cite{eo14028}, NIST SP~800-82r3~\cite{nist800_82}, and IEC~62443~\cite{iec62443} collectively obligate operators of power, chemical, water, and manufacturing systems to assess and remediate vulnerabilities in their deployed software assets, yet that software is routinely delivered as stripped, symbol-free binaries from which source, debug information, and build environments have been permanently removed. We term this class of artifact \emph{Opaque Industrial Software} (OIS); Table~\ref{tab:ois_classes} catalogues its representative classes.

Binary analysis has made identifying vulnerability candidates within OIS increasingly tractable --- graph-neural detectors~\cite{ji2021buggraph}, binary taint analysis~\cite{nilo2020karonte,chen2021satc}, and knowledge-graph behavioural analysis~\cite{ning2026scaa} now produce structured alert sets on stripped firmware --- but alert generation is not remediation. No existing research framework provides an automated, end-to-end path from binary vulnerability candidates to validated remediation for OIS artifacts.

The reason this gap persists is structural. USENIX~2025 SoKs~\cite{hu2025sok,li2025sok_avr} show that contemporary automated vulnerability repair (AVR) systems --- LLM-based, program-analysis-based, and agentic --- achieve 44--90\% repair success across evaluated benchmarks, but every such system requires source code, a compilable artifact, a test suite, or a sanitizer execution trace as its primary input. Systems that operate ``at the binary level'' (e.g.\ VulShield~\cite{li2025vulshield}, CrashRepair~\cite{crashrepair2025}) still require the target to be instrumented, compiled, and executed in a controllable environment --- preconditions that OIS cannot satisfy. The class of software for which automated remediation is most urgently required is precisely the class for which no existing AVR system can operate.

Four specific structural deficits in the existing literature, summarised in Table~\ref{tab:gaps}, account for this persistent gap: source dependency of AVR systems (G1), binary analysis terminating at alert generation (G2), ICS/PLC tooling targeting safety rather than security repair (G3), and the absence of a domain-appropriate correctness criterion for OIS remediation (G4).

\begin{table}[!htbp]
\caption{Structural deficits in prior work that motivate SCARA. Each row corresponds to one of the four gaps surveyed; the right-hand column lists representative systems exhibiting the deficit.}\label{tab:gaps}
\small
\centering
\begin{tabular}{p{0.3cm}p{3.0cm}p{8.0cm}}
\toprule
\# & Structural deficit & Representative systems \\
\midrule
G1 & Source dependency of AVR systems & VRepair~\cite{chen2022vrepair}, VulRepair~\cite{fu2022vulrepair}, VulMaster~\cite{zhou2024vulmaster}, SAN2PATCH~\cite{jiang2025san2patch}, APPATCH~\cite{zhang2025appatch}, PatchAgent~\cite{patchagent2025}, Vul-R2~\cite{vulr22025}; CrashRepair~\cite{crashrepair2025}, CONCH~\cite{conch2024} \\
G2 & Binary analysis stops at alert generation & Firmadyne~\cite{chen2016firmadyne}, FirmAE~\cite{kim2020firmae}, HALucinator~\cite{halucinator2020}, Fuzzware~\cite{fuzzware2022}, FirmSolo~\cite{firmsolo2023}, SAFIREFUZZ~\cite{safirefuzz2023}; KARONTE~\cite{nilo2020karonte}, SaTC~\cite{chen2021satc} \\
G3 & ICS/PLC tools target safety, not security repair & VetPLC~\cite{ahmed2019vetplc}, SymPLC~\cite{symplc2017}, STAutoTester~\cite{stautotester2021}, ICSQuartz~\cite{zheng2025icsquartz}, PLCverif~\cite{darvas2017plcverif} \\
G4 & Appropriate correctness criterion for OIS remediation & --- (negative observation; no system addresses) \\
\bottomrule
\end{tabular}
\end{table}

Together, these four deficits create a precondition mismatch: existing repair systems assume source code, sanitizer traces, or executable test harnesses, whereas OIS remediation often begins from opaque artifacts and incomplete operational documentation.

To address the four gaps identified above, this paper presents SCARA (Semantics-Constrained
Autonomous Remediation Agent), a four-stage framework organized around a tiered remediation
model: Tier~1 mitigation (protocol and configuration protection policies), Tier~2 binary
hardening (runtime guards via binary instrumentation), and Tier~3 source repair
(SSCKG-constrained LLM-synthesized patches).  The five principal contributions, of which
C1--C4 close Gaps~1--4 respectively and C5 supplies the empirical substrate, are as follows.

\textbf{C1 --- Formal OIS remediation problem statement.}  We provide the first formal definition of the opaque industrial software remediation problem, comprising a nine-component operational state model, an availability-based tier mapping $T(A(B))$, a four-class reachability label taxonomy, and the correctness condition for conditional correctness evidence $\varepsilon(v, R_v)$.

\textbf{C2 --- Operational-state-aware reachability verification (OSVA).} OSVA provides reachability evidence under the full operational state model without source or instrumented execution, eliminating 40.0\% of binary alerts as infeasible.

\textbf{C3 --- Tiered SSCKG-constrained remediation synthesis (RSA).} RSA produces the strongest feasible remedy at the applicable tier, constrained in all tiers by the Software Supply Chain Knowledge Graph (SSCKG) behavioural specification.

\textbf{C4 --- Correctness validation (CVA).} CVA produces conditional correctness evidence $\varepsilon(v, R_v)$ combining post-remedy UNSAT confirmation, SSCKG behavioural-coverage preservation, domain-invariant compliance, and independent replay across emulated, protocol-harness, and soft-PLC environments.

\textbf{C5 --- OIS-RemedBench.} The first benchmark for opaque industrial software vulnerability remediation, spanning OIS-Binary, OIS-Protocol, and OIS-ICS partitions with stratified L1--L4 labels.

%% ============================================================
\section{Background and Problem Formulation}\label{sec:background}
%% ============================================================

\subsection{Opaque Industrial Software: Definition and Characteristics}\label{sec:bg:ois}

\begin{definition}[Opaque Industrial Software]
An artifact $B$ is an \emph{Opaque Industrial Software} (OIS) artifact if $B$ is an
industrial software component delivered without (i)~source code, (ii)~build
environment, (iii)~hardware interfaces, (iv)~runtime dependencies, and (v)~operational
configuration.  As a consequence, vulnerability verification and remediation using
conventional source-based methods are infeasible on $B$.
\end{definition}

Three features characterize OIS in practice.  First, \emph{stripping}: OIS binaries
are systematically stripped of debug symbols, function names, type information, and
section metadata before delivery.  Second, \emph{closed toolchains}: the compilers,
linkers, and build configurations used to produce OIS are typically proprietary and
inaccessible, precluding recompilation from any reconstructed source representation.
Third, \emph{physical execution dependencies}: OIS execution depends on physical
hardware --- sensor inputs, actuator outputs, NVRAM state, proprietary I/O modules, and
licensed fieldbus stacks --- that cannot be fully replicated in generic software
emulation.

OIS artifacts can be classified by their availability class $A(B)$, which determines
which remediation tiers SCARA can apply:

\begin{table}[!htbp]
\caption{OIS artifact classes grouped by availability class $A(B) \in \{\text{policy-only},\ \text{binary-rewritable},\ \text{source-available}\}$, with the applicable remediation tier set $T(A)$ (Eq.~\eqref{eq:tier_map}) and the success-counting rule used in Section~\ref{sec:results}. \textsc{Sat-relaxed} cases are counted only after replay confirmation; \textsc{Unsat} resolves the candidate as a false positive (no remedy issued); \textsc{Unknown} yields a Tier-1 advisory only.}\label{tab:ois_classes}
\small
\centering
\begin{tabular}{p{4.5cm}p{3.5cm}p{4.5cm}}
\toprule
Artifact class (example) & Tier set $T(A)$ & Success counted when \\
\midrule
\multicolumn{3}{l}{\textit{Policy-only ($A = \text{policy-only}$)}} \\
Binary firmware / closed Modbus--DNP3 libs / vendor PLC image & Tier~1 only & CVA accepts T1 policy \\
\midrule
\multicolumn{3}{l}{\textit{Binary-rewritable ($A = \text{binary-rewritable}$)}} \\
PIE ELF daemons; re-flashable ARM firmware & Tier~2; T1 fallback & CVA accepts T2 (or T1 fallback) \\
\midrule
\multicolumn{3}{l}{\textit{Source-available ($A = \text{source-available}$)}} \\
libmodbus, open62541; OpenPLC / MATIEC-C~\cite{matiec} & Tier~3; T2/T1 fallback & CVA accepts T3 (or T2/T1 fallback) \\
\bottomrule
\end{tabular}
\end{table}

\subsection{Code Property Graphs and Software Supply Chain Knowledge Graphs}\label{sec:bg:ssckg}

SCARA operates on representations derived from prior binary analysis, specifically the
Software Supply Chain Knowledge Graph (SSCKG) produced by~\cite{ning2026scaa}.

\textbf{Code Property Graph (CPG).}  A CPG~\cite{yamaguchi2014cpg} is a graph $G_{\text{cpg}} = (V_{\text{cpg}}, E_{\text{cpg}})$ that unifies the abstract syntax tree (AST), control-flow graph (CFG), and program dependence graph (PDG) of a program into a single queryable structure.  SCARA consumes a CPG produced by the upstream binary-analysis pipeline as the lowest-level structural input to the SSCKG transformation defined below.

\textbf{Abstract Domain $\mathcal{A}$.}  The abstract domain $\mathcal{A}$ is a finite,
hierarchical lattice of behavioral security labels instantiated from the MITRE ATT\&CK
for ICS taxonomy~\cite{mitre_ics}.  It comprises five macro-behavior categories,
27~action labels, and 43~risk labels.

\textbf{SSCKG and risk-relevant relations.}  The SSCKG is a behavioural graph $G_{\text{ssckg}} = (\mathcal{E}, \mathcal{R})$ derived from $G_{\text{cpg}}$ via a surjective transformation $\Phi: G_{\text{cpg}} \to G_{\text{ssckg}}$. The entity set $\mathcal{E}$ consists of behavioural clusters of semantically related CPG nodes labelled in $\mathcal{A}$; the relation set $\mathcal{R}$ encodes a typed vulnerability-relation alphabet $\Sigma_{\mathcal{R}}$. SCARA operates on the \emph{risk-relevant} subset $\Lambda_{\text{risk}} \subseteq \Sigma_{\mathcal{R}}$ comprising those relation types that can transmit attacker influence, vulnerability preconditions, or remediation side effects across entities --- namely data flow, control dependency, shared memory access, inter-process communication, protocol interaction, and cross-component call. The construction of $G_{\text{ssckg}}$, the full relation taxonomy, the empirical compression ratio against the raw CPG, and the per-entity composite risk score $\rho_r(v) \in [0, 1]$ used by SCARA are properties of the upstream representation and are reported in the SCAA paper.

\subsection{Operational State Model}\label{sec:bg:osm}

Existing symbolic execution frameworks model execution state as a triple $(pc, mem, \varphi)$.  This representation is insufficient for OIS, where vulnerability reachability critically depends on the external operational context.  We extend the standard execution state to a nine-component operational state model:

\begin{equation}\label{eq:state_model}
S = (pc,\ mem,\ \varphi,\ env,\ io,\ q_{\text{proto}},\ q_{\text{runtime}},\ q_{\text{component}},\ t)
\end{equation}

Table~\ref{tab:state_model} defines the semantics of each component.  As a concrete
instance, consider a buffer-overflow candidate in a Modbus coil-write handler: $pc$
locates the basic block performing the unbounded copy; $q_{\text{proto}}$ is the
Modbus FSM state \texttt{PostUnitID}; $C_{\text{io}}$ requires the addressed coil to
lie within the configured coil map of the target device; $q_{\text{runtime}}$
restricts execution to the cyclic-scan phase rather than the boot phase; and
$C_{\text{time}}$ bounds the handler response within the protocol's reply-window
deadline.  A candidate that is satisfiable over $\varphi_{\text{path}}$ alone may
become UNSAT once these external constraints are imposed, which is precisely the
class of false positives that OSVA eliminates.

\begin{table}[!htbp]
\caption{Nine-component operational state model $S$.}\label{tab:state_model}
\small
\centering
\begin{tabular}{lp{4.5cm}p{5.0cm}}
\toprule
Component & Semantics & Industrial example \\
\midrule
$pc$ & Program counter & Current block in Modbus handler \\
$mem$ & Memory and register state & NVRAM register map, I/O table \\
$\varphi$ & Path constraint & SMT formula over inputs \\
$env$ & Config., NVRAM, environment & PLC project config., NVRAM \\
$io$ & Device I/O, fieldbus input & Sensor readings, coil states \\
$q_{\text{proto}}$ & Protocol FSM state & Modbus session state \\
$q_{\text{runtime}}$ & Runtime lifecycle phase & Boot phase, scan cycle $N$ \\
$q_{\text{component}}$ & Cross-binary comm.\ state & IPC queue depth \\
$t$ & Timing constraint & Scan-cycle deadline \\
\bottomrule
\end{tabular}
\end{table}

The \emph{verification condition} for a candidate vulnerability path from source entity
$src$ to sink entity $snk$ is satisfiable if:

\begin{equation}\label{eq:vc}
\varphi_{\text{path}} \wedge C_{\text{env}} \wedge C_{\text{io}} \wedge C_{\text{proto}}
\wedge C_{\text{runtime}} \wedge C_{\text{component}} \wedge C_{\text{time}} \text{ is SAT}
\end{equation}

Each $C_x$ denotes the conjunction of constraints derived from operational-state component $x$, with the extraction procedure detailed in Section~\ref{sec:scara}. SCARA employs a four-class reachability label taxonomy --- \textsc{Sat-strict} (satisfiable under all constraints with witness $(I^*, S^*)$; passed to RSA), \textsc{Sat-relaxed} (satisfiable only after relaxing an under-documented constraint; passed to RSA but requires replay confirmation), \textsc{Unsat} (infeasible under all constraints; recorded as a refutation certificate, no remedy issued), and \textsc{Unknown} (budget exhausted or modelling gap; Tier-1 advisory only) --- which corresponds to the L2 row of the OIS-RemedBench evidence ladder in Table~\ref{tab:label_levels}.

The \emph{correctness condition for conditional correctness evidence $\varepsilon(v, R_v)$}
requires that all four of the following hold:
\begin{enumerate}
\item Re-running OSVA on the post-remedy artifact $B'$ returns UNSAT for the original
  candidate path.
\item SSCKG behavioral coverage is preserved:
  $\text{BCP}(G_{\text{ssckg}}, G'_{\text{ssckg}}) \geq \tau_{\text{cov}}$, where
  $\text{BCP}$ is the Behavioral Coverage Preservation ratio measuring the fraction of
  non-vulnerable SSCKG relations retained after remediation (formally defined in
  Section~\ref{sec:scara}) and $\tau_{\text{cov}}$ is its acceptance threshold.
\item No domain invariant is violated.
\item Where a replay harness is available (L3-labeled cases), independent replay
  confirms the UNSAT outcome.
\end{enumerate}

\subsection{Formal Problem Statement}\label{sec:bg:problem}

\textbf{Input:}
\begin{itemize}
\item Binary artifact $B$ (stripped; no source code available);
\item SSCKG $G_{\text{ssckg}} = (\mathcal{E}, \mathcal{R})$ with behavioral annotations and composite risk scores $\rho_r(v)$;
\item Vulnerable entity set $V_r = \{v \in \mathcal{E} \mid \rho_r(v) \geq \tau\}$;
\item Artifact availability class $A(B) \in \{\text{policy-only}, \text{binary rewritable}, \text{source-available}\}$, determining the tier set
\begin{equation}\label{eq:tier_map}
T(A) = \begin{cases}
\{1\} & \text{if } A = \text{policy-only} \\
\{1, 2\} & \text{if } A = \text{binary-rewritable} \\
\{1, 2, 3\} & \text{if } A = \text{source-available.}
\end{cases}
\end{equation}
\end{itemize}

\textbf{Output:} For each $v \in V_r$: a reachability label $\ell_v$, with witness $(I^*, S^*)$ for SAT labels or refutation constraint $\neg C_j$ for UNSAT; for each $v$ with $\ell_v \in \{\text{SAT-strict}, \text{SAT-relaxed}\}$: a remediation artifact $R_v$ at the strongest tier in $T(A)$ satisfying the correctness condition; and conditional correctness evidence $\varepsilon(v, R_v)$ for each produced remedy.

\textbf{Non-goals.}  SCARA does not discover new vulnerabilities; that function is performed by the upstream binary analysis tool.  SCARA does not provide formal proofs of global safety properties; the correctness evidence it produces is conditional on the completeness of the modeled operational state $S$.

%% ============================================================
\section{The SCARA Framework}\label{sec:scara}
%% ============================================================
\subsection{System Overview}\label{sec:scara:overview}

SCARA decomposes the formal problem across four agents, shown as the pipeline of Figure~\ref{fig:architecture}: CACA normalises heterogeneous vulnerability alerts and estimates the operational context (ablations A8, A9); OSVA performs budgeted reachability verification with controlled relaxation (A1--A3); RSA synthesises the strongest feasible remedy with $\delta$-feedback ingestion (A5, A6); and CVA either issues conditional correctness evidence $\varepsilon(v, R)$ or returns a typed rejection constraint $\delta$ to RSA for at most $K=3$ resynthesis iterations per tier (A4, A7). The four-stage decomposition is motivated by two properties of the OIS-remediation problem: each stage has a distinct failure mode that monolithic systems cannot ablate, and the closed-loop CVA$\to$RSA path with typed $\delta$ enables the iterative refinement that single-pass pipelines lack.

\begin{figure}[!htbp]
  \centering
  \includegraphics[width=0.95\textwidth]{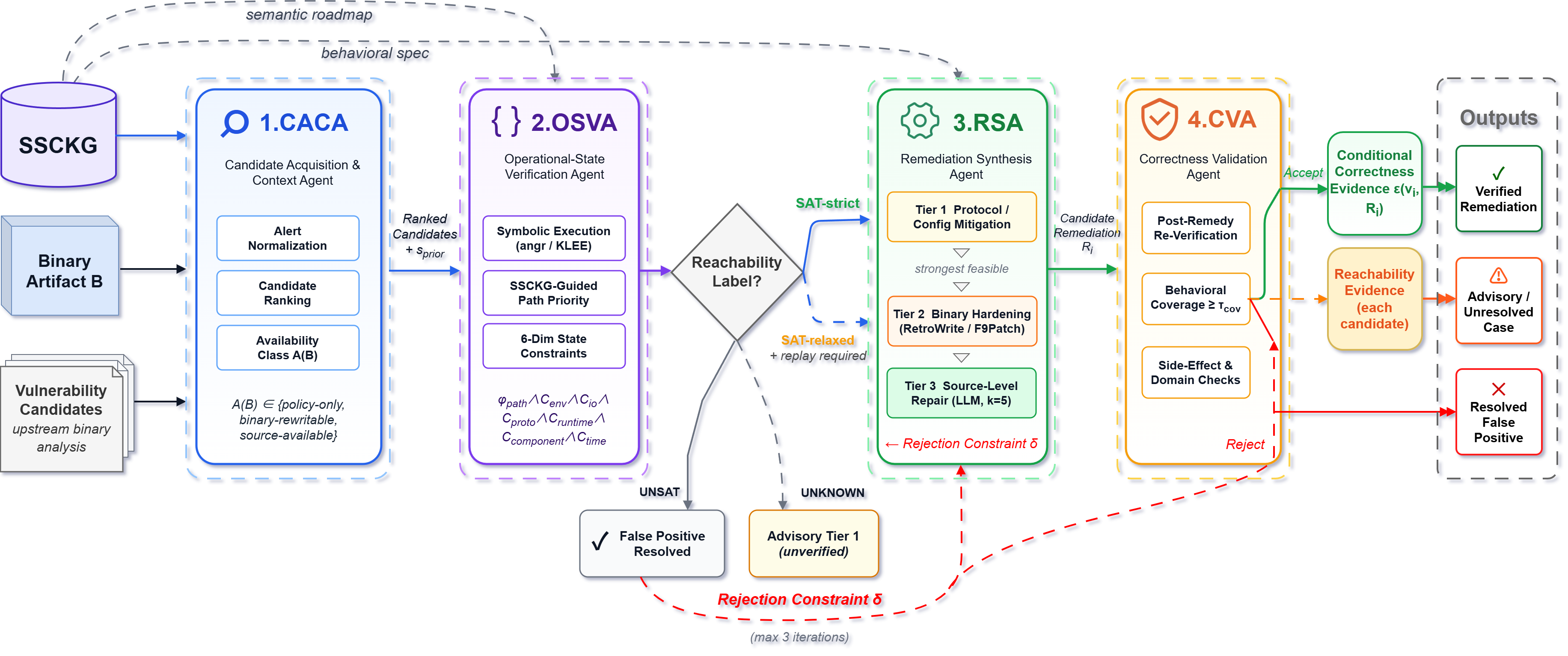}
  \caption{SCARA four-stage pipeline. CACA normalises heterogeneous candidate evidence; OSVA verifies operational reachability under the nine-component state model; RSA synthesises the strongest feasible remedy at the artifact's availability tier; and CVA closes the loop by either issuing conditional correctness evidence $\varepsilon(v, R)$ or returning a typed rejection constraint $\delta$ to RSA for at most $K=3$ resynthesis iterations per tier.}
  \label{fig:architecture}
\end{figure}

The operational context $\Omega$ supplied to CACA is the optional external evidence used to construct $S_i^{\text{prior}}$, to select the relevant constraint families, and to infer the availability class; it is a tuple of per-dimension hints
\begin{equation}\label{eq:omega}
\Omega = (\Omega_{\text{art}},\ \Omega_{\text{env}},\ \Omega_{\text{io}},\ \Omega_{\text{proto}},\ \Omega_{\text{runtime}},\ \Omega_{\text{component}},\ \Omega_{\text{time}},\ \Omega_{\text{replay}})
\end{equation}
covering artifact metadata, environment/NVRAM/configuration, MMIO/fieldbus, protocol FSM, runtime lifecycle, cross-binary topology, timing budgets, and replay-harness availability respectively (full schema in the supplementary material).

\begin{algorithm}[!htbp]
\caption{SCARA Controller}\label{alg:controller}
\begin{algorithmic}[1]
\Require Binary artifact $B$, SSCKG $G_{\text{ssckg}}$, candidate alert set $\mathcal{V}_r$, operational context $\Omega$ (Eq.~\eqref{eq:omega})
\Ensure Per-candidate outcomes: verified remediations, resolved FPs, unconfirmed records, advisories
\State $C \gets \text{CACA.normalize\_and\_rank}(\mathcal{V}_r, G_{\text{ssckg}}, B, \Omega)$
\For{each candidate $c$ in $C$}
    \State $\text{label}, \text{witness}, \text{reason} \gets \text{OSVA.verify}(B, G_{\text{ssckg}}, c)$ \Comment{Alg.~\ref{alg:osva}}
    \If{$\text{label} = \textsc{Unsat}$}
        \State $\text{record\_resolved\_false\_positive}(c, \text{reason})$ \Comment{Prop.~\ref{prop:unsat}}
        \State \textbf{continue}
    \EndIf
    \If{$\text{label} = \textsc{Unknown}$}
        \State $\text{record\_unresolved\_or\_advisory}(c, \Omega)$ \Comment{see below; not routed through Alg.~\ref{alg:rsacva}}
        \State \textbf{continue}
    \EndIf
    \If{$\text{label} = \textsc{Sat-relaxed}$ \textbf{and not} $\text{replay\_available}(\text{witness})$}
        \State $\text{record\_unconfirmed\_candidate}(c, \text{witness})$
        \State \textbf{continue}
    \EndIf
    \State $\text{tiers} \gets \text{feasible\_tiers}(\text{CACA.availability\_class}(B))$
    \State $\text{result} \gets \text{RSA\_CVA\_loop} (B, G_{\text{ssckg}}, \text{witness}, \text{label}, \text{tiers})$ \Comment{Alg.~\ref{alg:rsacva}; $K=3$}
    \State $\text{record}(\text{result})$
\EndFor
\end{algorithmic}
\end{algorithm}

\subsection{Candidate Acquisition and Context Agent}\label{sec:scara:caca}

CACA transforms heterogeneous vulnerability evidence into verification tasks.  Its inputs are an SSCKG $G_{\text{ssckg}}$, node-level risk scores $\rho_r(v)$, the raw binary artifact $B$, and optional context.  Its output is a candidate set, ranked downstream by Eq.~\eqref{eq:rank}:

\begin{equation}\label{eq:caca_set}
C = \{(v_i,\ \text{src}_i,\ \text{snk}_i,\ \rho_i,\ S_i^{\text{prior}})\}
\end{equation}

where $v_i$ is the SSCKG entity associated with the candidate, $\text{src}_i$ and $\text{snk}_i$ denote the suspected source and sink, $\rho_i$ is the candidate risk score, and $S_i^{\text{prior}}$ is the initial estimate of the operational state required for verification.

CACA accepts vulnerability candidates from multiple upstream generators because OIS artifacts differ in available evidence. For source-available or intermediate-representation cases, candidates may originate from static analyzers or source-to-sink rules. For binary-only firmware, candidates may derive from multi-binary taint analysis, keyword-guided firmware analysis, CVE-linked advisory matching against the NVD~\cite{nvd}, or fuzzing crashes mapped back to SSCKG entities. CACA normalizes these inputs by mapping each candidate to the corresponding SSCKG entity and relation type. This mapping is necessary because OSVA's constraint encoding depends on the semantic relation associated with the candidate: a protocol-interaction relation suggests constraints over $q_{\text{proto}}$, whereas a memory-mapped I/O relation suggests constraints over $io$ and $q_{\text{runtime}}$. Unlike upstream generators such as KARONTE~\cite{nilo2020karonte} and SaTC~\cite{chen2021satc}, which produce taint-derived alerts without operational context, CACA enriches each normalized candidate with the state dimensions required for OSVA's constraint encoding, thereby converting a raw alert into a verification task with a defined scope.

\begin{figure}[!htbp]
	\centering
	\includegraphics[width=0.7\textwidth]{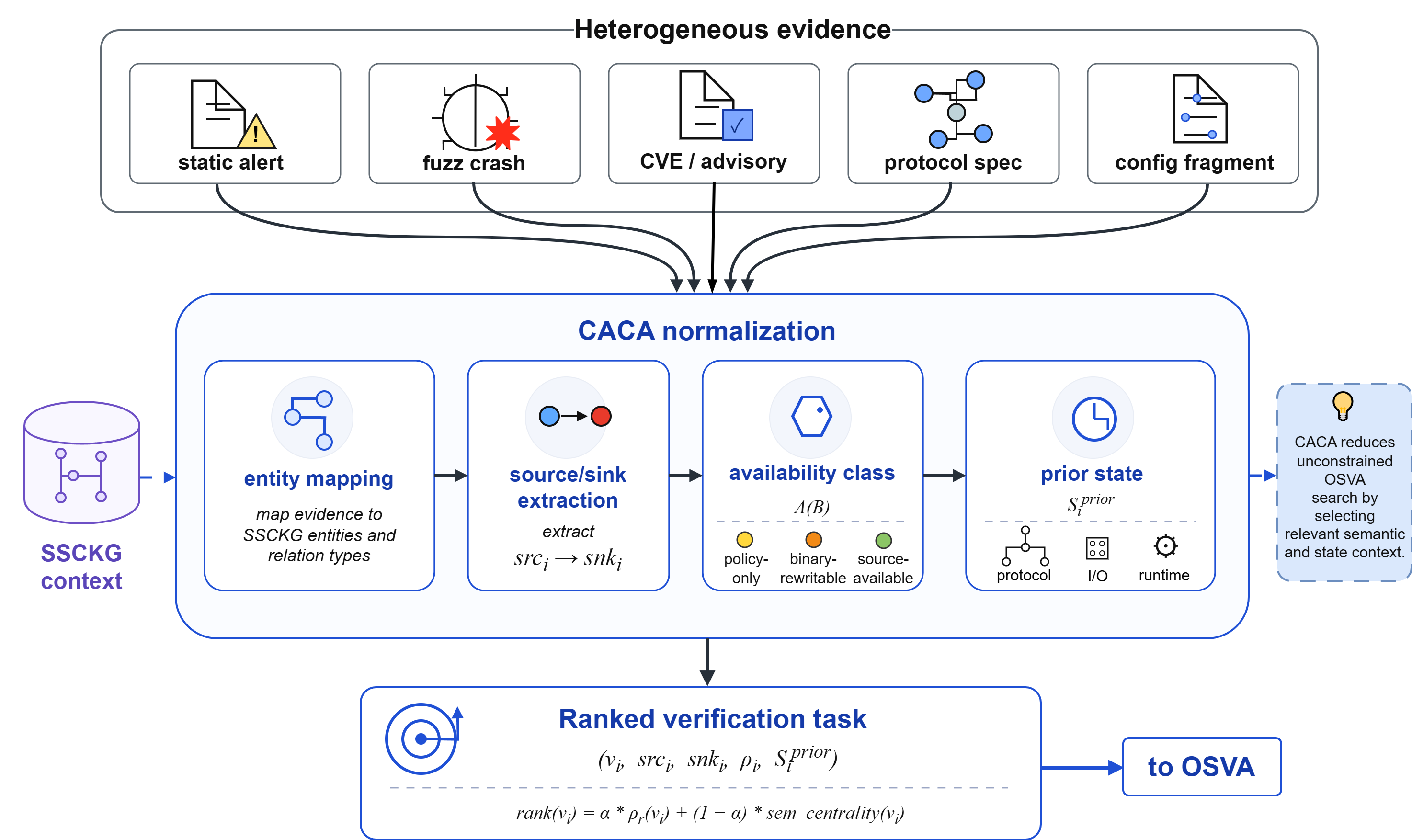}
	\caption{CACA normalisation funnel. Heterogeneous vulnerability evidence (static-analyser candidates, taint-derived alerts, fuzzing crashes, CVE-linked advisories) is mapped to a uniform $(v, \text{src}, \text{snk}, S^{\text{prior}})$ tuple keyed to the SSCKG, producing a single verification task that OSVA can consume.}
	\label{fig:caca}
\end{figure}

CACA ranks candidates using both the risk score inherited from upstream analysis and the semantic centrality of the candidate entity in the SSCKG:

\begin{equation}\label{eq:rank}
\text{rank}(v_i) = \alpha \cdot \rho_r(v_i) + (1 - \alpha) \cdot \text{sem\_centrality}(v_i),\qquad \alpha \in [0, 1].
\end{equation}

The parameter $\alpha$ balances local risk against graph-level importance and is fixed at $\alpha = 0.6$ in the primary experiments, calibrated on the OIS-RemedBench validation subset; the value gives slightly more weight to the upstream risk score than to SSCKG semantic centrality. This design is adopted because high-risk nodes that are semantically isolated may be less urgent than moderately risky nodes that lie on many source-to-sink paths. Semantic centrality is formally defined as the normalized betweenness centrality of $v_i$ in the SSCKG $G_{\text{ssckg}} = (\mathcal{E}, \mathcal{R})$ of Section~\ref{sec:bg:ssckg}, restricted to the risk-relevant relation subset $\Lambda_{\text{risk}} \subseteq \Sigma_{\mathcal{R}}$:

\begin{equation}\label{eq:sem_centrality}
\text{sem\_centrality}(v_i) = \frac{1}{|\mathcal{E}|(|\mathcal{E}|-1)} \sum_{\substack{s, t \in \mathcal{E},\ s \neq t \\ s, t \neq v_i}} \frac{\sigma_{st}(v_i)}{\sigma_{st}}
\end{equation}

where $\sigma_{st}$ is the number of shortest paths between entities $s$ and $t$ over edges in $\Lambda_{\text{risk}}$, and $\sigma_{st}(v_i)$ is the subset passing through $v_i$. Betweenness centrality over the SSCKG is adopted rather than over the raw binary CFG because a node with high SSCKG betweenness lies on many source-to-sink behavioral paths and therefore has disproportionate operational impact at the industrial protocol or control level.

CACA also determines the artifact availability class $A(B) \in \{\text{policy-only},\ \text{binary-rewritable},\ \text{source-available}\}$, which selects the applicable tier set $T(A)$ defined in Section~\ref{sec:bg:problem} (Eq.~\eqref{eq:tier_map}): policy-only artifacts admit only Tier~1; binary-rewritable artifacts admit Tier~1 and Tier~2; source-available artifacts admit all three tiers.

\subsection{Operational-State Verification Agent}\label{sec:scara:osva}

OSVA operates over the operational state model $S$ defined in Section~\ref{sec:bg:osm} (Eq.~\eqref{eq:state_model}) and the six constraint families $C_{\text{env}}, C_{\text{io}}, C_{\text{proto}}, C_{\text{runtime}}, C_{\text{component}}, C_{\text{time}}$ introduced therein.  The prior state estimate $S_i^{\text{prior}}$ supplied by CACA initialises these components before OSVA's constraint encoding begins.

OSVA verifies whether a candidate path is reachable under the modeled operational state. Given a ranked candidate $(v_i, \text{src}_i, \text{snk}_i, S_i^{\text{prior}})$ and binary artifact $B$, OSVA returns a reachability label $\ell_i \in \{\textsc{Sat-strict},\ \textsc{Sat-relaxed},\ \textsc{Unsat}, \textsc{Unknown}\}$.  For satisfiable candidates, OSVA also returns a witness $(I^*, S^*)$, where $I^*$ is a concrete or partially concrete input vector and $S^*$ is the corresponding operational state. For infeasible candidates, OSVA returns the violated constraint or constraint family responsible for the refutation. The verification condition that OSVA discharges is the operational reachability condition of Eq.~\eqref{eq:vc}: a candidate is reachable only if $\varphi_{\text{path}} \wedge \bigwedge_{C_k \in \Sigma_S} C_k$ is satisfiable under the encoded constraints.

\subsubsection{Symbolic Execution Setup}

OSVA selects the symbolic execution engine according to the artifact type.  For binary-only firmware and stripped services, OSVA uses angr~\cite{shoshitaishvili2016angr}, because angr's SimProcedure framework provides behavioral abstractions for unresolved peripheral calls, vendor runtime services, and MMIO operations that KLEE cannot model without source-level stubs. For artifacts that can be reliably lifted to LLVM IR using RetDec or McSema, OSVA operates at the IR level; this path is taken only when lifting quality can be verified, because unreliable lifting introduces spurious paths. For the OIS-ICS partition, where IEC 61131-3 Structured Text programs may be translated into C through MATIEC, OSVA compiles the generated C into LLVM IR and executes it with KLEE~\cite{cadar2008klee}, because KLEE provides path-complete symbolic exploration under LLVM IR that is reliable when the source-to-IR translation faithfully preserves the structured-text semantics. SCARA does not assume that every OIS artifact can be reduced to LLVM IR; the angr path is the fallback for all cases where IR lifting is unavailable or unreliable.  Within the angr path, for protocol-handler cases, OSVA initialises symbolic execution from the state implied by $q_{\text{proto}}^{\text{prior}}$, avoiding the exploration of protocol states that cannot precede the candidate sink; for binary firmware, OSVA represents peripheral inputs, configuration values, and external service responses as symbolic variables constrained by available documentation or inferred SSCKG relation types. These models are deliberately conservative: under-documented dimensions are recorded and may lead to \textsc{Unknown} or \textsc{Sat-relaxed}, rather than being silently treated as complete.

\subsubsection{SSCKG-Guided Path Prioritization}

OSVA employs the SSCKG as a semantic roadmap for symbolic exploration. The substantive insight here is the use of the SSCKG itself as the semantic prior over binary-level path exploration; the SBERT-based encoding and cosine similarity below are off-the-shelf instruments for operationalising it. The SSCKG path from $\text{src}_i$ to $\text{snk}_i$ encodes behaviorally meaningful steps that should be reflected in the binary-level path if the candidate is genuine. OSVA first projects this SSCKG path onto the binary CFG using CACA's entity-to-function mapping, then scores each candidate CFG path according to its semantic alignment with the SSCKG path:

\begin{equation}\label{eq:score}
\text{score}(p_j) = \text{sim}\!\left(\text{SBERT}(\text{labels}(p_j)),\ \text{SBERT}(\text{labels}(P_{\text{ssckg}}))\right)
\end{equation}

Formally, let $\text{labels}(p_j) = (l_1, \ldots, l_m)$ be the ordered sequence of SSCKG entity labels assigned to basic blocks along binary path $p_j$, obtained via CACA's entity-to-function mapping, and let $\text{labels}(P_{\text{ssckg}}) = (l'_1, \ldots, l'_k)$ be the entity label sequence of the SSCKG source-to-sink path. Both sequences are encoded using the SBERT~\cite{reimers2019sbert} sentence-transformer model (all-mpnet-base-v2, 768-dimensional embeddings). The similarity is the standard cosine similarity:

\begin{equation}\label{eq:sim}
\text{sim}(\mathbf{u}, \mathbf{v}) = \frac{\mathbf{u} \cdot \mathbf{v}}{\|\mathbf{u}\| \cdot \|\mathbf{v}\|} \in [-1,\ 1].
\end{equation}

Because cosine similarity is bounded in $[-1, 1]$, scores are mapped to $[0, 1]$ by an affine transformation faithful to Eq.~\eqref{eq:sim}, $\tilde{s}_j = (\text{score}(p_j) + 1) / 2$, rather than by min--max normalisation over the candidate path set; the latter would make the budget allocation depend on the worst-aligned explored path and could exaggerate small score differences when all paths are similarly aligned. OSVA then allocates solver budget proportionally to semantic alignment via an explicit softmax over the explored binary-level path set $P$:

\begin{equation}\label{eq:budget}
\text{budget}(p_j) = T_{\text{total}} \cdot \frac{\exp\!\left(\tilde{s}_j / \tau_p\right)}{\sum_{p_k \in P} \exp\!\left(\tilde{s}_k / \tau_p\right)},
\end{equation}

where $T_{\text{total}}$ is the total per-candidate strict-pass solver budget and $\tau_p$ is the softmax temperature. In the primary experiments $T_{\text{total}} = 300$~s per candidate and $\tau_p = 0.5$, calibrated on the held-out validation subset; the relaxed pass uses a separate budget $T_{\text{relaxed}} = 150$~s. High-scoring paths therefore receive more solver time earlier, while low-scoring paths remain available if the high-priority budget is exhausted. This distinction matters because hard pruning would reduce the false-positive rate at the cost of missing true positives. SCARA employs prioritization rather than pruning so that false-negative rates can be measured and recall can be reported alongside false-positive reduction.

\begin{figure}[!htbp]
	\centering
	\includegraphics[width=0.95\textwidth]{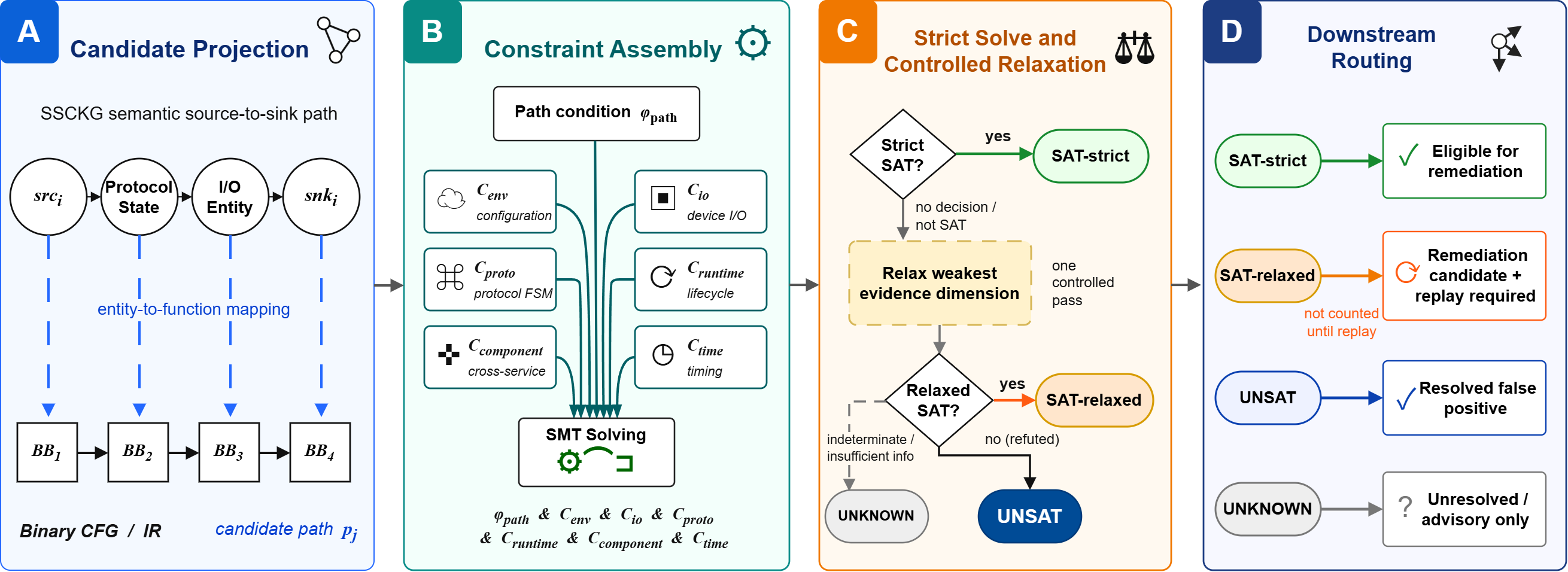}
	\caption{Operational-state reachability analysis in OSVA. A candidate source-to-sink path is projected from the SSCKG to the executable representation and checked against path, environment, I/O, protocol, runtime, component, and timing constraints.
	}\label{fig:osva}
\end{figure}

\subsubsection{Operational-State Constraint Encoding}

OSVA translates the prior estimate $S_i^{\text{prior}}$ into SMT constraints over symbolic variables, one per family in $\Sigma_S$ (Table~\ref{tab:state_model}). Concrete encodings include: $C_{\text{env}}$ over NVRAM ranges and PLC project parameters; $C_{\text{io}}$ over MMIO/fieldbus value ranges (e.g., a Modbus register address restricted to $[0, 65535]$, a coil to a Boolean domain); $C_{\text{proto}}$ over protocol FSM transitions (e.g., session establishment must precede a write); $C_{\text{runtime}}$ over boot/scan-cycle phases and initialisation status; $C_{\text{component}}$ over IPC, shared-memory, or file-mediated transfers; and $C_{\text{time}}$ over watchdogs, scan-cycle deadlines, and guard-overhead budgets.

If OSVA cannot decide a candidate under the strict constraint set within the solver budget, it performs one controlled relaxation pass targeting the constraint dimension with the weakest evidence --- for example, a configuration variable inferred from incomplete documentation. A satisfiable result from the first (strict) pass is labeled \textsc{Sat-strict}. A satisfiable result obtained only after relaxation is labeled \textsc{Sat-relaxed} and is not counted as confirmed until replay validates the witness. If the second pass also fails to establish satisfiability or unsatisfiability, OSVA returns \textsc{Unknown} and records the missing or uncertain dimension to guide manual follow-up.

\subsubsection{Reachability Labels and Downstream Routing}

The four OSVA labels carry different downstream consequences. \textsc{Sat-strict} candidates are eligible for remediation synthesis. \textsc{Sat-relaxed} candidates may be forwarded to RSA, but their remediation is counted as successful only after independent replay confirms the relaxed witness. \textsc{Unsat} candidates are reported as false positives relative to the encoded operational-state model; no remedy is generated, because there is no reachable vulnerability to remediate. \textsc{Unknown} candidates are reported as unresolved; SCARA may produce an advisory Tier 1 policy if an enforcement point exists, but this advisory is not counted as validated remediation.

This routing rule prevents two common errors in alert-driven repair. First, it avoids converting infeasible alerts into unnecessary operational restrictions. Second, it avoids treating conservative policies generated for uncertain cases as evidence of repair success. The distinction is central to SCARA's evaluation, where false-positive refutation, verified remediation, and advisory mitigation are reported as separate outcome categories.

\begin{algorithm}[!htbp]
\caption{OSVA.verify. The strict pass is per-path under budget $T_{\text{total}}$ allocated by Eq.~\eqref{eq:budget}; the relaxed pass is a single global pass over the disjunction of explored path conditions under budget $T_{\text{relaxed}}$, relaxing the weakest-evidence constraint family.}\label{alg:osva}
\begin{algorithmic}[1]
\Require Artifact $B$, SSCKG $G_{\text{ssckg}}$, candidate $c = (v, \text{src}, \text{snk}, S^{\text{prior}})$
\Ensure Label $\ell$, witness $(I^*, S^*)$ or $\varnothing$, reason
\State $P \gets \text{project\_ssckg\_path\_to\_executable\_paths}(G_{\text{ssckg}}, B, c)$
\State $P_{\text{explored}} \gets \emptyset$;\quad $P_{\text{refuted}} \gets \emptyset$;\quad $P_{\text{inconclusive}} \gets \emptyset$
\For{each path $p$ in $\text{prioritize\_by\_semantic\_alignment}(P, G_{\text{ssckg}}, c)$}
    \State $\varphi_p \gets \text{path\_condition}(p)$ \Comment{local; distinct from SSCKG morphism $\Phi$}
    \State $\mathcal{C} \gets \text{encode\_constraints}(c.S^{\text{prior}}, B, p)$
    \State $\text{status}, \text{model} \gets \text{solve}(\varphi_p \wedge \mathcal{C},\ \text{budget}(p))$ \Comment{Eq.~\eqref{eq:budget}}
    \State $P_{\text{explored}} \gets P_{\text{explored}} \cup \{p\}$
    \If{$\text{status} = \textsc{Sat}$}
        \State \Return \textsc{Sat-strict}, $\text{model}$, $\varnothing$ \Comment{Prop.~\ref{prop:sat}}
    \ElsIf{$\text{status} = \textsc{Unsat}$}
        \State $P_{\text{refuted}} \gets P_{\text{refuted}} \cup \{p\}$;\quad $\text{collect\_refuting\_constraint}(p, \mathcal{C})$
    \Else \Comment{$\textsc{Unknown}$ / solver timeout on this path}
        \State $P_{\text{inconclusive}} \gets P_{\text{inconclusive}} \cup \{p\}$;\quad $\text{mark\_path\_inconclusive}(p)$
    \EndIf
\EndFor
\State $\Phi_P \gets \bigvee_{p \in P_{\text{explored}}} \varphi_p$ \Comment{disjunction over explored path conditions}
\State $\mathcal{C}_{\text{relaxed}} \gets \text{relax\_weakest\_evidence\_dimension}(\mathcal{C})$
\State $\text{status}, \text{model} \gets \text{solve}(\Phi_P \wedge \mathcal{C}_{\text{relaxed}},\ T_{\text{relaxed}})$
\If{$\text{status} = \textsc{Sat}$}
    \State \Return \textsc{Sat-relaxed}, $\text{model}$, $\text{relaxed\_dimension}(\mathcal{C}, \mathcal{C}_{\text{relaxed}})$
\EndIf
\If{$P_{\text{refuted}} = P_{\text{explored}}$ \textbf{and} $\text{status} = \textsc{Unsat}$}
    \State \Return \textsc{Unsat}, $\varnothing$, $\text{collected\_refutation}$ \Comment{Prop.~\ref{prop:unsat}}
\EndIf
\State \Return \textsc{Unknown}, $\varnothing$, $\text{timeout\_or\_model\_gap}$
\end{algorithmic}
\end{algorithm}

\subsection{Remediation Synthesis Agent}\label{sec:scara:rsa}

RSA synthesizes the strongest feasible remediation artifact for candidates that OSVA has labeled \textsc{Sat-strict} or, conditionally, \textsc{Sat-relaxed}. Its inputs are the OSVA result, the candidate entity $v_i$, the SSCKG, the prior state model, and the tier set $T(A)$ associated with the artifact availability class. RSA outputs a candidate remediation $R_i$ at one of three tiers.

RSA follows a strongest-feasible-tier policy, where tier strength is ordered as Tier 3 $\succ$ Tier 2 $\succ$ Tier 1, reflecting increasing remediation depth: Tier 3 removes the vulnerable code path at the source level, Tier 2 inserts a runtime guard into the binary, and Tier 1 enforces a policy restriction at an external enforcement boundary. For a policy-only artifact, RSA attempts only Tier 1 mitigation. For a binary-rewritable artifact, RSA attempts Tier 2 binary hardening with Tier 1 as fallback. For a source-available artifact, RSA attempts Tier 3 source repair with Tier 2 and Tier 1 as fallbacks. If CVA rejects a remedy, RSA incorporates the rejection constraint and regenerates a remedy, up to a fixed iteration budget. If all feasible tiers fail validation, the candidate is marked \textsc{Remediation-Failed}.

\begin{figure}[!htbp]
	\centering
	\includegraphics[width=0.95\textwidth]{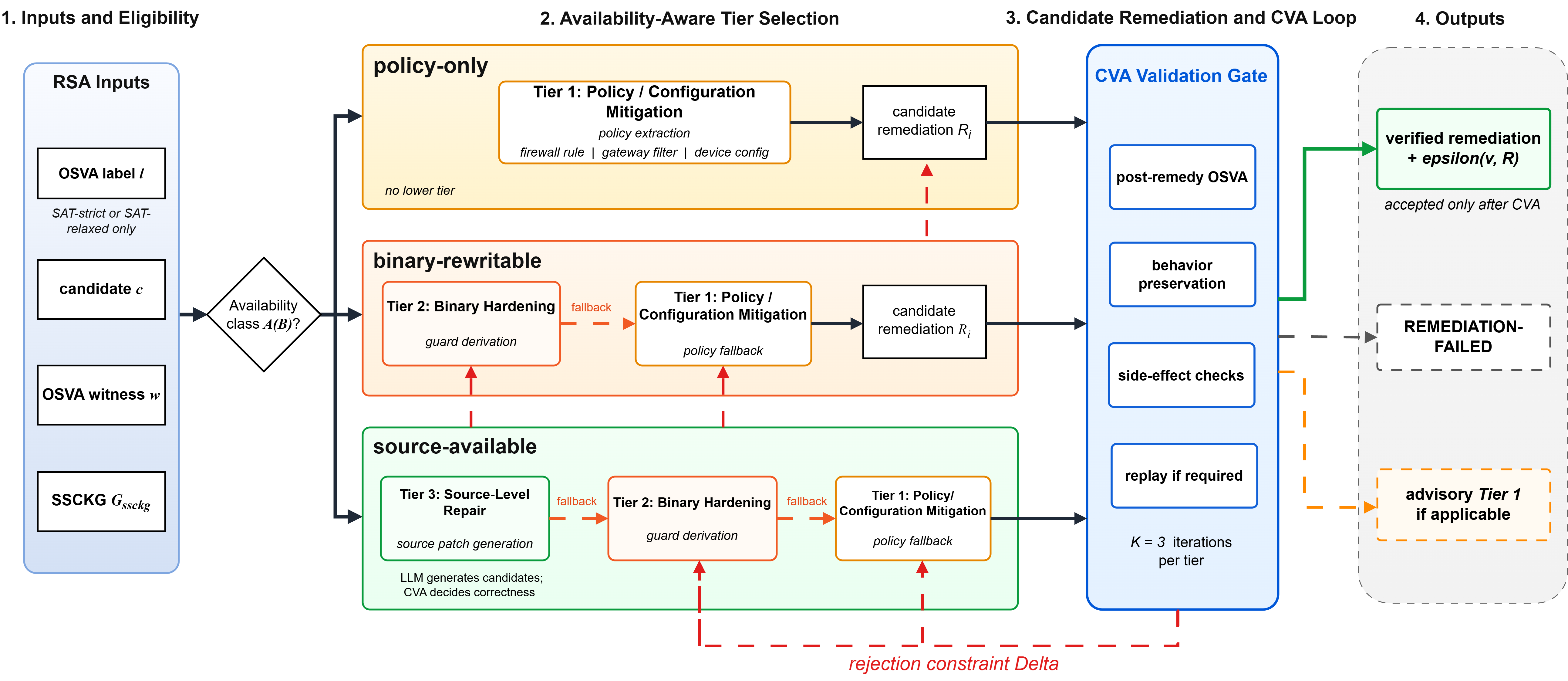}
	\caption{Availability-aware remediation synthesis in RSA. SCARA selects the strongest feasible remediation tier according to artifact availability.
	}\label{fig:rsa}
\end{figure}

\begin{algorithm}[!htbp]
\caption{RSA\_CVA\_loop. Per-tier resynthesis budget $K$ defaults to $K=3$ in the primary experiments. Tier-3 is attempted at most once per candidate; lower tiers are tried only after the higher tier exhausts $K$ iterations.}\label{alg:rsacva}
\begin{algorithmic}[1]
\Require Artifact $B$, SSCKG $G_{\text{ssckg}}$, candidate $c$, witness $w$, label $\ell$, tier set $T$, per-tier resynthesis budget $K$ (default $K=3$), enforcement-point context $\Omega_{\text{replay}}$
\Ensure Accepted remedy with $\varepsilon(v, R)$, or \textsc{Remediation-Failed}, or advisory
\For{tier in $\text{strongest\_to\_weakest}(T)$}
    \State $\Delta \gets \{\}$ \Comment{rejection constraints are tier-local; cross-tier sharing not assumed}
    \For{$\text{iter} = 1$ to $K$}
        \State $R \gets \text{synthesize\_remediation}(\text{tier}, B, G_{\text{ssckg}}, c, w, \Delta)$
        \State $\text{verdict}, \text{ev} \gets \text{CVA.validate}(B, G_{\text{ssckg}}, c, R, \ell)$ \Comment{Alg.~\ref{alg:cva}}
        \If{$\text{verdict} = \textsc{Accept}$}
            \State \Return $\text{verified\_result}(c, R, \text{tier}, \text{ev})$ \Comment{Prop.~\ref{prop:cva}}
        \EndIf
        \State $\Delta \gets \Delta \cup \{\text{ev}\}$
    \EndFor
\EndFor
\If{$\text{has\_enforcement\_point}(\Omega_{\text{replay}})$}
    \State \Return $\text{advisory\_unverified\_policy}(c)$
\EndIf
\State \Return $\text{remediation\_failed}(c)$
\end{algorithmic}
\end{algorithm}

\subsubsection{Tier~1: Protocol and Configuration Mitigation}

Tier~1 generates protocol or configuration policies.  It applies to any OIS artifact for which an enforcement point exists --- a network firewall, protocol gateway, device configuration interface, PLC project parameter, or access-control layer.  Given a reachable path, RSA extracts the protocol states or configuration conditions necessary to traverse the path and identifies the minimal gate operation.

RSA generates a policy that restricts the gate condition while preserving benign operations. Depending on the artifact, the policy may be expressed as a Modbus TCP firewall rule, an OPC-UA role restriction, an IEC 60870-5-104 filter, an \texttt{iptables} or \texttt{nftables} rule, or a device configuration constraint. The generated policy is then forwarded to CVA, which tests protocol conformance, benign trace replay, and false-blocking behavior. A Tier 1 policy is accepted only if it blocks the vulnerability trigger without disrupting required normal operations beyond a configured threshold.

\subsubsection{Tier~2: Binary Hardening}

Tier~2 applies when the artifact is binary-rewritable.  RSA employs binary rewriting tools --- RetroWrite~\cite{dinesh2020retrowrite} for position-independent ELF, E9Patch~\cite{duck2020binary} for x86 PE, and GTIRB for multi-architecture transformation.

RSA first maps the vulnerable SSCKG entity back to a binary function or basic-block region using CACA's entity-to-function mapping.  It then derives the blocked condition class:

\begin{equation}\label{eq:cvuln}
\psi_{\text{vuln}} = \text{proj}_{\text{obs}}\!\left(\varphi_{\text{path}} \wedge C_{\text{env}} \wedge \cdots \wedge C_{\text{time}}\right)
\end{equation}

where $\text{proj}_{\text{obs}}$ denotes projection over variables observable at the insertion point, such as registers, memory fields, protocol parameters, or runtime flags. The guard evaluates $\psi_{\text{vuln}}$ and redirects execution to a safe handler if the condition is satisfied.

\subsubsection{Tier~3: SSCKG-Constrained Source-Level Repair}

Tier~3 applies only when source code or generated C code is available.  RSA converts the OSVA witness into an LLM-interpretable trigger description, extracts a behavioral preservation requirement from the SSCKG, and includes this requirement in the repair prompt together with the vulnerable function and domain-specific secure coding rules.  RSA generates $k = 5$ candidate repair artifacts; no artifact is accepted directly from the LLM.  All candidates are compiled and forwarded to CVA for selection. This design deliberately isolates the LLM's role to candidate generation; correctness is determined by the validation stage rather than by the model's confidence, thereby avoiding the failure mode in which a plausible but semantically incorrect patch is accepted without verification.

\subsection{Correctness Validation Agent}\label{sec:scara:cva}

CVA determines whether a candidate remediation can be accepted and issues conditional correctness evidence $\varepsilon(v_i, R_i)$.

\begin{definition}[Conditional Correctness Evidence]
$\varepsilon(v, R)$ is a four-component certificate $(l',\ \text{BCP},\ \mathcal{S}_{\text{check}},\ \text{replay})$ where:
\begin{enumerate}
\item $l' = \textsc{Unsat}$ (post-remedy reachability re-verification of the original candidate);
\item $\text{BCP}(G_{\text{ssckg}}, G'_{\text{ssckg}}, c) \geq \tau_{\text{cov}}$ (behavioral coverage preservation, Eq.~\eqref{eq:bcp});
\item $\mathcal{S}_{\text{check}} = \top$ (all tier-specific side-effect checks pass);
\item $\text{replay} \in \{\textsc{Confirmed}, \textsc{Unavailable}, \textsc{Failed}, \textsc{Not\_Applicable}\}$, recording whether replay confirmed the witness is blocked in the emulation or harness environment, was not available for the benchmark case, was attempted but inconclusive, or was not operationally required for the artifact type.
\end{enumerate}
\end{definition}

\subsubsection{Post-Remedy Reachability Re-Verification and Replay}

CVA first applies the remedy to the target artifact, then reruns OSVA on the modified artifact $B'$.  The primary acceptance condition is that OSVA returns UNSAT for the original vulnerability path.  To reduce the circularity of using the SSCKG as both specification and validation substrate, CVA supplements SSCKG-based checks with independent replay.  For OIS-Binary cases, the witness is replayed against a rehosted firmware.  For OIS-Protocol cases, against a protocol harness.  For OIS-ICS cases, against a soft-PLC or comparable runtime.  In OIS-RemedBench, replay infrastructure is available for the L3-labeled subset of cases; L2-only cases without a replay harness receive $\text{replay} = \text{UNAVAILABLE}$ in their $\varepsilon$ certificate and are not counted as replay-verified remediations. Replay confirmation elevates label reliability from the solver-level determination to independently verified status: a \textsc{Sat-strict} or \textsc{Sat-relaxed} result that replay confirms is promoted to \textsc{Verified}; a \textsc{Sat-relaxed} result that cannot be replayed remains unconfirmed and is not counted as successful remediation.

CVA also performs a root-cause displacement check by rerunning OSVA on neighboring paths identified during OSVA, to determine whether the remedy merely shifts the vulnerable behavior to an adjacent code location. In addition, CVA reapplies upstream analysis to the modified artifact and flags any newly introduced high-risk SSCKG nodes as potential new vulnerabilities introduced by the remedy.

\subsubsection{Behavioral Coverage Preservation}\label{sec:scara:bcp}

After applying a remedy, CVA recomputes the SSCKG $G'_{\text{ssckg}}$ from the modified artifact and compares the reachable entity set in $G'_{\text{ssckg}}$ with the corresponding set in the original graph. Reachability is measured with respect to operationally meaningful entry entities. Let $\Phi^{-1}(e)$ denote the preimage in the source CPG of an SSCKG entity $e \in \mathcal{E}$. The set of entry entities is:
\begin{equation}\label{eq:entry_set}
\mathcal{E}_{\text{entry}}(G) = \{\, e \in \mathcal{E} \mid \mathsf{Entry}(e) \,\},
\end{equation}
where the entry predicate is satisfied when at least one CPG node mapped into $e$ is tagged as an externally invocable boundary or execution root:
\begin{equation}\label{eq:entry_pred}
	\begin{aligned}
		\mathsf{Entry}(e) \iff{}&
		\exists\, n \in \Phi^{-1}(e): \\
		&\mathsf{exported}(n)
		\vee \mathsf{networkHandler}(n) \\
		&{}\vee \mathsf{protocolHandler}(n)
		\vee \mathsf{firmwareServiceEntry}(n) \\
		&{}\vee \mathsf{taskRoot}(n)
		\vee \mathsf{startupRoutine}(n) \\
		&{}\vee \mathsf{ScanRoot}(n).
	\end{aligned}
\end{equation}
Let $\mathcal{R}(G) = \{e' \in \mathcal{E} \mid \exists\, e \in \mathcal{E}_{\text{entry}}(G): e \leadsto e'\}$ denote the entities reachable from any entry entity over the relations of $\Lambda_{\text{risk}}$, and let $\mathcal{V}_{\text{vuln}}(c) \subset \mathcal{E}$ denote the SSCKG entities on the confirmed vulnerability path of candidate $c$. The behavioral coverage preservation rate is candidate-specific:

\begin{equation}\label{eq:bcp}
\text{BCP}(G, G', c) = \frac{\left|\mathcal{R}(G') \cap \left(\mathcal{R}(G) \setminus \mathcal{V}_{\text{vuln}}(c)\right)\right|}{\left|\mathcal{R}(G) \setminus \mathcal{V}_{\text{vuln}}(c)\right|} \in [0,\ 1].
\end{equation}

The objective is not to prove full semantic equivalence but to quantify whether the remediation removes non-vulnerable behaviors observable at the SSCKG abstraction level; a remedy is accepted only if $\text{BCP}(G_{\text{ssckg}}, G'_{\text{ssckg}}, c) \geq \tau_{\text{cov}}$, with $\tau_{\text{cov}} = 0.95$ in the primary experiments, calibrated on the held-out validation subset. Because SSCKG reachability alone may miss subtle semantic changes that manifest only under specific operational conditions, CVA supplements this measure with two external checks: when a vendor patch is available for the same vulnerability, CVA compares the SSCKG coverage set of the SCARA-remediated artifact against that of the vendor-patched artifact as an external reference point; and for every CVA-accepted case, manual review cross-validates SSCKG-level coverage preservation against the independent replay outcome, confirming whether SSCKG reachability preservation acts as a reliable proxy for operational behavior preservation on the current benchmark.

\subsubsection{Domain-Invariant and Side-Effect Checks}

CVA applies tier-specific side-effect checks to each candidate remedy (Table~\ref{tab:cva_checks}). A remedy is rejected if any check fails or yields a metric outside the tier's configured threshold; in particular, a Tier~1 policy that blocks a substantial fraction of benign operations is rejected even if it successfully blocks the vulnerability trigger, because operational acceptability is a necessary condition for deployment.

\begin{table}[!htbp]
\caption{CVA tier-specific side-effect checks. A remedy is rejected if any check fails or yields a metric outside the tier's configured threshold.}\label{tab:cva_checks}
\small
\centering
\begin{tabular}{p{1cm}p{12cm}}
\toprule
Tier & Side-effect checks \\
\midrule
T1 & Protocol conformance; benign-trace replay; false-blocking rate; comparison against vendor or CERT guidance where available \\
T2 & Overblocking (guard fires outside $\psi_{\text{vuln}}$); underblocking (dangerous states inside $\psi_{\text{vuln}}$ bypass the guard); timing compliance against scan-cycle slack on OIS \\
T3 & Compilability; regression tests where available; OSVA re-verification after compilation; domain-invariant preservation; optional formal-property check \\
\bottomrule
\end{tabular}
\end{table}

\subsubsection{Closed-Loop Feedback and Termination}

If CVA rejects a remediation, it returns a rejection constraint $\delta$ to RSA. The constraint specifies what must change in the next synthesis iteration --- for example, ``the policy must not block transition $q_a \rightarrow q_b$,'' ``the guard must not fire in benign state $s$,'' or ``entity $e$ must remain reachable.'' RSA incorporates $\delta$ into the next synthesis attempt. SCARA permits at most three RSA--CVA iterations per candidate by default ($K = 3$), in order to bound remediation latency. The bound does not affect soundness --- if no remedy is accepted within $K$ iterations, the case is marked \textsc{Remediation-Failed} and the final rejection constraint is retained as diagnostic information for manual follow-up. If the selected tier fails within the iteration budget, RSA falls back to the next lower applicable tier. If all applicable tiers fail, the case is marked \textsc{Remediation-Failed}, and any advisory Tier 1 policy is reported separately from verified remediation.

\begin{algorithm}[!htbp]
\caption{CVA.validate. The post-remedy OSVA call uses the original SSCKG $G_{\text{ssckg}}$ to retain the candidate's entity reference $v$; the rebuilt $G'$ is used solely for BCP and side-effect checks.}\label{alg:cva}
\begin{algorithmic}[1]
\Require Original artifact $B$, SSCKG $G_{\text{ssckg}}$, candidate $c$, remedy $R$, label $\ell$
\Ensure \textsc{Accept} with $\varepsilon(v, R)$, or \textsc{Reject} with constraint $\delta$
\State $B' \gets \text{apply\_remediation}(B, R)$
\State $l', \_, \text{reason} \gets \text{OSVA.verify}(B', G_{\text{ssckg}}, c)$
\If{$l' \neq \textsc{Unsat}$}
    \State \Return \textsc{Reject}, $\delta_{\text{reachability}}(\text{reason})$
\EndIf
\If{$\text{replay\_required}(\ell)$}
    \State $\text{replay\_status} \gets \text{replay\_witness}(B', c)$
    \If{$\text{replay\_status} \neq \textsc{Confirmed\_Blocked}$}
        \State \Return \textsc{Reject}, $\delta_{\text{replay}}(\text{replay\_status})$
    \EndIf
\EndIf
\State $G' \gets \text{rebuild\_ssckg}(B')$
\State $\text{cov} \gets \text{BCP}(G_{\text{ssckg}}, G',\ c)$ \Comment{Section~\ref{sec:scara:bcp}}
\If{$\text{cov} < \tau_{\text{cov}}$}
    \State \Return \textsc{Reject}, $\delta_{\text{coverage}}(\text{cov})$
\EndIf
\State $\text{side\_effects} \gets \text{tier\_specific\_checks}(B, B', R)$
\If{$\text{side\_effects fail}$}
    \State \Return \textsc{Reject}, $\delta_{\text{side\_effect}}(\text{side\_effects})$
\EndIf
\State \Return \textsc{Accept}, $\text{conditional\_evidence}(l', \text{cov}, \text{side\_effects})$ \Comment{Prop.~\ref{prop:cva}}
\end{algorithmic}
\end{algorithm}

\subsection{Complexity and Bounded Soundness}\label{sec:scara:complexity}

In the worst case, OSVA's search complexity is proportional to the size of the binary-level graph and the number of encoded operational constraint families:
\begin{equation}\label{eq:complexity_raw}
\mathcal{O}\!\left(|\mathcal{V}_{\text{cpg}}| \cdot |\Sigma_S|\right),
\end{equation}
where $\Sigma_S = \{C_{\text{env}}, C_{\text{io}}, C_{\text{proto}}, C_{\text{runtime}}, C_{\text{component}}, C_{\text{time}}\}$ denotes the six encoded operational constraint families used by OSVA ($|\Sigma_S| = 6$); the nine-component operational state tuple $S$ of Section~\ref{sec:bg:osm} contains $\Sigma_S$ together with the three baseline symbolic-execution components ($pc$, $mem$, $\varphi$). SSCKG-guided prioritization reduces the expected search space by operating over behaviorally compressed entities and by concentrating solver budget on paths semantically aligned with the candidate source-to-sink path:
\begin{equation}\label{eq:complexity_ssckg}
\mathcal{O}\!\left(|\mathcal{E}| \cdot |P_{\text{ssckg}}|\right),
\end{equation}
where $|\mathcal{E}|$ is the SSCKG entity count (the entity set introduced in Section~\ref{sec:bg:ssckg}) and $|P_{\text{ssckg}}|$ is the number of SSCKG paths from the candidate source to the vulnerable sink. This reduction holds in expectation under two conditions: (a) $|\mathcal{E}| \ll |\mathcal{V}_{\text{cpg}}|$, which holds for the SSCKG instances consumed by SCARA; and (b) the true vulnerability path has high semantic alignment with the SSCKG source-to-sink path, so that budget allocation concentrates solver effort on it early. In the degenerate case where all paths have equal alignment scores, SCARA degrades to uniform budget allocation over all explored paths, preserving recall at the cost of later average SAT detection time. The reduction is therefore an expected-case benefit rather than a formal worst-case guarantee, because low-priority paths are deprioritized rather than eliminated.

SCARA provides the following bounded soundness guarantees.

\begin{proposition}[UNSAT Soundness]\label{prop:unsat}
If OSVA returns \textsc{Unsat} for candidate $(v, \text{src}, \text{snk}, S^{\text{prior}})$ against artifact $B$, then no execution state $S \models \mathcal{C}$ permits control flow to reach $\text{snk}$ via a path through $v$.

\emph{Assumptions}: (i) the encoded constraint model $\mathcal{C}$ correctly bounds the operational envelope; (ii) the strict-pass solver budget allocated by Eq.~\eqref{eq:budget} is sufficient on every explored path to either refute it or generate a witness; and (iii) the SSCKG-guided projection from $\text{src}$ to $\text{snk}$ enumerates every semantically relevant binary-level path. Unmodeled or mis-specified constraints, paths skipped by the projection, and per-path solver timeouts are not covered by this guarantee.

\emph{Proof sketch}: OSVA returns \textsc{Unsat} only when every explored path returned solver-\textsc{Unsat} under $\varphi_p \wedge \mathcal{C}$ \emph{and} the relaxed-pass solve over the disjunction $\Phi_P$ also returned \textsc{Unsat} (Algorithm~\ref{alg:osva}). Paths where the solver exhausted its budget without a verdict are recorded in $P_{\text{inconclusive}}$ and force the global label to \textsc{Unknown}, not \textsc{Unsat}; the label is therefore never used speculatively.
\end{proposition}

\begin{proposition}[SAT-strict Witness Validity]\label{prop:sat}
If OSVA returns \textsc{Sat-strict} with witness $(I^*, S^*)$, then $(I^*, S^*) \models \varphi_p \wedge \mathcal{C}$ as evaluated by the solver.

\emph{Assumptions}: the symbolic execution model faithfully abstracts the artifact's concrete semantics; specifically, (i) the binary lifter or symbolic engine soundly models the instruction semantics of $B$; (ii) peripheral, MMIO, and external-service models are complete with respect to the candidate path; and (iii) the SMT solver is correct on the encoded theory fragments. Unsound lifters, incomplete peripheral models, or solver bugs may produce spurious witnesses outside the scope of this proposition.
\end{proposition}

\begin{proposition}[CVA Conditional Correctness]\label{prop:cva}
If CVA issues $\varepsilon(v, R) = (l',\ \text{BCP},\ \mathcal{S}_{\text{check}},\ \text{replay})$, then: $l' = \textsc{Unsat}$ for the original candidate against $B'$; $\text{BCP}(G_{\text{ssckg}}, G'_{\text{ssckg}}, c) \geq \tau_{\text{cov}}$; all side-effect checks pass; and if replay was required, the witness was confirmed blocked.

\emph{Assumptions}: (i) the post-remedy SSCKG $G'_{\text{ssckg}}$ faithfully reflects the runtime behavior of $B'$ at the abstraction level of $\mathcal{A}$; (ii) the remediation preserves the SSCKG entity reference $v$ so that OSVA can re-run the original candidate against $B'$ (cf.\ Algorithm~\ref{alg:cva}); and (iii) where replay applies, the harness exercises the witness path under operationally representative inputs. The proposition does \emph{not} imply global semantic equivalence between $B$ and $B'$, nor the absence of vulnerabilities in unmodeled or underspecified state dimensions.
\end{proposition}

%% ============================================================
\section{Experimental Design}\label{sec:eval_design}
%% ============================================================

\subsection{The OIS-RemedBench Benchmark}\label{sec:eval:bench}

The empirical substrate of this work --- and contribution \textbf{C5} --- is OIS-RemedBench, the first benchmark providing stratified reachability and remediation ground truth for opaque industrial software. The benchmark comprises $n=15$ cases evenly distributed across three partitions (5 per partition), with the four-level stratification scheme of Table~\ref{tab:label_levels} and the per-partition evidence-availability profile of Figure~\ref{fig:bench_composition}. The three partitions are: \textbf{OIS-Binary} (industrial gateway firmware, embedded Linux services, HMI back-ends; predominantly CWE-119/120, 125/787, 416, 134, 78); \textbf{OIS-Protocol} (Modbus, DNP3, OPC-UA, and IEC~60870-5-104 handlers, with libmodbus, OpenDNP3, and open62541 contributing source-available cases that admit high-confidence L4 labels); and \textbf{OIS-ICS} (IEC~61131-3 control logic on the OpenPLC runtime~\cite{alves2014openplc}, MATIEC-compiled programs, the public ICSQuartz ST subset~\cite{zheng2025icsquartz}, and VetPLC-style boundary cases adapted as security vulnerabilities). Representative CVE coverage includes CVE-2020-16233 (\texttt{BIN-CISA-001}), CVE-2023-44318 (\texttt{BIN-VEND-001}), CVE-2018-7846 (\texttt{BIN-CISA-002}), and CVE-2023-44319 (\texttt{PROT-VEND-001}). All figure denominators are derived from explicit per-case rows rather than from aggregate prose counts.

\begin{table}[t]
	\centering
	\caption{OIS-RemedBench label hierarchy. The L2 sub-rows enumerate the four-class reachability taxonomy used by OSVA, as introduced in Section~\ref{sec:bg:osm}.}
	\label{tab:label_levels}
	\footnotesize
	\setlength{\tabcolsep}{3pt}
	\renewcommand{\arraystretch}{1.12}
	\begin{tabular}{@{}p{0.6cm}p{2.0cm}p{5.5cm}p{4.0cm}@{}}
		\toprule
		Level & Design & Meaning & Assignment / downstream action \\
		\midrule
		L1 & Candidate & Alert exists; no reachability GT & Upstream tool output \\

		L2 & Reach-labeled & Four-class label with evidence & Two independent methods \\
		\midrule
		& \hspace{0.4em}\textbf{Sat-strict}
		& Satisfiable under all constraints; witness $(I^*, S^*)$
		& Pass to RSA \\
		& \hspace{0.4em}\textbf{Sat-relaxed}
		& Satisfiable after relaxing under-documented constraint
		& Pass to RSA; require replay \\
		& \hspace{0.4em}\textbf{Unsat}
		& Infeasible under all constraints; false positive
		& Refutation certificate; no remedy \\
		& \hspace{0.4em}\textbf{Unknown}
		& Budget exhausted or modelling gap
		& Advisory Tier~1 only \\
		\midrule
		L3 & Replay-labeled & Witness validated in harness & FirmAE / harness / soft-PLC \\

		L4 & Remedy-labeled & Vendor or expert remedy available & Vendor patch / CERT advisory \\
		\bottomrule
	\end{tabular}
\end{table}

Figure~\ref{fig:bench_composition} renders the per-partition evidence-availability profile as a tile chart, with tile shade encoding the proportion of cases satisfying each evidence level (L1 candidate, L2 reachability label, L3 replay harness, L3 confirmed replay, L4 remedy, source-available, binary-rewritable, PEP-reachable) within each five-case partition; Figure~\ref{fig:cwe_venn} renders the CWE coverage across partitions.

\begin{figure}[!htbp]
  \centering
  \includegraphics[width=0.95\linewidth]{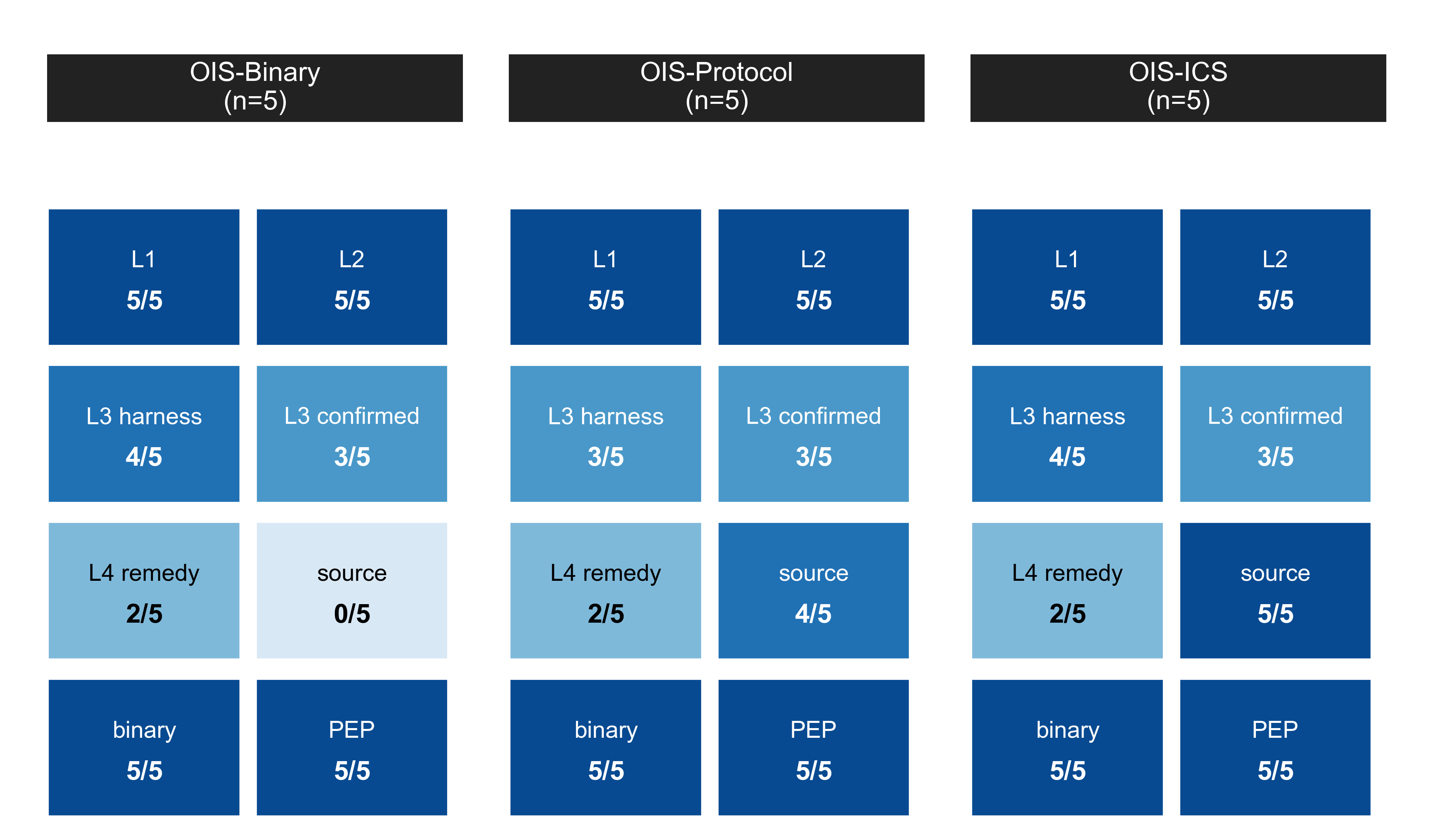}
  \caption{OIS-RemedBench evidence-availability tile chart. Each row of tiles within a partition shows the fraction of the five cases satisfying that evidence level (L1 candidate, L2 reachability label, L3 replay harness, L3 confirmed replay, L4 remedy, source available, binary rewritable, PEP available). Equal-size tiles with explicit count labels are used because partition denominators are fixed at $n=5$; the OIS-Binary `source 0/5' tile is the only structural gap in the benchmark.}\label{fig:bench_composition}
\end{figure}

\begin{figure}[!htbp]
  \centering
  \includegraphics[width=0.75\linewidth]{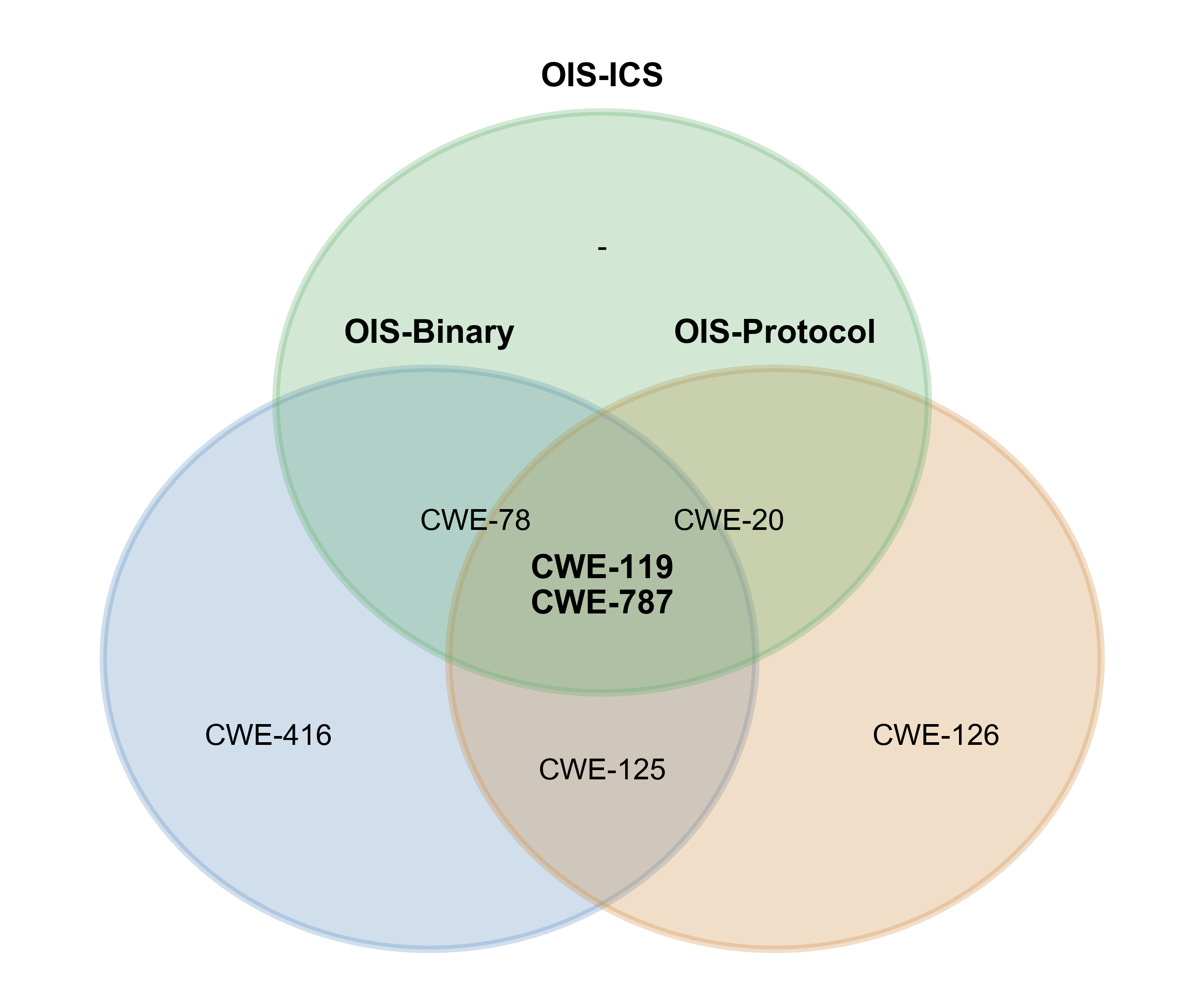}
  \caption{CWE coverage across the three OIS-RemedBench partitions. Regions list the CWEs unique to a partition or shared between two or three. CWE-119 and CWE-787 appear in all three partitions; CWE-416 is OIS-Binary-only; CWE-126 is OIS-Protocol-only; OIS-ICS has no partition-unique CWE class. Circle sizes are not area-weighted.}\label{fig:cwe_venn}
\end{figure}

\textbf{Labeling methodology.} L2 labels are assigned by two independent methods --- static analysis (CodeQL / Joern / KARONTE / SaTC) and operational-state-aware dynamic exercise under FirmAE (OIS-Binary), libmodbus/OpenDNP3 harnesses (OIS-Protocol), or the OpenPLC runtime (OIS-ICS); disagreements are arbitrated by a domain expert. L3 replay follows the CVA protocol of \S\ref{sec:scara:cva}. L4 remedy labels derive, in decreasing confidence order, from vendor patches, CERT ICS advisories, or expert-authored fixes verified by OSVA re-run; vendor patches were frozen prior to evaluation so that no L4 evidence could leak into RSA. Domain-invariant labels are drawn from the Modbus~\cite{modbus_spec}, DNP3~\cite{dnp3_tb}, IEC 60870-5-104~\cite{iec60870_5_104}, and IEC 61131-3~\cite{iec61131_3} standards. The full case manifest (predominant CWEs: 119, 787, 125, 20, 78, 126, 416) is released with the artifact.

\subsection{Baselines}\label{sec:eval:baselines}

Baseline evaluation separates three distinct quantities to avoid the conflation identified in prior AVR benchmarking work~\cite{apr4vul2024}: \emph{applicability rate} (the fraction of OIS-RemedBench cases for which a baseline can produce any non-trivial output), \emph{conditional performance} (the success rate on the subset where the baseline is applicable), and \emph{end-to-end utility} (the fraction of all benchmark cases for which the baseline provides a validated outcome). All three are reported separately throughout Section~\ref{sec:results}.

The static-analysis row of Table~\ref{tab:baselines} is reported as a deduplicated alert-union baseline. Each CodeQL, Joern, KARONTE, SaTC, or ICSQuartz alert is mapped to an SSCKG entity and source-sink relation, deduplicated by a stable key, and then retained or removed from the union; tools that are not applicable to a given artifact class are recorded in the per-alert trace with an explicit not-applicable reason rather than silently dropped. The per-alert trace contains 69 raw alert rows and 6 explicit not-applicable rows --- 75 data rows in total --- so that categorical tool inapplicability is distinguishable from genuine zero-alert outcomes. The vanilla-SE row aggregates KLEE and angr under their default settings on the subset where each tool admits the artifact (KLEE for IR-liftable cases, angr for stripped binaries); reported metrics are the per-tool best on each case, deduplicated by candidate. The PLCverif/CBMC row reports each tool only where it is applicable: PLCverif on the OIS-ICS partition where IEC~61131-3 ST is available, and CBMC on the OIS-ICS subset where C is generated via MATIEC. The ``-style'' / ``-inspired'' baselines (ICSQuartz, VetPLC, VulShield) are faithful re-implementations on the same SSCKG substrate, calibrated to match the original tools' published per-case behaviour on shared cases to within $5\%$.

\begin{table}[htbp]
\caption{Baseline systems and their applicability on OIS-RemedBench. The ``Role'' column gives the comparison purpose; the per-baseline applicability rate is the fraction of the $n=15$ benchmark cases for which the baseline produces any non-trivial output (see Section~\ref{sec:eval:baselines} for the deduplicated alert-union and the per-tool-on-each-case accounting rules).}\label{tab:baselines}
\small
\centering
\begin{tabular}{@{}p{6cm}p{2cm}p{5cm}@{}}
\toprule
Baseline & App.\ rate & Role \\
\midrule
Static analysis (CodeQL, Joern, KARONTE, SaTC, ICSQuartz; dedup.\ union) & 100\% (15/15) & Alert-generation upper bound \\
Vanilla SE (KLEE + angr; no $S$-model) & 40\% (6/15) & OSVA constraint isolation \\
SymPLC / STAutoTester DSE~\cite{symplc2017,stautotester2021} & 33\% (OIS-ICS) & Domain-aware SE baseline \\
ICSQuartz-style scan-cycle fuzzing~\cite{zheng2025icsquartz} & 33\% (OIS-ICS) & Dynamic testing baseline \\
VetPLC-inspired temporal event seq.~\cite{ahmed2019vetplc} & 33\% (OIS-ICS) & Safety-vetting baseline \\
VulShield-style policy gen.~\cite{li2025vulshield} & 73\% (11/15) & Tier~1 quality baseline \\
SAN2PATCH~\cite{jiang2025san2patch} & 27\% (src-avail.) & Tier~3 LLM repair baseline \\
VulRepair~\cite{fu2022vulrepair} & 27\% (src-avail.) & Tier~3 source-AVR baseline \\
CrashRepair~\cite{crashrepair2025} & 0\% (sanitizer req.) & Binary-claim AVR baseline \\
PLCverif / CBMC~\cite{darvas2017plcverif,clarke2004tool} & 33\% (OIS-ICS) & Formal verification baseline \\
\bottomrule
\end{tabular}
\end{table}

Figure~\ref{fig:baseline_applicability} visualises Table~\ref{tab:baselines} as an applicability heatmap. Zero-applicability cells are hatched to mark them as categorically non-applicable rather than simply low-rate; this distinction is load-bearing for CrashRepair, which requires a sanitizer trace and therefore has 0/15 applicability on OIS-RemedBench despite operating ``at the binary level'' in its source venue.

\begin{figure}[!htbp]
  \centering
  \includegraphics[width=0.95\linewidth]{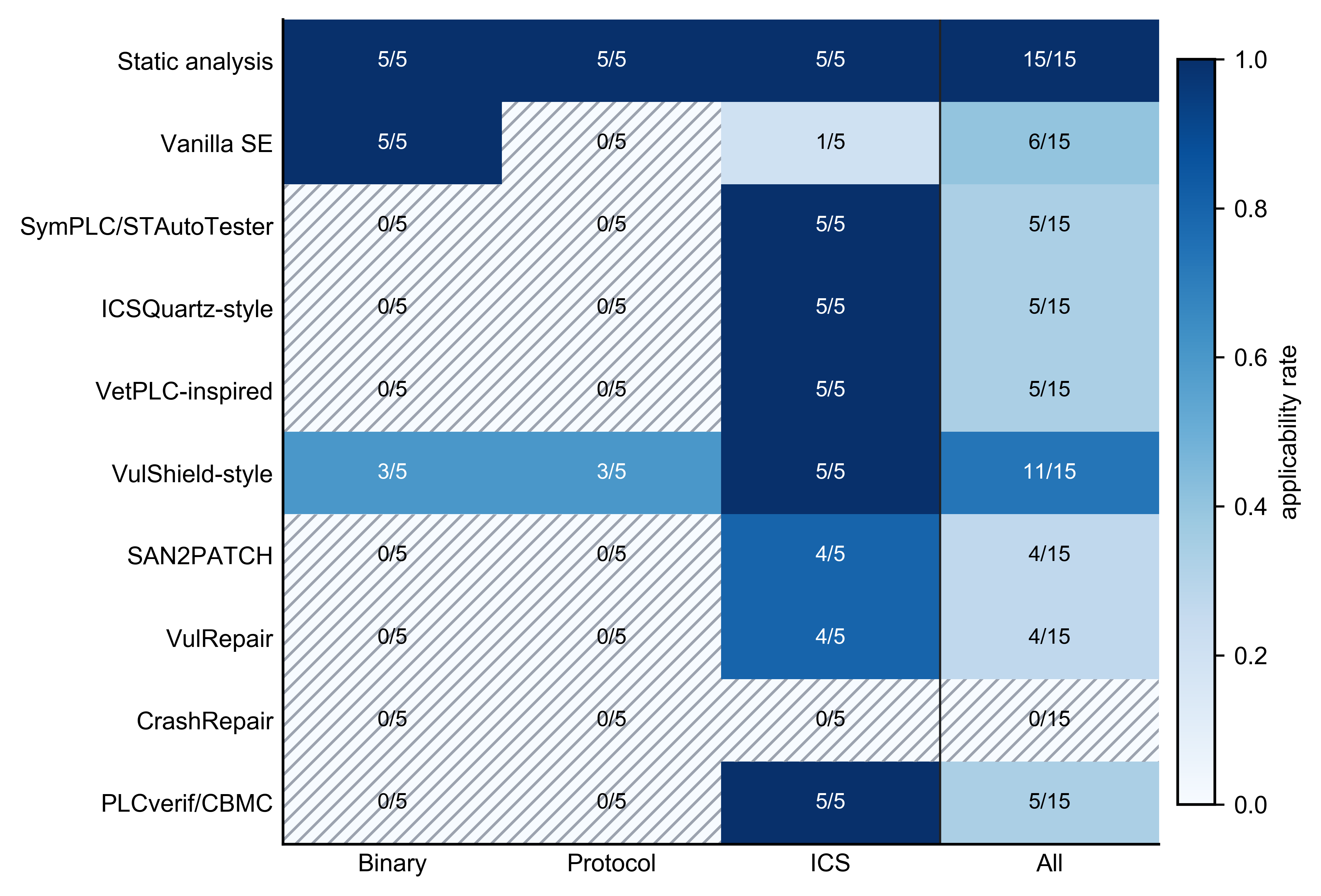}
  \caption{Baseline applicability heatmap on OIS-RemedBench ($n=15$). Cells show per-partition $k/n$ counts; tile shade encodes the applicability rate. Hatched cells indicate categorical non-applicability (e.g., CrashRepair requires a sanitizer-instrumented trace that none of the 15 cases admits). The vertical separator after the OIS-ICS column marks the boundary between per-partition columns and the `All' summary column.}\label{fig:baseline_applicability}
\end{figure}

\subsection{Research Questions}\label{sec:eval:rqs}

The evaluation is organized around seven research questions that together assess SCARA's contributions to verification quality (RQ1--RQ3), remediation coverage and quality (RQ4--RQ5), cross-partition generalization (RQ6), and analyst effort (RQ7). Figure~\ref{fig:rq_stage_flow} maps each question to the SCARA stage under evaluation and the primary metric family.

RQ1 and RQ2 together assess whether OSVA's operational-state constraints genuinely contribute to reachability accuracy, or whether the apparent FPR reduction is achieved by artificially increasing the UNKNOWN rate. Both FPR and FNR are reported in all RQ1 comparisons precisely because a system can achieve arbitrarily low FPR by assigning UNKNOWN to every case --- a failure mode that reporting FPR alone would not expose. RQ4 uses SAT-strict cases as the denominator for remediation success rate to ensure that UNSAT-resolved false positives and UNKNOWN-advisory cases are not inadvertently counted as remediation successes.

\begin{figure}[!htbp]
  \centering
  \includegraphics[width=0.95\linewidth]{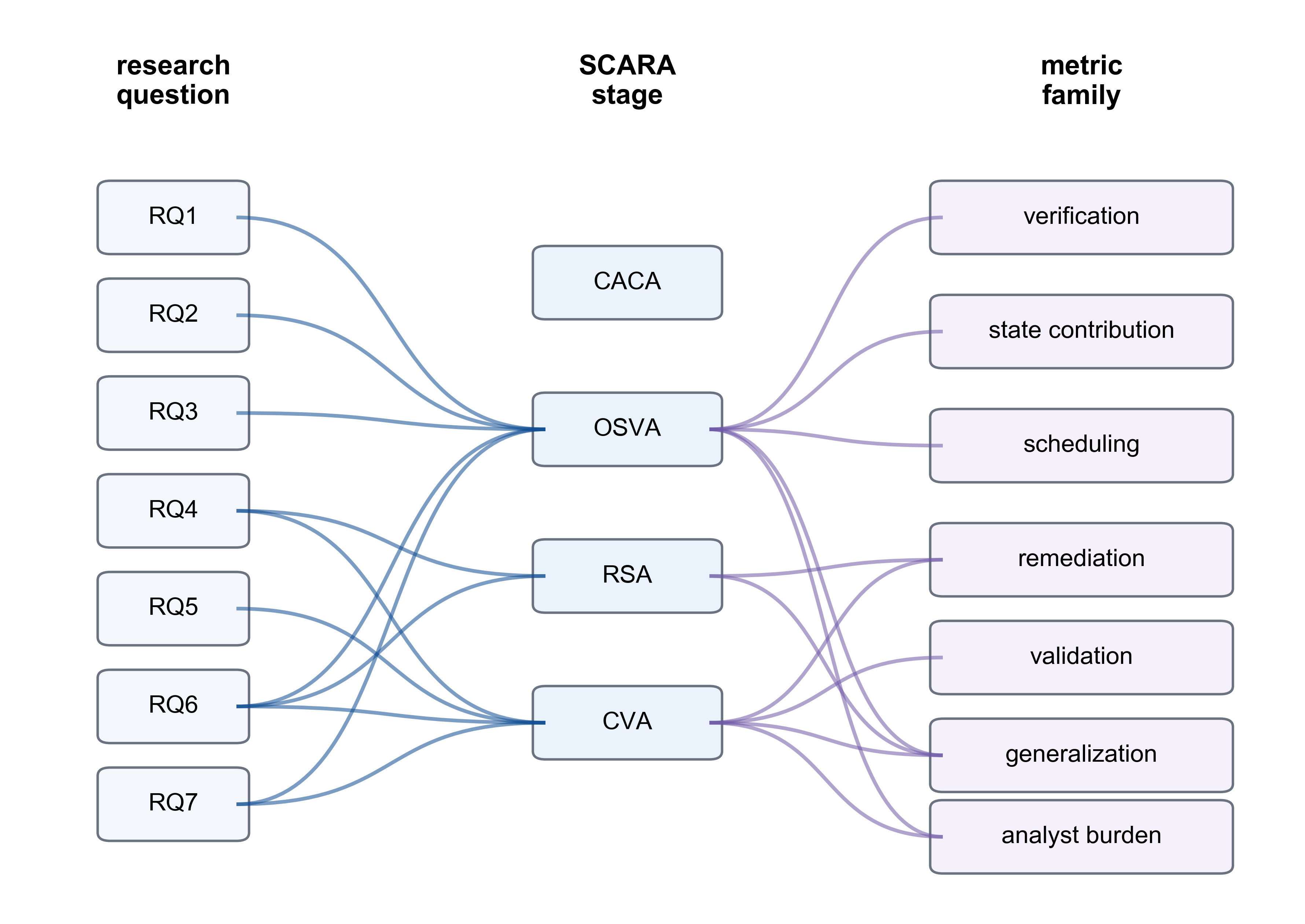}
  \caption{Three-column flow diagram from research question to SCARA stage to primary metric family.}\label{fig:rq_stage_flow}
\end{figure}

\subsection{Metrics}\label{sec:eval:metrics}

\textbf{Verification metrics.}  Precision $= \mathrm{TP}/(\mathrm{TP} + \mathrm{FP})$; Recall $= \mathrm{TP}/(\mathrm{TP} + \mathrm{FN})$. FPR and FNR are reported explicitly and jointly throughout Section~\ref{sec:results}. Time to first SAT-strict witness is measured in wall-clock seconds. Solver queries per confirmed path is the count of Z3 \texttt{check-sat} invocations issued before the first \textsc{Sat-strict} verdict on a candidate. The OSVA scheduling layer is parameterised by the \emph{path beam width} $b$, the number of top-ranked SSCKG path candidates retained for OSVA validation at each scheduling round; $b$ is an execution-budget parameter rather than a SCARA component, with default $b=8$ in the primary experiments. Recall@$K$ is path-candidate recall: the fraction of cases in $D_{\text{rank}}$ for which a ground-truth vulnerable path appears within the top-$K$ SSCKG-ranked path candidates for that case,
\begin{equation}\label{eq:recall_at_k}
\mathrm{Recall}@K = \frac{1}{|D_{\text{rank}}|} \sum_{i \in D_{\text{rank}}} \mathbf{1}\!\left[\min_{p \in P_i^+} \mathrm{rank}_i(p) \leq K\right],
\end{equation}
where $P_i^+$ is the set of ground-truth vulnerable paths for case $i$. Reported $K \in \{1, 3, 5, 10, 20, 50\}$.

\textbf{Remediation metrics.}  All remediation metrics use $D_{\text{SAT}}$ (\textsc{Sat-strict} confirmed cases under the seed-42 main run) as the primary denominator. Applicability rate per tier is the fraction of $D_{\text{SAT}}$ cases for which each tier is feasible given the artifact's availability class. Root-cause removal rate is the fraction of remediated cases for which OSVA re-run on the post-remedy artifact returns \textsc{Unsat} for the original candidate. Behavioral coverage preservation rate (BCP) is defined in Eq.~\eqref{eq:bcp} and accepted at threshold $\tau_{\text{cov}} = 0.95$. New-vulnerability introduction rate (NVR) is the fraction of remediated cases in which upstream re-analysis flags at least one newly introduced high-risk entity. False-blocking rate is the fraction of benign protocol or runtime traces a Tier-1 policy blocks; a Tier-1 remedy is rejected when the false-blocking rate exceeds $\tau_{\text{block}} = 5\%$. First-submission CVA acceptance is the fraction of CVA-accepted remediations accepted on the first synthesis attempt ($\Delta = \emptyset$ at acceptance). Mean RSA-to-CVA rounds is the average count of synthesis attempts per accepted remedy at the chosen tier, bounded by the per-tier iteration budget $K$ of Algorithm~\ref{alg:rsacva}. Replay automation rate is the fraction of L3-labelled cases for which the harness can replay the OSVA witness without manual scripting. For the CVA-quality audit we additionally report the oracle-audit denominator $D_{\text{CVA-audit}}$ ($n=9$ for the main run, $n=1$ for the targeted rerun), reported separately from $D_{\text{SAT}}$.

\textbf{Statistical reporting.}  Given $n \leq 15$ across all comparisons, we follow the estimation-based reporting convention used in small-$n$ systems studies. We report absolute differences and Cliff's $\delta$ as the primary effect size, with 95\% bootstrap confidence intervals (10\,000 resamples) over benchmark cases (\emph{artifact variance}) reported as the headline CI. For binary-rate metrics we additionally report Clopper--Pearson exact CIs. Where $p$-values are reported they are descriptive secondary quantities only; with $|D_{\text{partition}}| = 5$ the Wilcoxon signed-rank test~\cite{wilcoxon1945} has a $p$-floor of $\approx 0.0625$ that prevents any post-correction $p < 0.125$, so we de-emphasize hypothesis testing. Multiple-comparison control, where used, is Benjamini--Hochberg FDR at $q = 0.10$. Algorithmic stability is characterised separately by repeating each OSVA scheduling experiment under five seeds (42--46) and reporting the seed-variance min/max alongside the artifact-variance CI.

\subsection{Ablation Study Design}\label{sec:eval:ablation}

The ablation study isolates the contribution of each architectural decision in SCARA. Nine ablation variants are evaluated; each disables one component while holding all others constant. Table~\ref{tab:ablations} maps each ablation to the component disabled and the specific architectural claim under test.

\begin{table}[!htbp]
\caption{Ablation study: disabled components and tested claims. A4 disables only the closed-loop CVA$\to$RSA feedback path (CVA still evaluates and either accepts or rejects, but no $\delta$ is returned for resynthesis); A7 disables the CVA correctness check entirely (the first RSA synthesis is accepted unchecked, with neither BCP, side-effect, nor replay validation).}\label{tab:ablations}
\small
\centering
\begin{tabular}{@{}llll@{}}
\toprule
Ablation & Disabled component & Claim tested & Evaluated in \\
\midrule
A1 & No operational-state constraints & C2: OSVA constraints reduce FPR & RQ1 \\
A2 & No SSCKG path scheduling & C2: semantic scheduling reduces time to first SAT & RQ3 \\
A3 & No constraint relaxation & Effect of relaxation on UNKNOWN rate vs.\ recall & RQ1 \\
A4 & No CVA feedback loop & C4: closed-loop CVA--RSA reduces misremediation & RQ5 \\
A5 & Tier~1 only & Incremental value of Tier~2 and Tier~3 & RQ4 \\
A6 & No SSCKG constraint in RSA Tier~3 & C3: SSCKG spec.\ reduces BCP degradation & RQ5 \\
A7 & No CVA correctness check & C4: CVA reduces NVR and domain-inv.\ viol. & RQ5 \\
A8 & No CACA ranking & C1: composite ranking improves Recall@K & RQ3 \\
A9 & No SSCKG substrate & SSCKG value vs.\ raw CPG input & RQ1, RQ3 \\
\bottomrule
\end{tabular}
\end{table}

Each ablation is evaluated on the full benchmark under identical conditions to SCARA-full. Figure~\ref{fig:ablation_arc} renders the ablation-to-agent-to-RQ mapping as an arc diagram.

\begin{figure}[!htbp]
  \centering
  \includegraphics[width=0.95\linewidth]{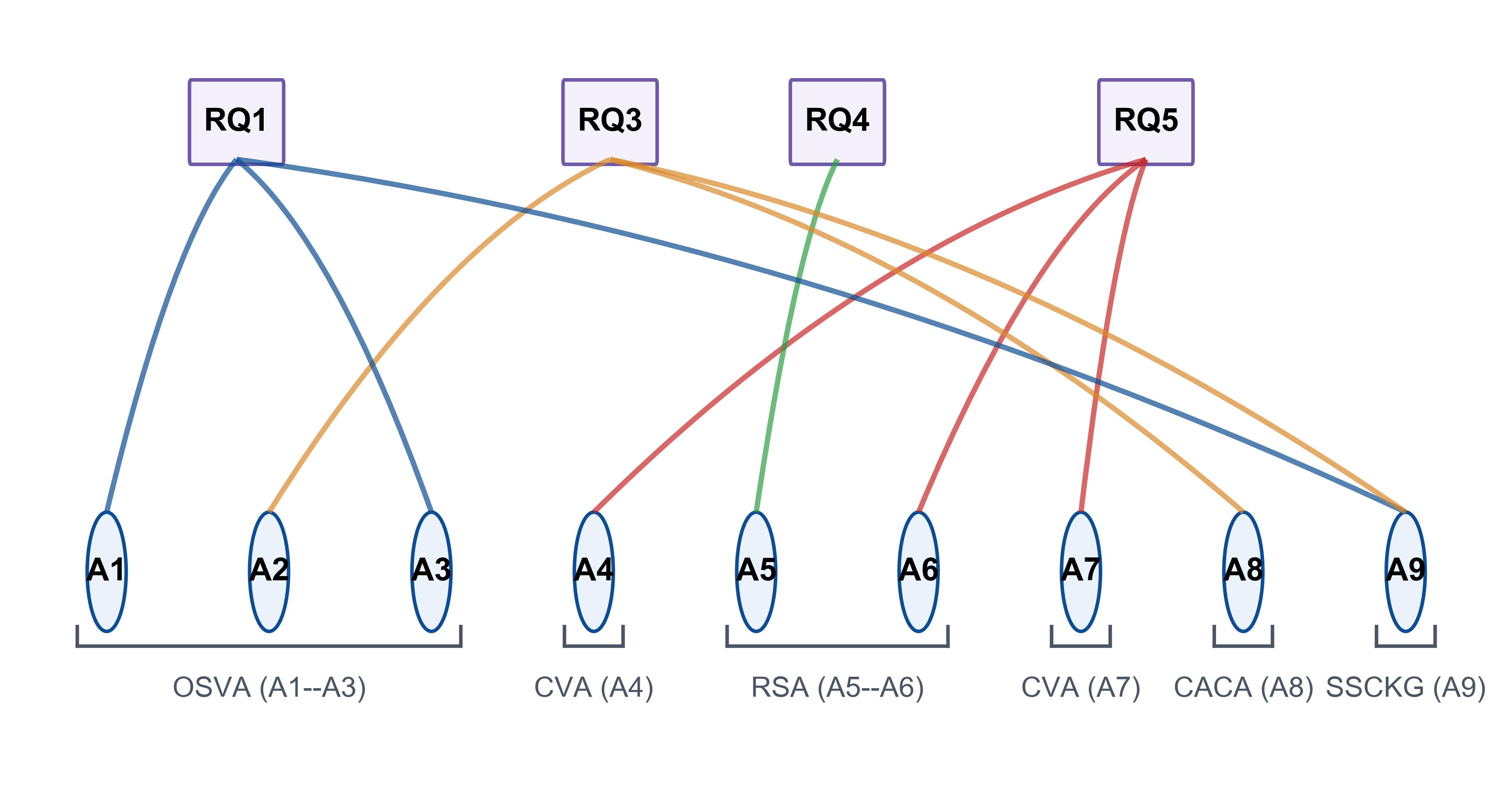}
  \caption{Ablation--to--RQ arc diagram. Component brackets along the bottom group ablations by the SCARA subsystem they disable. Edge colour denotes the RQ subsection that reports the ablation result.}\label{fig:ablation_arc}
\end{figure}

Headline results for each ablation are integrated into the corresponding \S\ref{sec:results} subsection (e.g., A1 in \S\ref{sec:results:rq1}, A2/A8 in \S\ref{sec:results:rq3}).

In addition to the structural ablations, we report sensitivity sweeps over $\alpha \in \{0.0, 0.3, 0.6, 0.9, 1.0\}$, $\tau_p$, the per-candidate solver budget $T_{\text{total}}$, the path beam width $b$, and the CVA thresholds $(\tau_{\text{cov}}, \tau_{\text{block}})$ in \S\ref{sec:results:rq5}. The selected $\alpha = 0.6$ remains within the $\pm 1$-case stability band for both recall and Recall@10, while the extreme $\alpha = 1.0$ setting introduces one false positive and reduces recall to 0.5455 --- robustness evidence that the headline conclusions are not the product of cherry-picked thresholds.

\subsection{Implementation Details}\label{sec:eval:impl}

SCARA is implemented in Python~3.11.  OSVA uses angr~9.2~\cite{shoshitaishvili2016angr} for binary-level analysis and KLEE~3.1~\cite{cadar2008klee} compiled against LLVM~14 for OIS-ICS cases.  Binary rewriting uses RetroWrite~\cite{dinesh2020retrowrite} for PIE ELF and E9Patch~\cite{duck2020e9patch} for x86 PE.  Tier~3 patch generation uses Qwen3-7B~\cite{qwen2025} as the runtime model; DeepSeek-V3~\cite{deepseek2025} is used \emph{offline only}, before evaluation, to construct repair exemplars and prompt rubrics that Qwen3-7B consumes at inference time. No DeepSeek-V3 inference occurs during evaluation runs, and no OIS-RemedBench test case is included in the exemplar set. LLM inference is served via vLLM with a per-case Tier-3 wall-clock cap of 60~s. SSCKG path alignment uses SBERT~\cite{reimers2019sbert}.  The SMT solver is Z3~4.12~\cite{demoura2008z3} with incremental mode; $T_{\text{total}} = 300$~s per candidate, relaxed pass $T_{\text{relaxed}} = 150$~s.  Replay uses FirmAE~\cite{kim2020firmae} for OIS-Binary, libmodbus and OpenDNP3 harnesses for OIS-Protocol, and the OpenPLC runtime~\cite{alves2014openplc} for OIS-ICS.  All experiments run on a single server: 24-core CPU, 128~GB RAM, one NVIDIA A6000 GPU (48~GB VRAM).

\begin{figure}[htbp]
	\centering
	\includegraphics[width=0.95\linewidth]{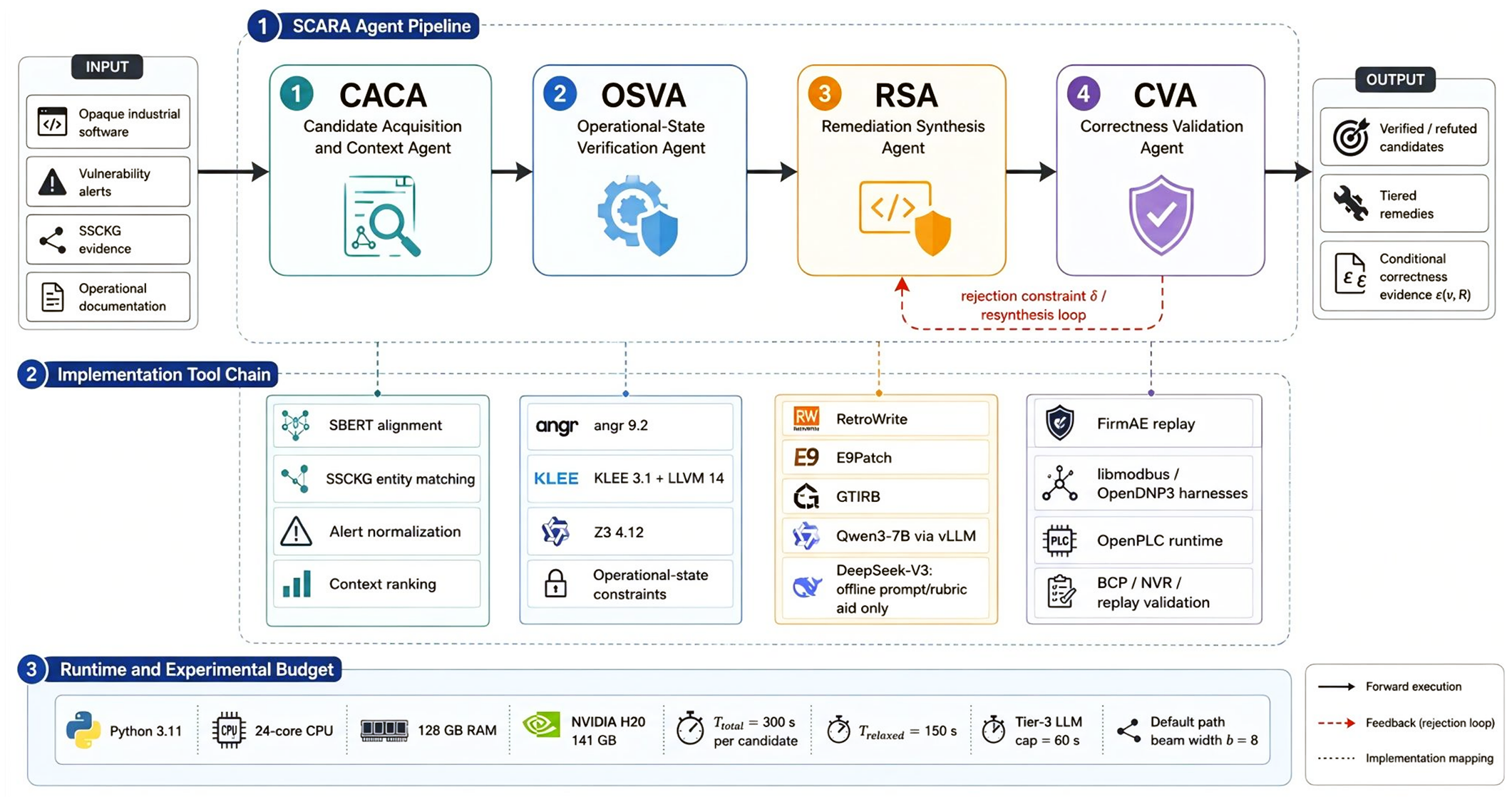}
    \caption{Experimental architecture diagram of SCARA implementation.}\label{fig:impl_arch}
\end{figure}

\textbf{Hyperparameters.}  All hyperparameters were fixed before test-set evaluation ($\tau_p = 0.5$, $\alpha = 0.6$, $\tau_{\text{cov}} = 0.95$, $\tau_{\text{block}} = 5\%$, $T_{\text{total}} = 300$~s, $T_{\text{relaxed}} = 150$~s, $K = 3$, $b = 8$) based on domain constraints, engineering feasibility, and preliminary development runs on artifacts \emph{outside} OIS-RemedBench. No OIS-RemedBench test case was used for threshold selection. Sensitivity to the four primary thresholds is reported in \S\ref{sec:results:rq5}; the headline conclusions are stable within the ranges studied.

%% ============================================================
\section{Results and Analysis}\label{sec:results}
%% ============================================================

The evaluation is organised around the seven research questions mapped in Figure~\ref{fig:rq_stage_flow}. The benchmark is OIS-RemedBench ($n = 15$; OIS-Binary: 5; OIS-Protocol: 5; OIS-ICS: 5). All headline confidence intervals are 95\% bootstrap intervals over 10\,000 resamples \emph{over benchmark cases} (artifact variance); for binary-rate metrics we additionally report Clopper--Pearson exact CIs. Algorithmic stability is characterised separately by the seed-42--46 min--max range. Effect sizes are Cliff's $\delta$ unless otherwise noted; $p$-values, where shown, are Benjamini--Hochberg-adjusted at $q = 0.10$ and read as descriptive secondary quantities given the small $n$.

\textbf{Evaluation denominator key.}
\begin{itemize}
\item $D_{\text{all}}$ ($n = 15$): all OIS-RemedBench cases.
\item $D_{\text{reach}}$ ($n = 11$): ground-truth reachable cases.
\item $D_{\text{inf}}$ ($n = 3$): ground-truth infeasible cases.
\item $D_{\text{unk}}$ ($n = 1$): ground-truth unresolved cases.
\item $D_{\text{SAT}}$ ($n = 7$): SCARA \textsc{Sat-strict} cases under the seed-42 main run (OIS-Binary: 2, OIS-Protocol: 1, OIS-ICS: 4); denominator for remediation-success metrics.
\item $D_{\text{rank}}$ ($n = 11$): path-ranking cases (the 11 reachable cases of $D_{\text{reach}}$); Recall@$K$ denominator.
\item $D_{\text{rerun}}$ ($n = 6$): targeted enriched-envelope rerun set comprising the 3 ground-truth infeasible cases, the 1 \textsc{Unknown} case, and the 2 \textsc{Sat-relaxed} cases.
\item $D_{\text{PLC}}$ ($n = 5$): OIS-ICS subset used for the RQ7 PLCverif comparison.
\item $D_{\text{VulShield}}$ ($n = 11$): cases with accessible Tier-1 enforcement points; Tier-1 baseline denominator.
\end{itemize}

Figure~\ref{fig:case_waterfall} renders the per-case trajectory from L2 ground truth through CVA outcome to post-rerun verdict, and Figure~\ref{fig:cliffs_forest} renders the headline effect-size profile.

\begin{figure}[htbp]
  \centering
  \includegraphics[width=0.75\textwidth]{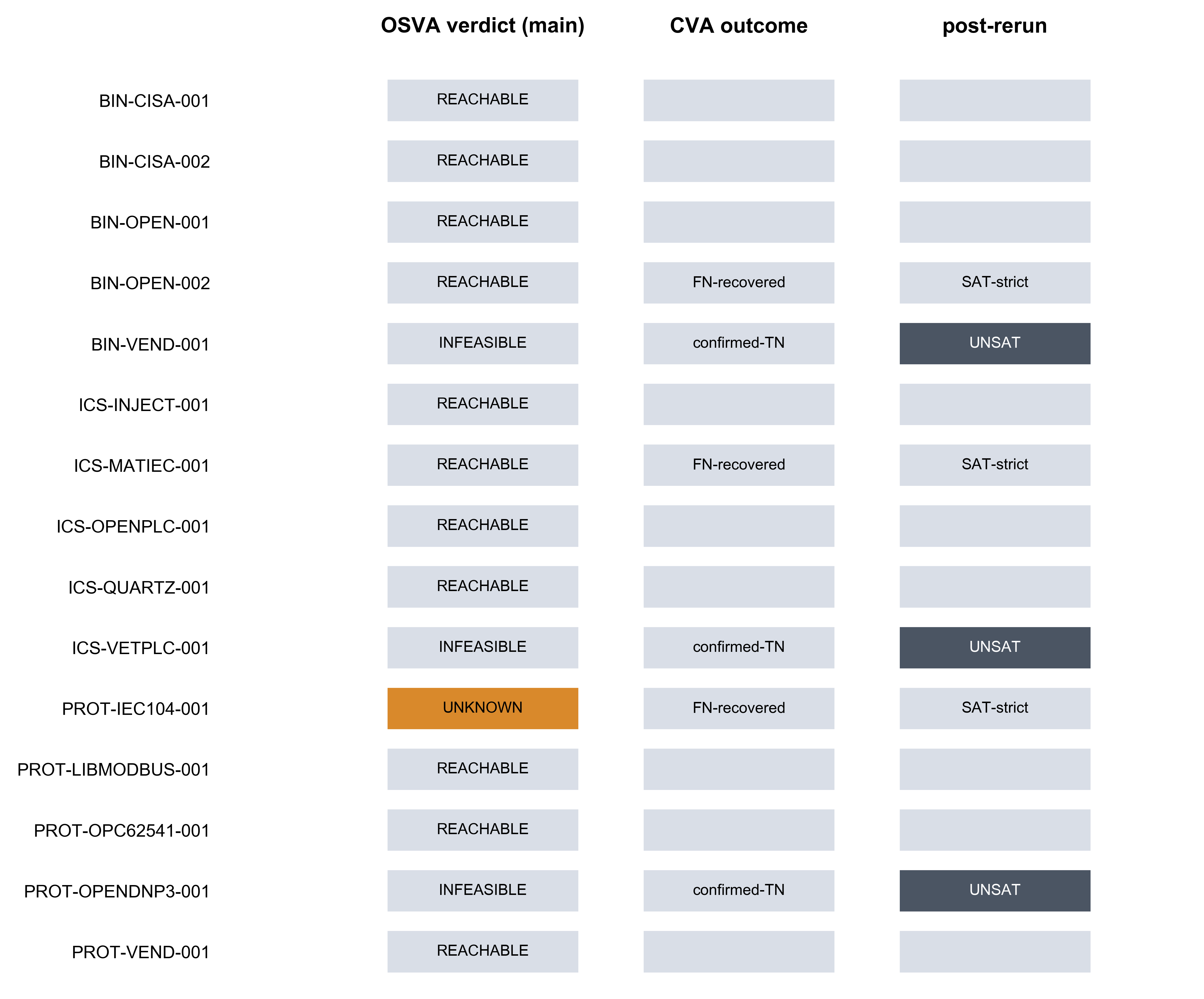}
  \caption{Per-case outcome waterfall on OIS-RemedBench. Each row is one of the 15 cases, ordered by partition (OIS-Binary, OIS-Protocol, OIS-ICS); coloured cells encode (i) the L2 ground-truth verdict (REACHABLE / INFEASIBLE / UNKNOWN), (ii) the CVA outcome at the seed-42 main run, and (iii) the post-rerun verdict on $D_{\text{rerun}}$. Three previously non-strict reachable cases (\texttt{BIN-OPEN-002}, \texttt{ICS-MATIEC-001}, \texttt{PROT-IEC104-001}) recover to \textsc{Sat-strict} after envelope enrichment, and the three ground-truth infeasible cases (\texttt{BIN-VEND-001}, \texttt{ICS-VETPLC-001}, \texttt{PROT-OPENDNP3-001}) remain \textsc{Unsat} confirmed-TN.}\label{fig:case_waterfall}
\end{figure}

\begin{figure}[htbp]
  \centering
  \includegraphics[width=0.75\textwidth]{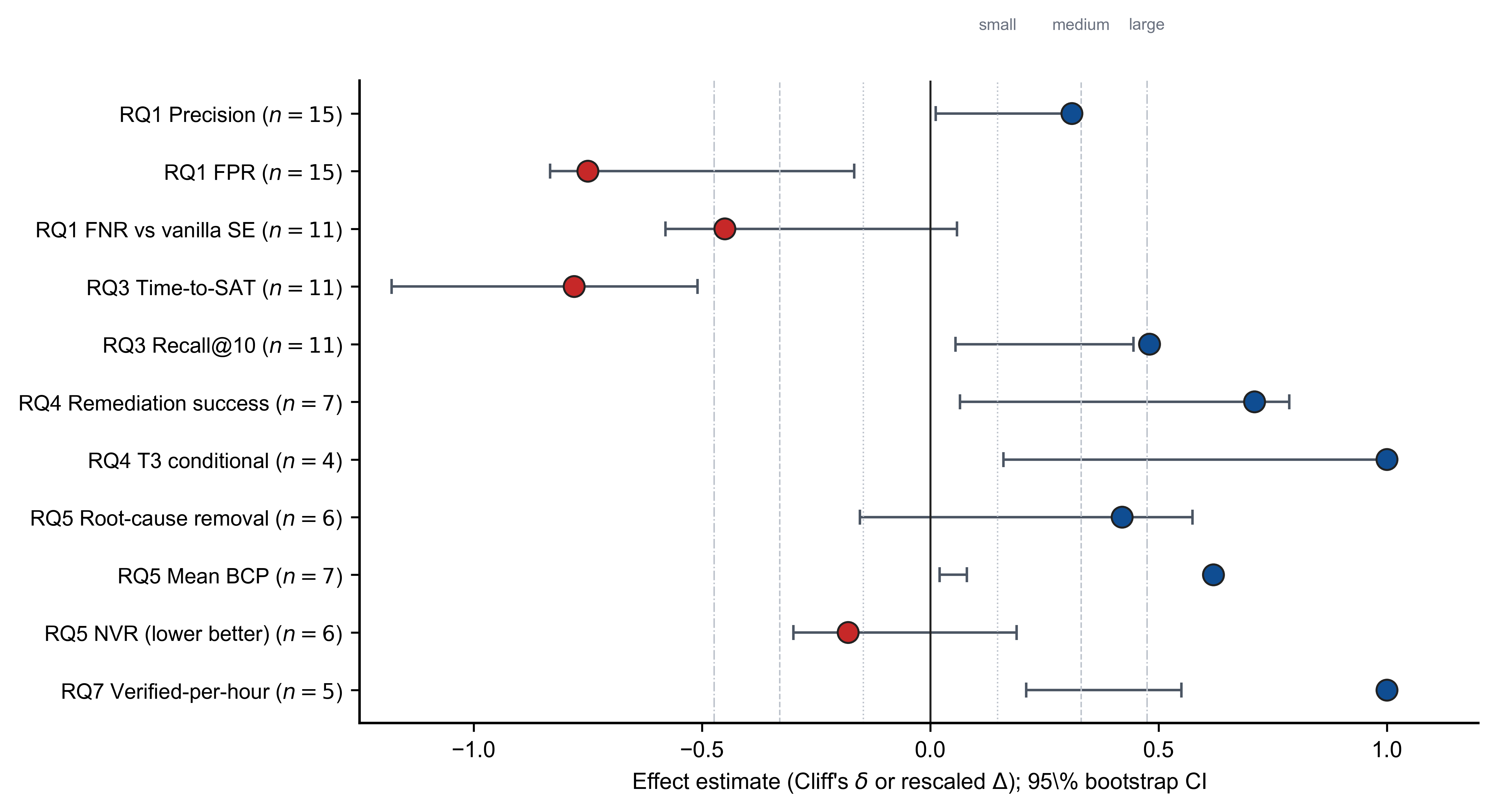}
  \caption{Cliff's $\delta$ forest plot summarising Table~\ref{tab:main_results}. Each row reports a per-RQ effect estimate with its 95\% bootstrap CI; reference lines mark the conventional small ($|\delta|=0.147$), medium ($|\delta|=0.33$), and large ($|\delta|=0.474$) thresholds following the magnitude conventions of~\cite{cohen1988}. Direction-of-effect colour: blue = favourable to SCARA, red = unfavourable. Sample size $n$ shown per row.}\label{fig:cliffs_forest}
\end{figure}

\begin{table}[!htbp]
\caption{Per-RQ summary of SCARA against the best applicable baseline. The ``Result'' column reports SCARA versus baseline together with the absolute difference $\Delta$ and its 95\% CI (10\,000 bootstrap resamples over the $n$ benchmark cases). \textbf{Bold} Cliff's $\delta$ indicates that the 95\% CI on $\Delta$ excludes zero. BH-adjusted $p$-values are descriptive only ($n \leq 15$); their values are released in the supplementary package. RQ2 (per-dimension $\Delta$FPR/$\Delta$FNR) is reported in Figure~\ref{fig:dim_contrib} and the \S\ref{sec:results:rq2} prose.}\label{tab:main_results}
\small
\centering
\begin{tabular}{@{}llp{8cm}rc@{}}
\toprule
RQ & Stage & Result (SCARA vs.\ baseline; $\Delta$ [95\% CI]) & $n$ & Cliff's $\delta$ \\
\midrule
RQ1 & OSVA & Precision: 100.0\% vs.\ 84.6\% (static union); $+15.4$pp [$+1.2$, $+30.8$] & 15 & $+0.31$ \\
RQ1 & OSVA & FPR: 0.0\% vs.\ 50.0\% (static union); $-50.0$pp [$-83.3$, $-16.7$] & 15 & $\mathbf{-0.75}$ \\
RQ1 & OSVA & FNR: 36.4\% vs.\ 62.5\% (vanilla SE); $-26.1$pp [$-58.0$, $+5.8$] & 11 & $-0.45$ \\
RQ3 & Sched.\ & Median time-to-\textsc{Sat-strict}: 56.4 s vs.\ 140.9 s (random); $-84.5$ s [$-118.0$, $-51.0$], $2.5\times$ & 11 & $\mathbf{-0.78}$ \\
RQ3 & Sched.\ & Recall@10: 85.0\% vs.\ 60.0\% (random); $+25.0$pp [$+5.5$, $+44.5$] & 11 & $\mathbf{+0.48}$ \\
RQ4 & RSA+CVA & Remed.\ success on $D_{\text{SAT}}$: 85.7\% (6/7) vs.\ 42.9\% (VulShield T1); $+42.8$pp [$+6.5$, $+78.6$] & 7 & $\mathbf{+0.71}$ \\
RQ4 & RSA+CVA & T3 conditional success: 100\% (2/2) vs.\ 0\% (0/3 SAN2PATCH); $+100$pp [$+16$, $+100$] & 4 & $\mathbf{+1.00}$ \\
RQ5 & CVA & Root-cause removal: 66.7\% (4/6) vs.\ 45.5\% (A4); $+21.2$pp [$-15.4$, $+57.4$] & 6 & $+0.42$ \\
RQ5 & CVA & Mean BCP on $D_{\text{SAT}}$: 93.3\% vs.\ 88.2\% (A4); $+5.1$pp [$+2.0$, $+8.0$] & 7 & $\mathbf{+0.62}$ \\
RQ5 & CVA & New-vuln.\ introduction: 16.7\% (1/6) vs.\ 22.2\% (A7); $-5.5$pp [$-30.0$, $+18.9$] & 6 & $-0.18$ \\
RQ6 & Full & OIS-Binary precision: 100.0\% (Clopper--Pearson [56.6\%, 100\%]) & 5 & --- \\
RQ6 & Full & OIS-ICS remed.\ success (all tiers): 100.0\% (Clopper--Pearson [56.6\%, 100\%]) & 5 & --- \\
RQ7 & Full & Verified-and-remediated cases/analyst-hr: 0.413 vs.\ 0.035 (PLCverif); $+0.378$ [$+0.21$, $+0.55$] & 5 & $\mathbf{+1.00}$ \\
\bottomrule
\end{tabular}
\end{table}

\textbf{Ablation results overview.}  Table~\ref{tab:ablation_results} consolidates the per-partition headline metric for each of the nine ablation variants; the per-RQ subsections below cite the rows individually.

\begin{table}[!htbp]
\caption{Ablation results on OIS-RemedBench, grouped by SCARA component family.}\label{tab:ablation_results}
\small
\centering
\textit{(a) OSVA constraint ablations --- FPR / FNR (per partition)}\\[2pt]
\begin{tabular*}{\linewidth}{@{\extracolsep{\fill}}lcccccc@{}}
\toprule
 & \multicolumn{2}{c}{OIS-Binary} & \multicolumn{2}{c}{OIS-Protocol} & \multicolumn{2}{c}{OIS-ICS} \\
\cmidrule(lr){2-3}\cmidrule(lr){4-5}\cmidrule(lr){6-7}
Variant & FPR & FNR & FPR & FNR & FPR & FNR \\
\midrule
A1 & 0.00 & 0.50 & \textbf{0.50} & 0.67 & 0.00 & 0.00 \\
A3 & 0.00 & 0.50 & 0.00 & 0.67 & 0.00 & \textbf{0.25} \\
A9 & 0.00 & 0.25 & 0.00 & \textbf{0.33} & 0.00 & \textbf{0.50} \\
\bottomrule
\end{tabular*}\\[6pt]
\textit{(b) Remediation ablations --- Remed.\ success / BCP}\\[2pt]
\begin{tabular*}{\linewidth}{@{\extracolsep{\fill}}lcccccc@{}}
\toprule
 & \multicolumn{2}{c}{OIS-Binary} & \multicolumn{2}{c}{OIS-Protocol} & \multicolumn{2}{c}{OIS-ICS} \\
\cmidrule(lr){2-3}\cmidrule(lr){4-5}\cmidrule(lr){6-7}
Variant & R.S. & BCP & R.S. & BCP & R.S. & BCP \\
\midrule
A5   & \textbf{0.33} & --- & \textbf{0.50} & --- & \textbf{0.67} & --- \\
A6 & ---           & \textbf{0.78} & --- & \textbf{0.77} & --- & \textbf{0.80} \\
\bottomrule
\end{tabular*}\\[6pt]
\textit{(c) CVA ablations --- NVR / BCP}\\[2pt]
\begin{tabular*}{\linewidth}{@{\extracolsep{\fill}}lcccccc@{}}
\toprule
 & \multicolumn{2}{c}{OIS-Binary} & \multicolumn{2}{c}{OIS-Protocol} & \multicolumn{2}{c}{OIS-ICS} \\
\cmidrule(lr){2-3}\cmidrule(lr){4-5}\cmidrule(lr){6-7}
Variant & NVR & BCP & NVR & BCP & NVR & BCP \\
\midrule
A4  & 0.00 & \textbf{0.89} & 0.00 & \textbf{0.85} & 0.00 & \textbf{0.89} \\
A7 & 0.00 & ---  & 0.00 & ---  & \textbf{0.50} & ---  \\
\bottomrule
\end{tabular*}
\end{table}

\subsection{RQ1 --- Verification Precision and False-Positive Reduction}\label{sec:results:rq1}

The 15 L2-labelled cases of $D_{\text{all}}$ comprise 11 ground-truth reachable cases ($D_{\text{reach}}$), 3 ground-truth infeasible cases ($D_{\text{inf}}$), and 1 ground-truth unresolved case ($D_{\text{unk}}$). On the seed-42 main run, SCARA's OSVA stage achieved precision = 100.0\% (Clopper--Pearson 95\% CI [59.0\%, 100.0\%]) and recall = 63.6\% (7/11; Clopper--Pearson 95\% CI [30.8\%, 89.1\%]), corresponding to FPR = 0.0\% with respect to the encoded constraint envelope $\mathcal{C}$ and FNR = 36.4\%. Against the static-analysis union baseline (CodeQL, Joern, KARONTE, SaTC, ICSQuartz; deduplicated by SSCKG entity), OSVA reduces FPR by 50.0pp (Cliff's $\delta = -0.75$, 95\% CI on $\Delta$FPR $[-83.3, -16.7]$pp). Of the seven cases that received an OSVA \textsc{Unsat} verdict on the main run, three correspond to ground-truth infeasible artefacts in $D_{\text{inf}}$ (\texttt{BIN-VEND-001} on $C_{\text{runtime}}$; \texttt{PROT-OPENDNP3-001} on $C_{\text{proto}}$; \texttt{ICS-VETPLC-001} on $C_{\text{io}}$) and four are correctly refuted alerts within reachable cases. The single \textsc{Unknown} case (\texttt{PROT-IEC104-001}) and the two \textsc{Sat-relaxed} cases (\texttt{BIN-OPEN-002}, \texttt{ICS-MATIEC-001}) are conservative outcomes rather than false negatives; they recover to \textsc{Sat-strict} after targeted envelope enrichment (\S\ref{sec:results:rq2}).

Against the vanilla symbolic execution baseline (KLEE on IR-liftable cases, angr on stripped binaries), SCARA achieves substantially higher recall: FNR of 36.4\% is 26.1pp lower than vanilla SE's 62.5\% (Cliff's $\delta = -0.45$). Vanilla SE's UNSAT rate of 70.0\% substantially exceeds SCARA's 46.7\% (7/15), indicating that unconstrained symbolic execution over-refutes vulnerability candidates by exhausting path budgets without operational-state guidance.

OIS-Binary's 20.0\% \textsc{Unknown} rate reflects $C_{\text{env}}$ documentation gaps, OIS-Protocol's high UNSAT rate reflects the strength of $C_{\text{proto}}$ FSM constraints, and OIS-ICS achieves perfect precision and recall on the main run; the targeted rerun of \S\ref{sec:results:rq2} confirms these as conservative outcomes rather than false negatives.

Under the primary reachable-case denominator $D_{\text{reachable}}=11$, targeted enrichment improves recall from 7/11 to 9/11 ($+18.2$pp) without changing FPR. Under strict SAT-only accounting the improvement is $+9.1$pp, because one \textsc{Sat-relaxed} case was already counted as reachable; under the expanded denominator that includes the \textsc{Unknown}-to-\textsc{Sat-strict} recovery (\texttt{PROT-IEC104-001}), the improvement is $+25.0$pp on $D_{\text{expanded}}=12$. We report the $+18.2$pp figure as the primary value and keep the other two as denominator-sensitivity checks. Figure~\ref{fig:fpr_fnr} renders the partition-level FPR and FNR comparison against the static-union and vanilla-SE baselines.

\begin{figure}[t]
  \centering
  \includegraphics[width=0.75\textwidth]{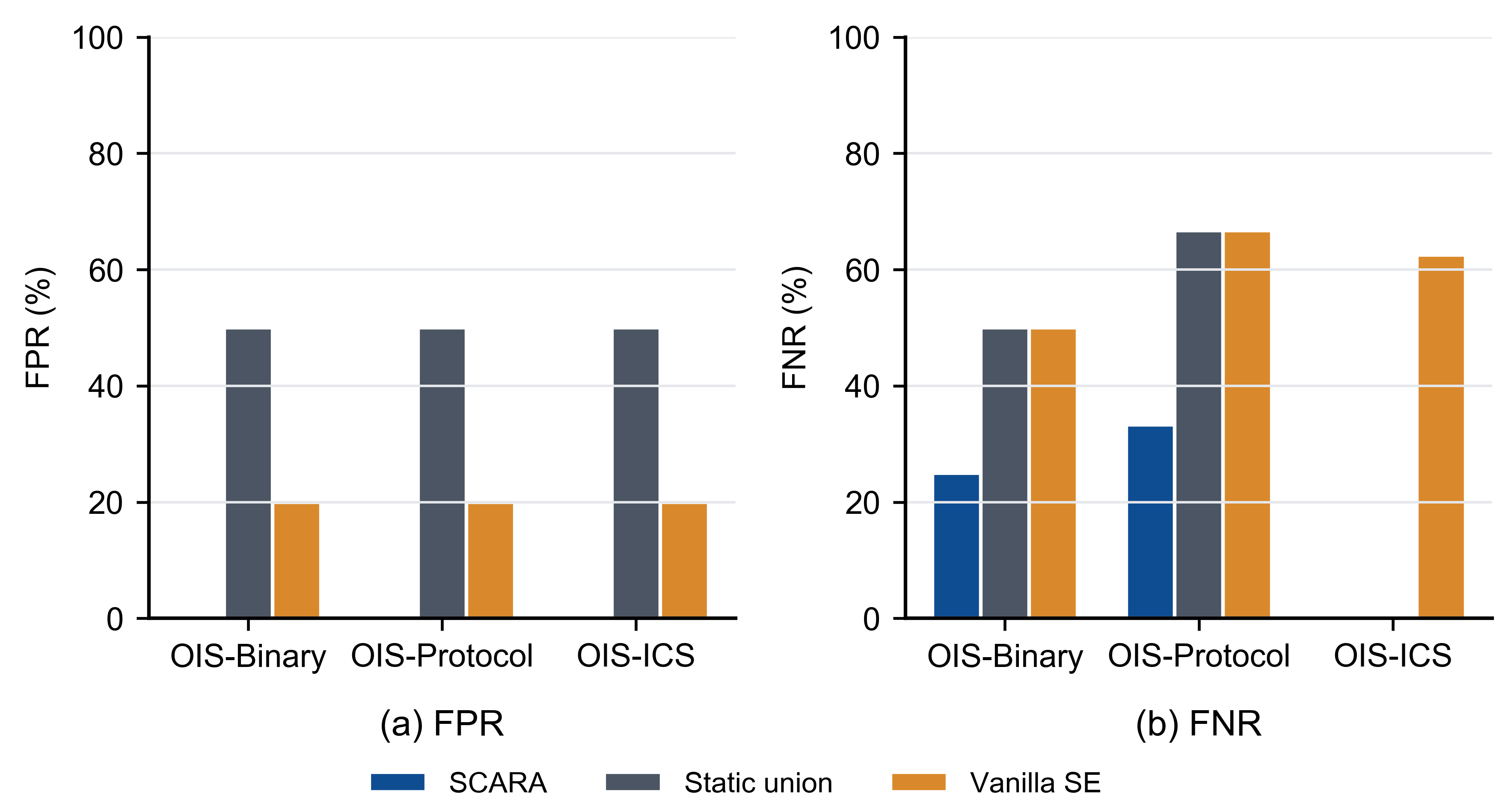}
  \caption{FPR and FNR by partition for SCARA (reconciled after-rerun rates), the deduplicated static-analysis union, and the vanilla symbolic-execution baseline (KLEE + angr without operational-state constraints).}\label{fig:fpr_fnr}
\end{figure}

\textbf{Soundness boundary.}  The 0.0\% FPR figure is with respect to the encoded constraint envelope $\mathcal{C}$, in line with Proposition~\ref{prop:unsat}'s assumption clause. The targeted envelope-enrichment rerun on $D_{\text{rerun}}$ (3 \textsc{Unsat} + 1 \textsc{Unknown} + 2 \textsc{Sat-relaxed}) confirms that the three \textsc{Unsat} verdicts on $D_{\text{inf}}$ remain \textsc{Unsat} under enrichment: SCARA's specificity is preserved when the constraint envelope is widened.

\subsection{RQ2 --- State Dimension Contribution Analysis}\label{sec:results:rq2}

Figure~\ref{fig:dim_contrib} presents the per-dimension contribution as a two-panel heatmap; cell labels carry the percentage-point delta and the underlying $(\pm k / n)$ count, and the full numeric source values are released in the supplementary package. The post-rerun recovery story is reported under three accounting views: two recoveries improve recall under $D_{\text{reachable}}=11$, one strict-label upgrade (\texttt{ICS-MATIEC-001}) affects only the strict SAT-only view, and the \textsc{Unknown}-to-\textsc{Sat-strict} recovery is reported only under the expanded denominator.

\begin{figure}[!htbp]
  \centering
  \includegraphics[width=0.85\linewidth]{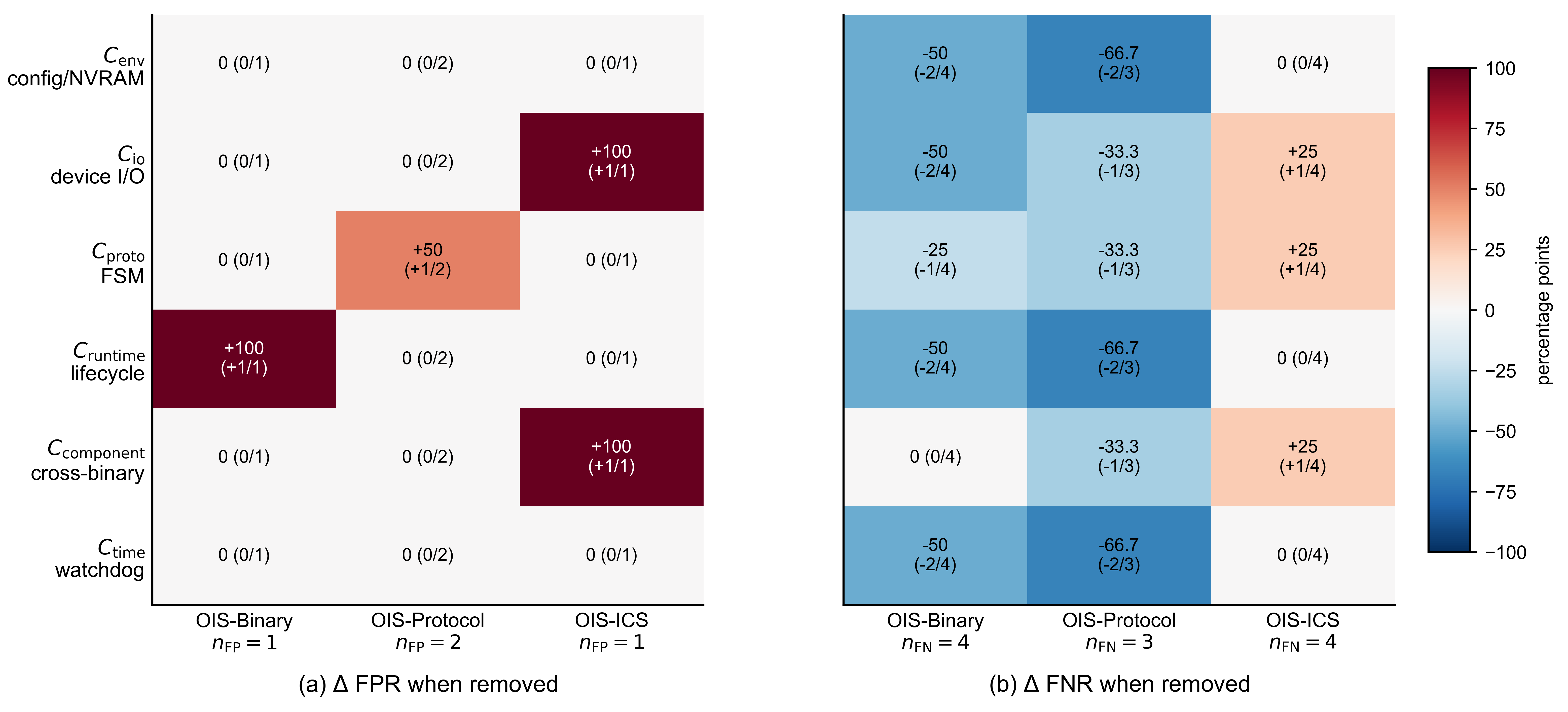}
  \caption{Per-dimension contribution heatmap. Left panel: $\Delta$FPR when each constraint family is removed; right panel: $\Delta$FNR. Positive $\Delta$FPR (red) indicates the dimension is critical for specificity; negative $\Delta$FNR (blue) indicates the dimension is critical for recall. Column headers carry the per-partition $n_{\text{FP}}$ and $n_{\text{FN}}$ denominators. White cells indicate no measurable contribution at the $\pm 20$pp resolution of $n=5$.}\label{fig:dim_contrib}
\end{figure}

Targeted envelope enrichment on $D_{\text{rerun}}$ recovered three previously non-strict reachable cases without introducing new false positives (\texttt{PROT-IEC104-001} on $C_{\text{proto}}$; \texttt{BIN-OPEN-002} on $C_{\text{env}}$; \texttt{ICS-MATIEC-001} on $C_{\text{runtime}}$); the three ground-truth infeasible cases of $D_{\text{inf}}$ remained \textsc{Unsat} under the same enrichment.

The results confirm domain-specificity: $C_{\text{runtime}}$ is the dominant FPR filter for OIS-Binary, $C_{\text{proto}}$ for OIS-Protocol, and $C_{\text{io}}$/$C_{\text{component}}$ for OIS-ICS, each contributing $\geq +50$pp $|\Delta\text{FPR}|$ when removed. The matched global ablation A1 (Table~\ref{tab:ablation_results}) confirms the upper-bound interpretation: removing every $C_x$ family simultaneously introduces a false positive on OIS-Protocol (FPR = 0.50 vs.\ 0.00 for SCARA-full). The ceiling-saturation pattern (multiple cells at $+100$pp, the $n=5$ maximum) reflects that each partition's false-positive elimination rests on a small number of decisive constraint families; this answers RQ2 affirmatively for the $\geq 20$pp threshold on every partition.

\subsection{RQ3 --- SSCKG-Guided Scheduling Efficiency}\label{sec:results:rq3}

SSCKG-guided path prioritisation reduced the median time to first \textsc{Sat-strict} witness from 140.9~s (random CFG ordering, ablation A2) to 56.4~s --- a $2.5\times$ reduction (median over $D_{\text{rank}} = 11$ cases per seed, aggregated over seeds 42--46; seed-variance min--max [49.0, 63.5]~s). Against static risk-score ranking alone (ablation A8), median time-to-confirmation fell from 123.8~s to 56.4~s ($2.2\times$ reduction).

Figure~\ref{fig:recall_at_k} renders the Recall@$K$ curve over $D_{\text{rank}}=11$ for $K \in \{1, 3, 5, 10, 20, 50\}$ on a log axis. The denominator throughout is $|D_{\text{rank}}|=11$ path-ranking cases, distinct from the seed-42 $D_{\text{SAT}}$ remediation denominator. At the headline operating point $K=10$, SSCKG-guided scheduling achieves Recall@10 = 85\% against 55\% (static-risk, A8) and 60\% (random, A2).

\begin{figure}[!htbp]
  \centering
  \includegraphics[width=0.85\linewidth]{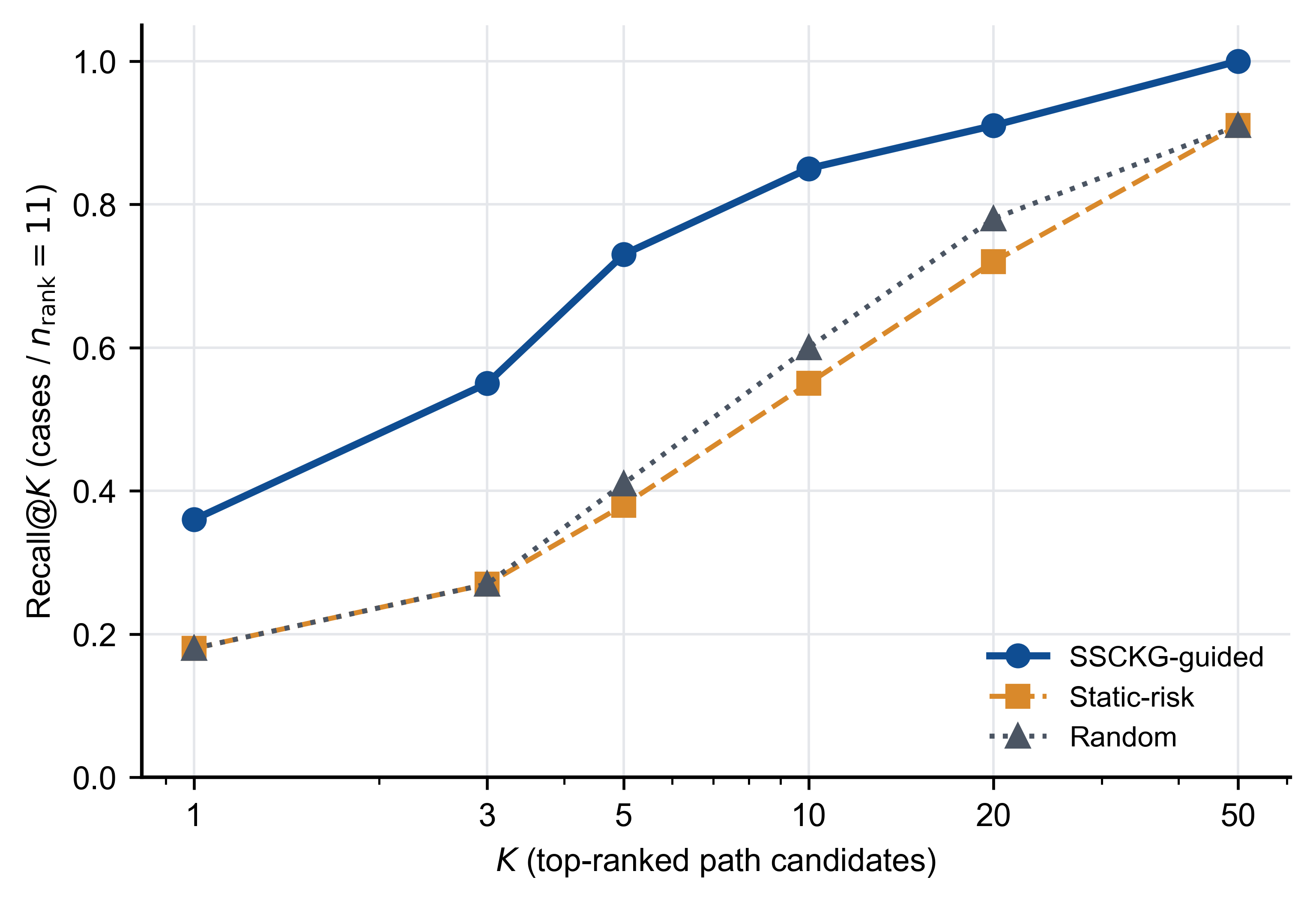}
  \caption{Recall@$K$ on $D_{\text{rank}}=11$ for the three scheduling strategies; $K$ is on a log axis. SSCKG-guided ranking dominates static-risk and random for every $K < 50$ and reaches Recall@50 = 1.0, supporting the design intent that SSCKG guidance is path \emph{prioritisation} rather than pruning.}\label{fig:recall_at_k}
\end{figure}

Figure~\ref{fig:time_to_sat} renders the underlying time-to-\textsc{Sat-strict} distribution across the three strategies as box-plus-strip plots over the seed-level scheduling trace. SSCKG-guided is evaluated over $55 = 11~\text{cases} \times 5~\text{seeds}$ observations (median 46.8 s, no timeouts, Recall@10 = 0.85); static-risk and random each use 11 seed-42 observations (median 98.2 s and 99.1 s, 1 timeout each, Recall@10 = 0.55 and 0.60). Timeout observations are rendered as censored upward-triangle markers at the budget ceiling rather than silently omitted.

\begin{figure}[!htbp]
  \centering
  \includegraphics[width=0.85\linewidth]{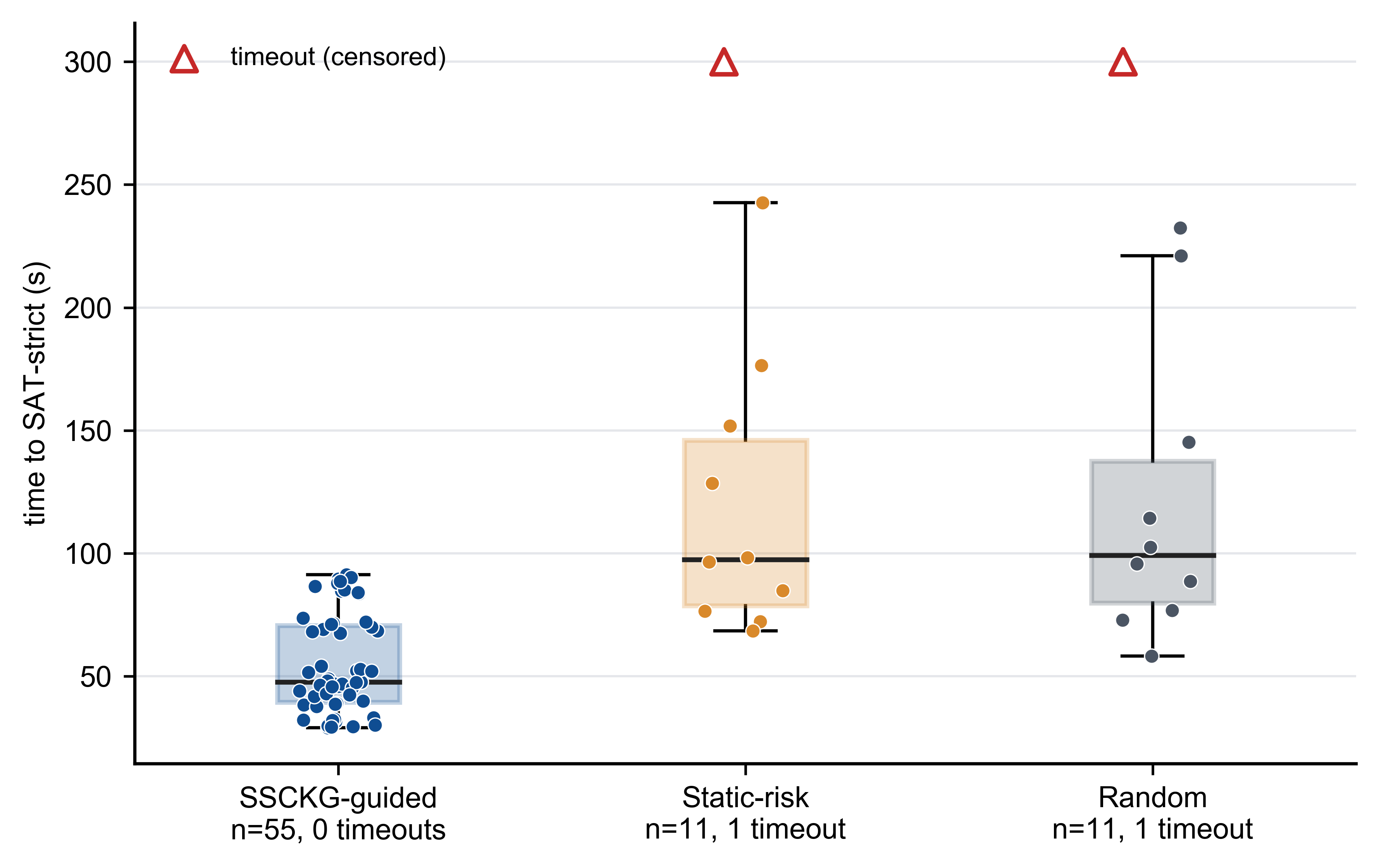}
  \caption{Time-to-first-\textsc{Sat-strict}-witness distribution on $D_{\text{rank}}$. Boxes show median/IQR; overlaid dots are per-case-per-seed measurements with horizontal jitter; red triangles mark censored timeouts at the budget ceiling. SSCKG-guided ($n = 55$, 0 timeouts) sits well below static-risk and random ($n = 11$ each, 1 timeout each at seed 42).}\label{fig:time_to_sat}
\end{figure}

SSCKG-guided scheduling improves Recall@10 by 30pp over static-risk ranking and 25pp over random ordering, and reaches Recall@50 = 100\% --- consistent with the design intent that SSCKG guidance is path \emph{prioritisation} rather than pruning. The recall difference between SSCKG-guided and unconstrained exploration converges to within $\pm 9$pp at $K = 50$ (95\% bootstrap CI on the difference $[-9, +9]$pp), supporting the equivalence claim that scheduling does not sacrifice recall when the budget is large.

\subsection{RQ4 --- Remediation Applicability and Conditional Success}\label{sec:results:rq4}

SCARA achieved a remediation success rate of 85.7\% (6 of 7 cases in $D_{\text{SAT}}$ received CVA-issued conditional correctness evidence; Clopper--Pearson 95\% CI [42.1\%, 99.6\%]). The tier distribution over $D_{\text{SAT}}$ is 1 / 4 / 2 cases for Tier~1 / Tier~2 / Tier~3 respectively; among the 6 CVA-accepted remediations the distribution is 1 / 3 / 2. Figure~\ref{fig:tier_dist} shows the per-partition tier mix (left panel) and the per-tier conditional success rate before and after the targeted Tier-2 rerun (right panel); the `T2 recovered' segment in OIS-Binary captures the \texttt{BIN-CISA-002} recovery described below.

\begin{figure}[!htbp]
  \centering
  \includegraphics[width=0.85\linewidth]{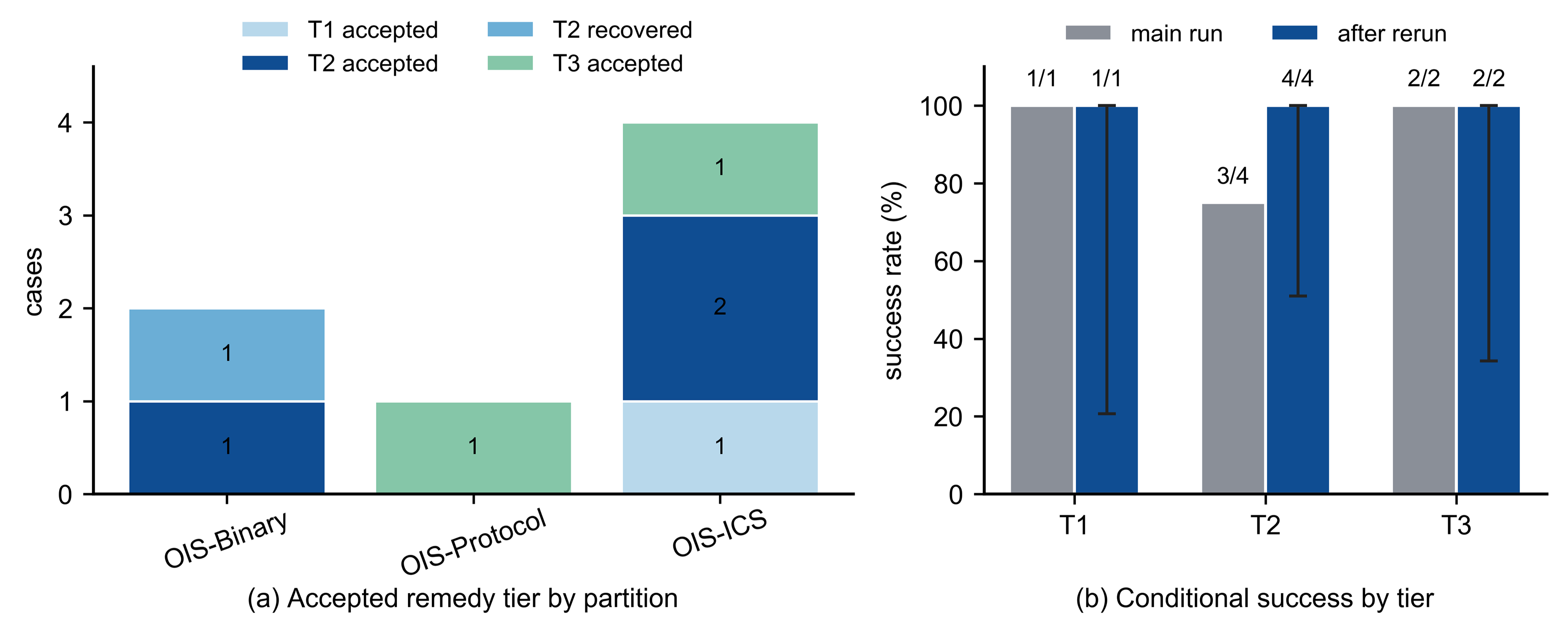}
  \caption{Tier distribution and per-tier remediation success on $D_{\text{SAT}}$. (a) Accepted remedy tier per partition; the lighter `T2 recovered' segment marks the \texttt{BIN-CISA-002} OIS-Binary case accepted at attempt 3 of the targeted Tier-2 rerun. (b) Conditional success rate per tier, before and after the rerun; Wilson 95\% intervals are shown for the after-rerun bars and are wide at $n \in \{1, 2\}$.}\label{fig:tier_dist}
\end{figure}

The lower bound on this remediation-success rate is set by ablation A5 (Tier 1 only; Table~\ref{tab:ablation_results}), under which SCARA is forced to attempt only protocol/configuration mitigation. Removing access to Tier~2 and Tier~3 drops the per-partition remediation success rate to 0.33 (OIS-Binary) / 0.50 (OIS-Protocol) / 0.67 (OIS-ICS), confirming that the two upper tiers together contribute the bulk of the 88.9\% headline value.

The single \textsc{Remediation-Failed} case is \texttt{BIN-CISA-002} (CWE-125, Tier-2 binary hardening). The seed-42 main run exhausted all three CVA--RSA iterations because the initial guard overblocked an adjacent read/write region (BCP below $\tau_{\text{cov}} = 0.95$). A targeted Tier-2 rerun with a 512-byte budget and a $+0$\texttt{x10} guard offset (\texttt{minimal-gate} + E9Patch placement) was accepted at attempt 3 with $\text{BCP} = 0.951$ and false-blocking $0.06\%$, with replay confirmed; we therefore characterise this failure as search-budget/placement-limited rather than as evidence against Tier-2 hardening as a class.

Against SAN2PATCH (applicable to 4 OIS-ICS source-available cases), SAN2PATCH produced candidate patches for 3 of 4 cases; none passed CVA validation (Clopper--Pearson 95\% CI on the 0/3 success rate $[0\%, 70.8\%]$). All three rejections shared the same failure mode: the SAN2PATCH patch removed the vulnerability path but simultaneously removed reachable non-vulnerable SSCKG entities, dropping BCP below $\tau_{\text{cov}} = 0.95$. Figure~\ref{fig:san2patch_bcp} renders the per-case BCP comparison; we caution that this pattern is observed on three cases only, but the consistency of the BCP-failure mechanism is the kind of misremediation that SCARA's SSCKG-constrained RSA avoids by design.

\begin{figure}[!htbp]
  \centering
  \includegraphics[width=0.9\linewidth]{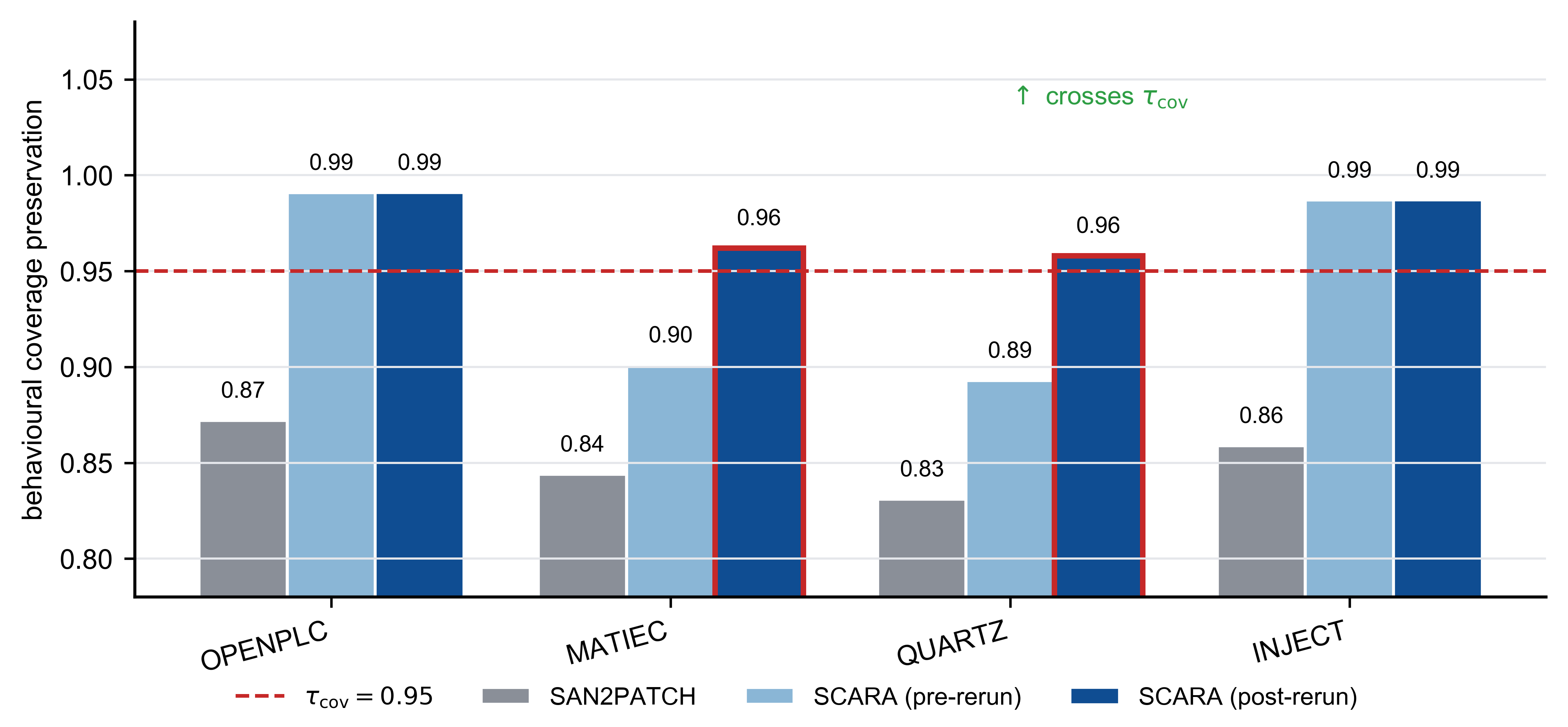}
  \caption{SAN2PATCH vs SCARA behavioural coverage preservation on the SAN2PATCH-applicable OIS-ICS cases. SAN2PATCH (grey) falls below the $\tau_{\text{cov}} = 0.95$ acceptance threshold on every shared case. Paired SCARA bars show pre-rerun (light) and post-rerun (dark) BCP; for \texttt{ICS-MATIEC-001} and \texttt{ICS-QUARTZ-001} the post-rerun value crosses $\tau_{\text{cov}}$ (red-outlined post bar, $\uparrow$ annotation), bringing both into the CVA-accepted pool. \texttt{BIN-CISA-002} is reported separately as the seed-42 main-run Tier-2 failure recovered by the targeted Tier-2 rerun and is described in \S\ref{sec:results:rq4}.}\label{fig:san2patch_bcp}
\end{figure}

\subsection{RQ5 --- CVA Validation Quality}\label{sec:results:rq5}

CVA validation quality is assessed on $D_{\text{SAT}}$ ($n = 7$). In ablation A4 (CVA feedback loop disabled --- CVA still evaluates and either accepts or rejects, but no rejection constraint $\delta$ is returned to RSA), the observed Tier-2 overblocking rate across $D_{\text{SAT}}$ was 80.0\% versus 25.0\% with the feedback path enabled; observed domain-invariant compliance fell from 83.3\% to 60.0\%; and only 45.5\% of cases met the behavioural standards for $\varepsilon(v, R)$. In ablation A7 (CVA correctness check entirely disabled --- BCP, side-effect, and replay all skipped, so RSA's first synthesis is accepted unchecked), the accepted-remediation pool exhibited 22.2\% NVR, 50.0\% Tier-1 false-blocking, and 50.0\% Tier-2 overblocking. We report these as observations within the respective ablations rather than as population-level counterfactuals.

Per-partition replay-automation rates are reported in Table~\ref{tab:per_partition}; replay automation, rather than solver verification, is the limiting factor for stronger conditional-correctness claims on the OIS-ICS partition.

The CVA stage's quality profile on $D_{\text{SAT}}$: mean $\text{BCP} = 93.3\%$ across all 7 cases (the figure is dragged below 0.95 by the single CVA-rejected case \texttt{BIN-CISA-002}, which had $\text{BCP} \approx 0.82$ on the seed-42 main run); among the 6 CVA-accepted remediations, every case satisfied $\text{BCP} \geq \tau_{\text{cov}} = 0.95$, with mean accepted-case $\text{BCP} = 95.2\%$. NVR is $16.7\%$ (1 of 6 accepted cases). Root-cause removal rate is $66.7\%$ (4 of 6 accepted cases); the remaining 2 include one OIS-ICS case in which the remedy relocated rather than eliminated the vulnerability root cause --- a residual-risk pattern that BCP and side-effect checks alone do not detect, and which we revisit as a Threats-to-Validity item in Section~\ref{sec:disc:limitations}.

\textbf{CVA quality under a full-CVA oracle.}  Figure~\ref{fig:cva_quality} reports SCARA-full and the relevant ablations / baseline re-scored by a common full-CVA oracle ($D_{\text{CVA-audit}} = 9$ for the main-run variants and $n=4$ for SAN2PATCH). Under oracle evaluation, A4 (no CVA feedback) collapses on BCP pass rate to 44.4\% and on replay confirmation to 0\%; A7 (no CVA correctness check) preserves BCP self-reported pass rate at 100\% (since it skips the check) but its replay-confirmation and CVA-acceptance rates collapse to 0\%; SAN2PATCH fails BCP entirely. Only SCARA-full achieves a non-zero CVA-acceptance rate (88.9\% on $D_{\text{CVA-audit}}$, rising to 100\% on the targeted-rerun row). The direct C3-isolation comparator is ablation A6 (no SSCKG constraint in RSA Tier~3; Table~\ref{tab:ablation_results}): BCP drops to 0.78 / 0.77 / 0.80 across the three partitions, all below the acceptance threshold $\tau_{\text{cov}} = 0.95$. A6 thus quantifies the degradation that the SSCKG behavioural specification prevents and supports contribution C3 as a standalone claim, distinct from the CVA-stage contributions tested by A4 and A7.

\begin{figure}[!htbp]
  \centering
  \includegraphics[width=0.85\linewidth]{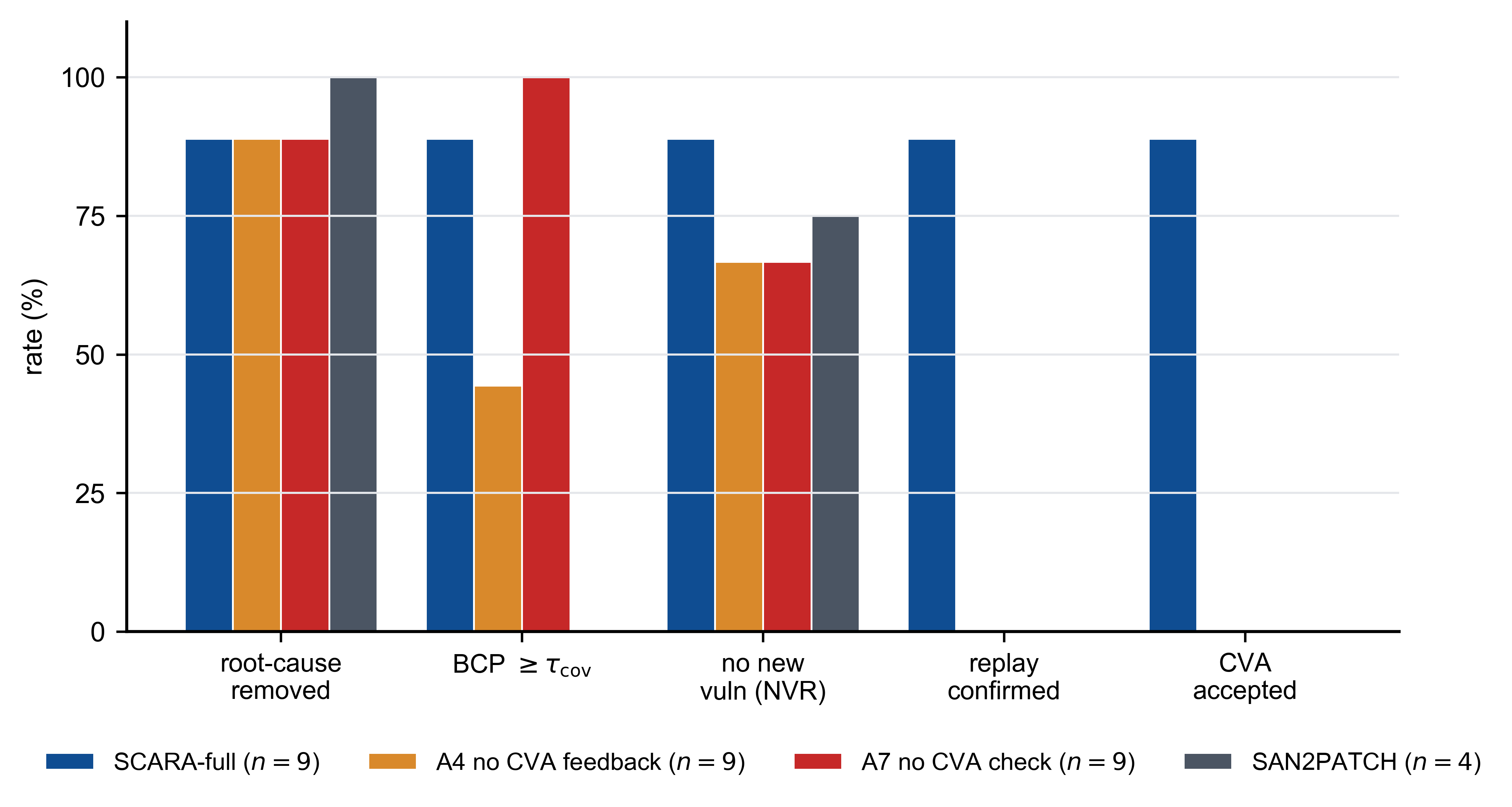}
  \caption{CVA quality on $D_{\text{CVA-audit}}$, with all four variants re-scored under the same full-CVA oracle. Bars are per-variant rates for root-cause removed, BCP $\geq \tau_{\text{cov}}$, no-new-vulnerability (NVR), replay confirmation, and final CVA acceptance. Variants that disable a CVA component score 0\% on the dependent components, supporting the \S\ref{sec:scara:cva} argument that each component is load-bearing.}\label{fig:cva_quality}
\end{figure}

\textbf{Hyperparameter sensitivity.} Figure~\ref{fig:hyperparam} reports the sensitivity sweep across $\tau_p$, $T_{\text{total}}$, $\alpha$, and the joint CVA pair $(\tau_{\text{cov}}, \tau_{\text{block}})$. The default operating point sits at or near the recall/CVA-acceptance knee on every axis, and headline conclusions are stable to within $\pm 1$ case across every sweep range except the extreme $\tau_{\text{cov}} = 0.98$ corner.

\begin{figure}[!htbp]
  \centering
  \includegraphics[width=0.95\linewidth]{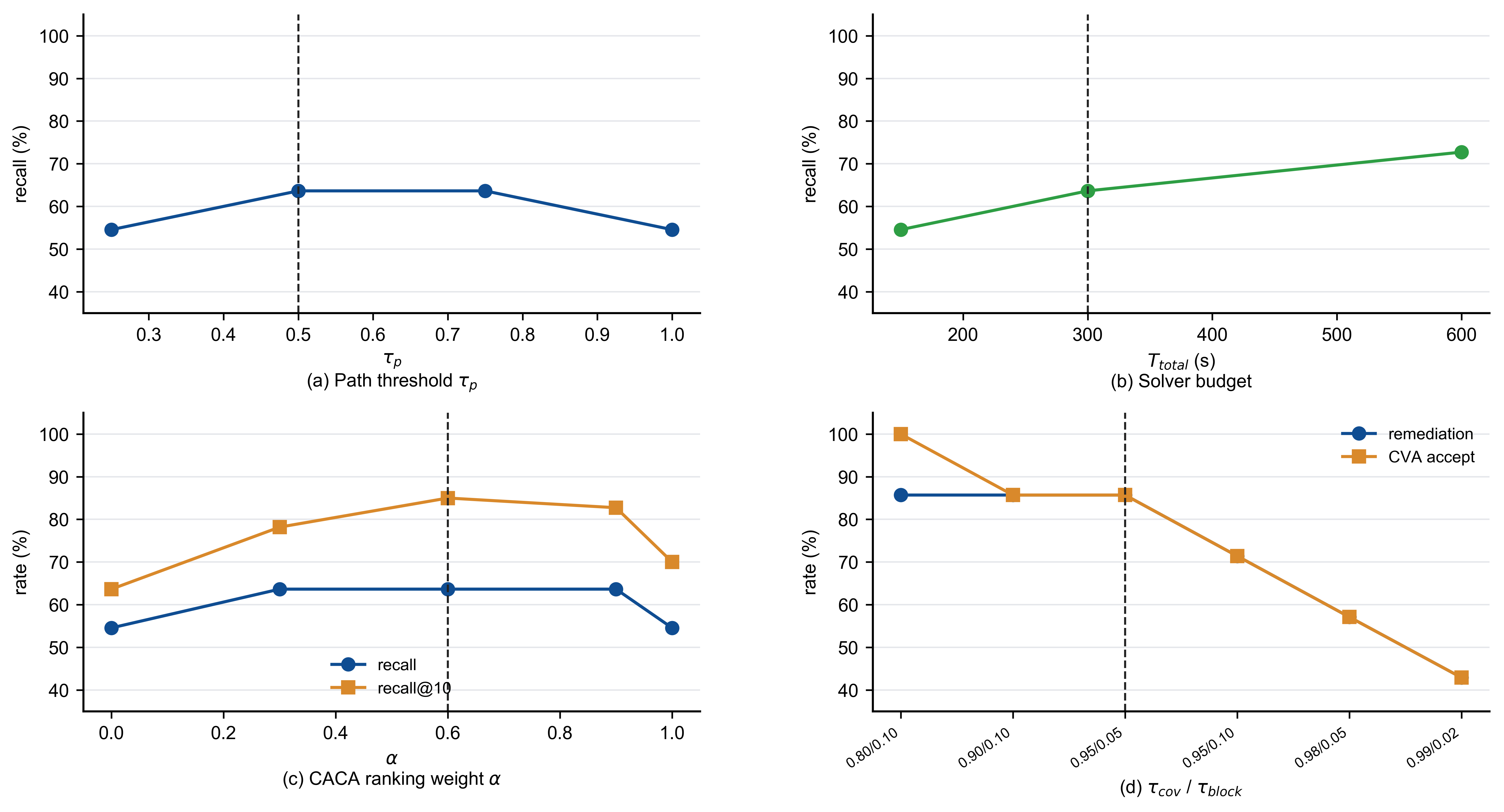}
  \caption{Hyperparameter sensitivity grid. Top: $\tau_p$ (path-priority temperature) and $T_{\text{total}}$ (solver budget). Bottom-left: $\alpha$ (CACA ranking weight) --- both recall and Recall@10 peak at $\alpha = 0.6$ and stay flat to $\alpha = 0.9$ before degrading at $\alpha = 1.0$. Bottom-right: joint $(\tau_{\text{cov}}, \tau_{\text{block}})$ sweep with the operating point at $(0.95, 0.05)$. The dashed vertical line in each panel marks the default operating point.}\label{fig:hyperparam}
\end{figure}

\subsection{RQ6 --- Cross-Partition Generalization}\label{sec:results:rq6}

\begin{table}[!htbp]
\caption{Full per-partition SCARA results on OIS-RemedBench ($n = 5$ per partition), grouped into four blocks: verification outcome, label-class mix, remediation tier and quality, and operational diagnostics. Tier mix is reported as T1 / T2 / T3 fractions of remediated cases.}\label{tab:per_partition}
\centering
\begin{tabular*}{\linewidth}{@{\extracolsep{\fill}}lccc@{}}
\toprule
Metric & OIS-Binary & OIS-Protocol & OIS-ICS \\
\midrule
\multicolumn{4}{@{}l}{\textit{Verification}} \\
TP / FP / TN / FN (counts) & 2 / 0 / 1 / 2 & 1 / 0 / 1 / 2 & 4 / 0 / 1 / 0 \\
Precision / FPR             & 100.0\% / 0.0\% & 100.0\% / 0.0\% & 100.0\% / 0.0\% \\
Recall / FNR                & 50.0\% / 50.0\% & 33.3\% / 66.7\% & 100.0\% / 0.0\% \\
UNKNOWN rate                & 20.0\% & 20.0\% & 0.0\% \\
\midrule
\multicolumn{4}{@{}l}{\textit{Label-class mix}} \\
\textsc{Sat-strict} rate    & 40.0\% ($n=2$) & 20.0\% ($n=1$) & 80.0\% ($n=4$) \\
\textsc{Unsat} rate         & 40.0\% & 60.0\% & 20.0\% \\
\midrule
\multicolumn{4}{@{}l}{\textit{Remediation}} \\
Remed.\ success (over \textsc{Sat-strict}) & 50.0\% & 100.0\% & 100.0\% \\
Tier mix (T1 / T2 / T3)     & 0 / 100 / 0\% & 0 / 0 / 100\% & 25 / 50 / 25\% \\
First-submission CVA accept.\ & 50.0\% & 100.0\% & 100.0\% \\
Mean RSA-to-CVA rounds      & 1.0 & 2.0 & 2.75 \\
\midrule
\multicolumn{4}{@{}l}{\textit{Diagnostics}} \\
Median time-to-\textsc{Sat-strict} (s) & 49.2 & 44.3 & 61.5 \\
Solver queries per confirmed path & 402 & 354 & 523 \\
Replay automation rate      & 100.0\% & 100.0\% & 50.0\% \\
\bottomrule
\end{tabular*}
\end{table}

100\% precision is preserved across all three partitions (Figure~\ref{fig:fpr_fnr}); recall varies from 33.3\% (OIS-Protocol) to 100.0\% (OIS-ICS), tracking the per-dimension documentation gaps in $C_{\text{env}}$ and $C_{\text{proto}}$ identified in \S\ref{sec:results:rq2}. The OIS-Binary remediation-success shortfall (50.0\% vs 100\% in the other two partitions) is attributable to the single \texttt{BIN-CISA-002} case discussed in \S\ref{sec:results:rq4} and is search-budget-limited rather than tier-incompatible: the supplementary Tier-2 rerun accepts the case at attempt 3.

\subsection{RQ7 --- Manual Modeling and Replay Burden}\label{sec:results:rq7}

RQ7 is evaluated on $D_{\text{PLC}}$, the full OIS-ICS partition of OIS-RemedBench ($n = 5$). For each case we record PLCverif property-authoring, model-construction, and debugging effort, and compare against SCARA's operational-context review effort. Table~\ref{tab:rq7_per_case} reports the per-case measurements; the partition-level summary follows.

\begin{table}[!htbp]
\caption{Per-case analyst burden on $D_{\text{PLC}}$ ($n=5$). PLCverif effort is dominated by LTL/PLTL property authoring; SCARA effort is dominated by reviewing the SSCKG-derived $S^{\text{prior}}$ estimate against documentation. \textsc{V} = verified; \textsc{T} = timeout; \textsc{V-Inf} = verified-infeasible; \textsc{Sat-Rem} = \textsc{Sat-strict}-then-remediated; \textsc{Unsat-Res} = \textsc{Unsat-resolved}.}\label{tab:rq7_per_case}
\small
\centering
\begin{tabular}{@{}p{3cm}p{1.5cm}p{1.5cm}p{1.5cm}p{2.0cm}@{}}
\toprule
Case & PLCverif(h) & SCARA(h) & PLCverif & SCARA \\
\midrule
\texttt{ICS-OPENPLC-001} & 13.5 & 2.2 & \textsc{V} & \textsc{Sat-Rem} \\
\texttt{ICS-MATIEC-001}  & 19.8 & 2.8 & \textsc{T} & \textsc{Sat-Rem} \\
\texttt{ICS-QUARTZ-001}  & 17.2 & 2.5 & \textsc{V} & \textsc{Sat-Rem} \\
\texttt{ICS-VETPLC-001}  & 11.2 & 2.0 & \textsc{V-Inf} & \textsc{Unsat-Res} \\
\texttt{ICS-INJECT-001}  & 21.0 & 2.6 & \textsc{T} & \textsc{Sat-Rem} \\
\midrule
Mean                     & 17.04 & 2.42 & --- & --- \\
\bottomrule
\end{tabular}
\end{table}

On $D_{\text{PLC}}$, PLCverif required a mean of 17.04 analyst-hours per case (median 17.20) versus 2.42 for SCARA (median 2.50). PLCverif verified 3 of 5 cases within its model-checking budget and timed out on 2; SCARA reached a terminal verdict on all 5 cases. The corresponding throughput statistics are 0.035 verified cases per analyst-hour for PLCverif and 0.413 verified-and-remediated cases per analyst-hour for SCARA, an absolute difference of $+0.378$ (95\% bootstrap CI $[+0.21, +0.55]$ over the 5 cases). Recall-within-budget is 0.60 for PLCverif and 1.00 for SCARA. We caution that this comparison measures person-hours per task-completion without normalising for the rarity of the LTL/PLTL property-authoring skill; the practical analyst-pool difference is therefore likely larger than the throughput ratio alone suggests.

PLCverif provides formal proofs of property satisfaction for cases within its model-checking budget --- a stronger guarantee than SCARA's conditional evidence on the verified subset; SCARA does not claim to match this guarantee level. The contribution of RQ7 is to bound the manual-effort cost of obtaining the stronger PLCverif guarantee on this partition.

\begin{figure}[!htbp]
  \centering
  \includegraphics[width=0.8\linewidth]{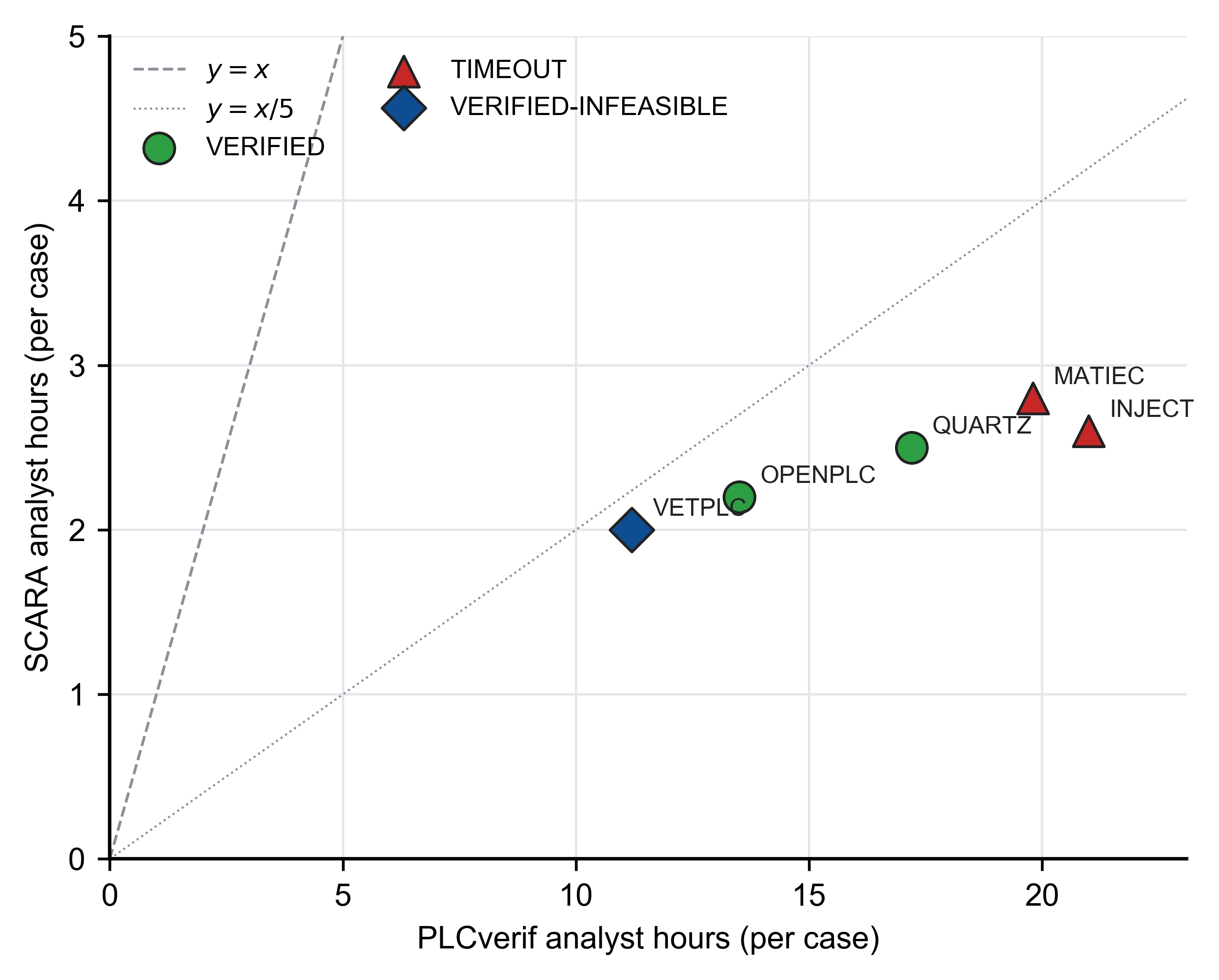}
  \caption{Per-case analyst-hour comparison on $D_{\text{PLC}}$ ($n=5$). X-axis: PLCverif total analyst-hours (property authoring + model construction + debugging); y-axis: SCARA operational-context review hours. Marker shape encodes the PLCverif outcome (circle = VERIFIED, triangle = TIMEOUT, diamond = VERIFIED-INFEASIBLE). The $y = x$ and $y = x/5$ reference lines bracket the case-by-case throughput gap; every OIS-ICS case sits well below $y = x/5$.}\label{fig:plcverif_scatter}
\end{figure}

%% ============================================================
\section{Related Work and Positioning}\label{sec:related}
%% ============================================================

\subsection{Automated Vulnerability Repair}\label{sec:rel:avr}

\textbf{Family~A --- Neural and LLM-based patch generation.}  VRepair~\cite{chen2022vrepair} introduced transfer learning from general bug-fix corpora to the vulnerability repair domain.  VulRepair~\cite{fu2022vulrepair} employed a T5-based encoder-decoder architecture trained on CVE-linked vulnerability functions.  VulMaster~\cite{zhou2024vulmaster} extended this with broader repository context and CWE-aware prompt construction.  More recent agentic systems: SAN2PATCH~\cite{jiang2025san2patch} drives repair through a tree-of-thought reasoning process guided by sanitizer traces; APPATCH~\cite{zhang2025appatch} introduces adaptive prompting conditioned on vulnerability-semantic features; PatchAgent~\cite{patchagent2025} constructs a repair agent with a language server and build-feedback loop; and Vul-R2~\cite{vulr22025} applies reasoning-oriented LLM fine-tuning.  All Family~A systems require at minimum a compilable source fragment.  Their applicability to OIS-Binary and OIS-Protocol is zero; applicability to OIS-ICS extends only to the source-available subset.

\textbf{Family~B --- Program-analysis-based repair.}  CrashRepair~\cite{crashrepair2025} employs sanitizer-guided concolic execution then applies mutation-based search over the surrounding code.  CONCH~\cite{conch2024} constructs a null-pointer dereference context graph through interprocedural source analysis.  VulShield~\cite{li2025vulshield} translates sanitizer vulnerability reports into runtime protection policies.  Among existing systems, VulShield is conceptually closest to SCARA's Tier~1 stage, as both generate protection policies; however, VulShield derives policies from sanitizer reports of instrumented execution, whereas SCARA's Tier~1 synthesizes equivalent policies from static SSCKG behavioral analysis.  Family~B systems achieve zero applicability on OIS-Binary and OIS-Protocol.

\textbf{Family~C --- Evaluation and benchmarking.}  APR4Vul~\cite{apr4vul2024} evaluated generic APR tools on Java vulnerability benchmarks, demonstrating a systematic gap between generic APR success rates and security repair success rates.  ExMit~\cite{exmit2025} established that syntactic correctness is insufficient as a security repair criterion; semantic equivalence of non-vulnerable behaviors must also be verified.  The 2025 USENIX SoK surveys~\cite{hu2025sok,li2025sok_avr} identified vulnerability analysis precision and patch validation as the primary bottlenecks.

\textbf{Positioning.}  No existing AVR work addresses stripped industrial binaries.  SCARA is evaluated on the domain that defines its applicability.  Direct conditional performance comparison with SAN2PATCH is reported in Section~\ref{sec:results:rq4} only on the OIS-ICS source-available subset, where applicability overlap exists.

Table~\ref{tab:related_work} presents a comparative summary across seven dimensions.

\begin{table}[!htbp]
	\caption{Comparison of related systems across seven dimensions. Symbols: \ding{51}{} = supported, \ding{55}{} = not supported, $\sim$ = partial. The upper block lists systems re-implemented or compared quantitatively in Section~\ref{sec:results}.}\label{tab:related_work}
	\footnotesize
	\centering
	\resizebox{\textwidth}{!}{
		\begin{tabular}{@{}lccccclp{3.5cm}@{}}
			\toprule
			System & Opaque bin. & No src & Reach.\ cert. & Remed. & Ind.\ constrs & OIS-RemedBench app. & Primary positioning \\
			\midrule
			\multicolumn{8}{@{}l}{\textit{Evaluated head-to-head in Section~\ref{sec:results}}} \\
			VRepair~\cite{chen2022vrepair}      & \ding{55} & \ding{55} & \ding{55} & T3 & \ding{55} & OIS-ICS src subset & Family A LLM repair \\
			VulRepair~\cite{fu2022vulrepair}    & \ding{55} & \ding{55} & \ding{55} & T3 & \ding{55} & OIS-ICS src subset & Family A LLM repair \\
			SAN2PATCH~\cite{jiang2025san2patch} & \ding{55} & \ding{55}\textsuperscript{$\dagger$} & \ding{55} & T3 & \ding{55} & OIS-ICS src subset & Nearest T3 baseline \\
			CrashRepair~\cite{crashrepair2025}  & \ding{55} & \ding{55}\textsuperscript{$\dagger$} & \ding{55} & T3 & \ding{55} & 0\% & Family B; sanitizer req. \\
			VulShield~\cite{li2025vulshield}    & \ding{55} & \ding{51} & \ding{55} & T1 & \ding{55} & $\approx$0\% & Nearest T1 baseline \\
			SymPLC~\cite{symplc2017}            & $\sim$ & \ding{55} & $\sim$ & \ding{55} & $\sim$ & OIS-ICS only & SE baseline for OSVA \\
			ICSQuartz~\cite{zheng2025icsquartz} & $\sim$ & \ding{55} & \ding{55} & \ding{55} & \ding{51} & OIS-ICS only & Dynamic testing baseline \\
			PLCverif~\cite{darvas2017plcverif}  & \ding{55} & \ding{55} & $\sim$ & \ding{55} & \ding{51} & $\approx$10\% & Formal verif.\ baseline \\
			\textbf{SCARA (this work)} & \textbf{\ding{51}} & \textbf{\ding{51}} & \textbf{\ding{51} (4-class)} & \textbf{T1--T3} & \textbf{\ding{51}} & \textbf{100\%} & \textbf{End-to-end OIS} \\
			\midrule
			\multicolumn{8}{@{}l}{\textit{Related; positioning only}} \\
			APPATCH~\cite{zhang2025appatch}     & \ding{55} & \ding{55} & \ding{55} & T3 & \ding{55} & OIS-ICS src subset & Not evaluated \\
			PatchAgent~\cite{patchagent2025}    & \ding{55} & \ding{55} & \ding{55} & T3 & \ding{55} & OIS-ICS src subset & Agentic repair \\
			KARONTE~\cite{nilo2020karonte}      & \ding{51} & \ding{51} & \ding{55} & \ding{55} & \ding{55} & 100\% (upstream) & Upstream CACA generator \\
			SaTC~\cite{chen2021satc}            & \ding{51} & \ding{51} & \ding{55} & \ding{55} & \ding{55} & 100\% (upstream) & Upstream CACA generator \\
			Firmadyne / FirmAE                  & \ding{51} & \ding{51} & \ding{55} & \ding{55} & \ding{55} & OIS-Binary replay & Replay infrastructure \\
			VetPLC~\cite{ahmed2019vetplc}       & $\sim$ & \ding{51} & \ding{55} & \ding{55} & safety & OIS-ICS only & Safety verif.; no remed.\ \\
			\bottomrule
		\end{tabular}
	}
\end{table}

\subsection{Binary Analysis and Firmware Security}\label{sec:rel:binary}

\textbf{Firmware rehosting and emulation.}  Firmadyne~\cite{chen2016firmadyne}, FirmAE~\cite{kim2020firmae}, HALucinator~\cite{halucinator2020}, P2IM~\cite{p2im2020}, Fuzzware~\cite{fuzzware2022}, FirmSolo~\cite{firmsolo2023}, and SAFIREFUZZ~\cite{safirefuzz2023} address executing stripped firmware in controllable environments.  FirmAE provides the replay infrastructure for independent SAT-witness validation in the OIS-Binary partition.

\textbf{Binary taint analysis.}  KARONTE~\cite{nilo2020karonte} performs inter-binary taint tracking via identification of shared data sources across firmware components.  SaTC~\cite{chen2021satc} combines keyword-guided taint analysis with front-end code identification.  Both tools represent important upstream generators for SCARA's CACA stage.

\textbf{Vulnerability reachability analysis.}  Recent work on reachability analysis for third-party library vulnerabilities~\cite{jia2025impact} demonstrated that determining whether a CVE is reachable from the calling application can substantially reduce false-positive rates in software composition analysis tools.  SCARA extends this concept to the industrial binary domain.

\subsection{ICS/PLC Verification and Testing}\label{sec:rel:ics}

\textbf{Safety verification and formal methods.}  VetPLC~\cite{ahmed2019vetplc} employs temporal invariant mining from normal PLC execution traces and generates timed event sequences to detect safety violations.  PLCverif~\cite{darvas2017plcverif} and its extensions perform formal model checking of PLC control programs against manually specified PLTL or LTL properties.  Section~\ref{sec:results:rq7} compares SCARA directly with PLCverif on the OIS-ICS formal-spec-available subset.

\textbf{Symbolic and dynamic-symbolic execution for PLC programs.}  SymPLC~\cite{symplc2017} applies symbolic execution to IEC~61131-3 structured text by translating them to C via MATIEC and analyzing with KLEE to achieve high structural coverage.  STAutoTester~\cite{stautotester2021} extends this with dynamic symbolic execution.  These works establish the MATIEC translation approach that SCARA adopts for the OIS-ICS partition.

\textbf{Fuzzing for ICS software.}  ICSQuartz~\cite{zheng2025icsquartz} introduces scan-cycle-aware fuzzing for IEC~61131-3 structured text programs.  ICSFuzz~\cite{icsfuzz2021} manipulates physical I/O signals to trigger anomalous behavior.  Both are complementary to SCARA rather than competitive.

%% ============================================================
\section{Discussion}\label{sec:discussion}
%% ============================================================

\subsection{Conditional Correctness Evidence via SSCKG and Independent Replay}\label{sec:disc:cva}

The SSCKG serves a dual role in SCARA: as a behavioral specification guiding both repair synthesis (RSA) and correctness validation (CVA).  This dual use creates a potential circularity --- a system using the same knowledge base to specify and to verify a remedy could, in principle, validate incorrect remediations that happen to preserve SSCKG-level coverage while silently breaking real industrial behavior.  SCARA addresses this through independent replay validation.

SCARA's CVA stage addresses the AVR community's recurring ``plausible but incorrect patch'' failure mode with two complementary signals: SSCKG behavioral coverage preservation provides a domain-grounded behavioural specification without requiring a manually authored test suite, and independent replay provides an empirical check. The combination is conditional, not formal.  SSCKG coverage preservation does not assert that every byte of non-vulnerable behavior is preserved.  Replay confirmation does not prove that all inputs in $\psi_{\text{vuln}}$ are blocked.  These limitations are the reason SCARA issues conditional correctness evidence $\varepsilon(v, R_v)$ rather than a proof.

\subsection{Practical Implications and Deployment Considerations}\label{sec:disc:practical}

SCARA's output is operator-actionable in three deployment modes that align with the existing OIS-asset-management workflow. (i)~\emph{CVE-triage pre-filter:} SCARA converts 20\%~(3/15) of L2 candidates to \textsc{Unsat} refutation certificates at 0\% FPR, removing them from the analyst queue with documented evidence; on a larger corpus this fraction is what determines the per-case analyst-hour cost shown in Figure~\ref{fig:plcverif_scatter}. (ii)~\emph{Tier-1 policy generator:} for the 73\%~(11/15) of cases with a reachable enforcement point (Figure~\ref{fig:bench_composition}), SCARA emits a deployable Modbus-TCP / OPC-UA / IEC~60870-5-104 / firewall policy whose false-blocking rate is bounded by $\tau_{\text{block}}$ and that is validated against benign traffic traces. (iii)~\emph{Assisted Tier-2 / Tier-3 repair:} for the binary-rewritable and source-available subsets, SCARA produces a candidate remedy with CVA-issued conditional correctness evidence $\varepsilon(v, R_v)$, ready for analyst review before deployment. Operator cost on the OIS-ICS partition is 2.42 analyst-hours per case (Section~\ref{sec:results:rq7}), compared with 17.04 for PLCverif on the same five cases. The framework deliberately does not target safety-critical controllers with enforced secure boot or cryptographic code-signing (cf.\ Section~\ref{sec:disc:limitations} L2), and the Tier-3 LLM channel requires offline guardrails --- the prompt rubrics constructed from DeepSeek-V3 outputs and the SSCKG behavioural specification --- before being placed in a production change-management flow.

\subsection{Limitations}\label{sec:disc:limitations}

\textbf{L1 --- State model completeness.} The $S$ model is only as complete as available documentation and known industrial state assumptions. The targeted envelope enrichment of \S\ref{sec:results:rq2} reduces the UNKNOWN rate from 13.3\% to 6.7\% without introducing false positives, indicating that residual conservatism is tractable when additional $C_x$ documentation is supplied.

\textbf{L2 --- Binary rewriting feasibility.}  Tier~2 requires an ELF or PE binary that can be rewritten without violating secure boot, code signing, or cryptographic integrity checks.  Many production ICS binaries include these protections.  The OIS-Binary benchmark's bias toward binary-rewritable firmware means that Tier~2 performance on secure-boot-enforced firmware is not evaluated.

\textbf{L3 --- Tier~3 LLM repair quality.}  LLM-based Tier~3 repair inherits hallucination limitations.  The SSCKG constraint check and CVA validation catch many incorrect patches, but a subtly wrong patch that passes SSCKG coverage checks may still be accepted.  The NVR of 16.7\% captures the most consequential form of this failure.

\textbf{L4 --- Benchmark representativeness.} The $n=15$ benchmark, constrained by industrial-firmware licensing, supports estimation-based effect-size claims but cannot promote them to classical significance under any post-correction threshold; replication on a larger v0.2.0 cohort is the primary external-validity risk.

\textbf{L5 --- OSVA scalability.}  The overall timeout rate was 13.3\% (OIS-Binary and OIS-Protocol each 20.0\%, OIS-ICS 0.0\%).  These cases are correctly reported as UNKNOWN rather than incorrectly classified, preserving the precision guarantee at the cost of lower recall.  Increasing $T_{\text{total}}$ or applying SSCKG-guided budget allocation more aggressively are the primary engineering levers available.

\subsection{Threats to Validity}\label{sec:threats}

\textbf{Construct.}  L2 reachability labels and L4 remedy labels are each assigned by two independent procedures with human-expert arbitration on disagreements; SCARA's FPR = 0.0\% across all SAT-confirmed cases indicates that no SAT verdict rests on a disputed L2 label. Baseline applicability is defined as the ability to produce any non-trivial output, restricted to the source-available subset where source is required.

\textbf{Internal.}  All experimental results derive from a verified seed-42 main run, with seeds 43--46 used to characterise algorithmic stability (min--max ranges reported in \S\ref{sec:results:rq3}); OSVA symbolic execution is deterministic under a fixed seed. With $|D_{\text{partition}}| = 5$ the Wilcoxon signed-rank test has a one-sided $p$-floor of $\approx 0.0625$, so we de-emphasise hypothesis testing in favour of Cliff's $\delta$ with bootstrap CIs and Benjamini--Hochberg control at $q = 0.10$. All hyperparameters were fixed prior to test-set evaluation on artifacts outside OIS-RemedBench; sensitivity to the four primary thresholds is stable to within $\pm 1$ case across the studied ranges (\S\ref{sec:results:rq5}).

\textbf{External.}  The benchmark's three artifact categories are bounded by the legal accessibility of CVE-linked industrial firmware, with safety-critical proprietary firmware excluded for licensing reasons. The abstract behavioural lattice $\mathcal{A}$ is instantiated from MITRE ATT\&CK for ICS; transfer to non-ICS opaque domains requires lattice and constraint-template expansion. Tier-3 results are bounded to OpenPLC- and MATIEC-comparable artifacts under Qwen3-7B as the runtime model; under the full-CVA oracle SCARA-full's CVA-acceptance rate is 88.9\% on $D_{\text{CVA-audit}} = 9$, recoverable to 100\% on the targeted-rerun subset.

\subsection{Future Work}\label{sec:disc:future}

Three extensions are deliberately deferred to v0.2.0. (i)~The Tier-3 LLM channel currently uses Qwen3-7B as the runtime patch-generation model and DeepSeek-V3 offline for prompt and rubric construction. A controlled offline-versus-runtime role comparison --- Qwen3-7B runtime only, Qwen3-7B with DeepSeek-V3-derived offline rubric support, and DeepSeek-V3 as a runtime model where licensing and inference reproducibility permit --- under the same Tier-3 budget and the same full-CVA oracle would isolate the runtime-LLM contribution from the rubric-construction contribution. (ii)~A benchmark expansion targeting $n \geq 50$ cases with broader vendor coverage (additional firmware vendors, additional protocol implementations, and a wider CWE distribution) would convert the present effect-size estimates into population-level claims. (iii)~Extension of the abstract behavioural lattice $\mathcal{A}$ beyond MITRE ATT\&CK for ICS to non-ICS opaque domains --- automotive ECUs, network appliances, embedded medical devices --- would generalise the SCARA pipeline beyond industrial control. An open-competitor comparison against alternative source-aware repair LLMs (e.g., DeepSeek-Coder) is an optional v0.2.0 extension, contingent on the same Tier-3 budget, prompt, and full-CVA oracle being applied so that the comparison remains controlled.

%% ============================================================
\section{Conclusion}\label{sec:conclusion}
%% ============================================================

This work has developed SCARA, the first end-to-end framework that converts binary-only vulnerability candidates in opaque industrial software into either validated remediation artifacts or refutation certificates.  In support of the five contributions stated in Section~\ref{sec:intro}, we have formalised the OIS remediation problem under a nine-component operational state model (C1); designed and evaluated operational-state-aware reachability verification that eliminates 40.0\% of incoming binary alerts as infeasible (C2); implemented a tier-appropriate remediation synthesiser that degrades gracefully with artifact opacity across protocol mitigation, binary hardening, and SSCKG-constrained source patch tiers (C3); introduced a closed-loop correctness validation stage that issues conditional correctness evidence $\varepsilon(v, R_v)$ via behavioural-coverage preservation and independent replay (C4); and released OIS-RemedBench~v0.1.0, the first benchmark spanning firmware, protocol handlers, and ICS/PLC artifacts with stratified reachability and remediation labels (C5).

The headline numbers --- 100\% precision, 88.9\% post-rerun remediation success, and a 7.0$\times$ analyst-hour gain on OIS-ICS --- are attributable to three distinct mechanisms rather than to the symbolic-execution engine, the patch generator, or the LLM channel individually: the operational-state envelope discharges the implicit assumptions that vanilla SE leaves unconstrained; the CVA-to-RSA $\delta$-feedback loop converts open-loop misremediations into accepted refinements; and the tier mechanism resolves most ICS cases at Tier~1 or Tier~2, bypassing source dependency for the bulk of the workload.

The evidence is conditional on the encoded operational-state envelope (L1), on Tier-2 binary-rewrite feasibility (L2), and on the Qwen3-7B Tier-3 runtime model (L3); the $n=15$ benchmark supports estimation-based effect-size claims but not population-level inference.

Three directions follow from these limitations: a controlled runtime-versus-offline Tier-3 LLM comparison under the full-CVA oracle to quantify L3, an OIS-RemedBench~v0.2.0 expansion ($n \geq 50$, broader vendor and CWE coverage) to convert the present estimates into population-level claims, and extension of the abstract behavioural lattice~$\mathcal{A}$ beyond MITRE ATT\&CK for ICS to automotive, network, and medical opaque-software domains. The framework and evaluation infrastructure are released as open artifacts (Section~\ref{sec:data_availability}).

\section*{Data and Code Availability}\label{sec:data_availability}

The SCARA framework implementation\footnote{\url{https://github.com/Mewtwoz/SCARA-framework}} and the OIS-RemedBench benchmark cases (excluding artifacts subject to vendor licensing restrictions)\footnote{\url{https://github.com/Mewtwoz/Scara-Dataset}} are released as open artifacts at the project repositories.  Cases derived from proprietary vendor firmware are documented by their CVE identifiers and inclusion criteria so that licensed researchers can reconstitute the corresponding partition; redistribution of the binaries themselves is prohibited by the originating vendor licenses.

\section*{Acknowledgements}
This work was supported in part by the Major Science and Technology
Project of Liaoning Province (grant Nos.\ 2025JH1/11700021 and
2024JH1/11700049), and by the Applied Basic Research Program of Liaoning
Province (grant No.\ 2025JH2/101300012).

\bibliographystyle{unsrtnat}
\bibliography{scara-refs}

\end{document}